\documentclass[11pt,a4paper]{article}

\usepackage{indentfirst}
\usepackage{amsfonts}
\usepackage{amsmath}
\usepackage{amsthm}
\usepackage{amssymb}
\usepackage{lineno}
\usepackage{algorithm} 
\usepackage{verbatim}
\usepackage[makeroom]{cancel}
\usepackage{rotating}
\usepackage{float}
\usepackage{cite}
\usepackage{mathtools}
\usepackage{amssymb}
\usepackage{mathrsfs}
\usepackage{epsfig}
\usepackage{graphicx}
\usepackage{indentfirst}
\usepackage{framed}
\usepackage{xcolor}
\usepackage[backref=page]{hyperref}
\usepackage[capitalize]{cleveref}
\usepackage{enumitem}
\usepackage{bm}
\usepackage[normalem]{ulem}

\usepackage{revsymb}
\usepackage{dsfont}

\usepackage{upgreek}

\usepackage[font=footnotesize,labelfont=bf]{caption}
\usepackage{soul}
\usepackage{todonotes}

\newcommand{\ket}[1]{|#1\rangle}
\newcommand{\bra}[1]{\langle#1|}

\newcommand{\ketbra}[2]{|#1\rangle\! \langle #2|}

\newcommand{\Tr}{\operatorname{Tr}}
\newcommand{\tr}[1]{\Tr\left(#1\right)}

\newcommand{\eps}{\varepsilon}

\def\01{\{0,1\}}

\newcommand{\cA}{\ensuremath{\mathcal{A}}}
\newcommand{\cB}{\ensuremath{\mathcal{B}}}

\newcommand{\A}{\mathcal A}

\newcommand{\cE}{\mathcal{E}}

\newcommand{\B}{\mathcal{B}}
\renewcommand{\H}{\mathcal{H}}

\newcommand{\id}{\mathrm{id}}
\newcommand{\supp}{\operatorname{supp}}
\newcommand{\W}{\mathcal{W}}

\newcommand{\V}{\mathcal{V}}
\newcommand{\CP}{\mathrm{CP}}
\newcommand{\CB}{\mathrm{CB}}

\newcommand{\hide}[2][]{{\color{blue} #1}}
\newcommand{\RR}{\mathbb{R}}
\newcommand{\KK}{\mathbb{K}}

\DeclarePairedDelimiterX{\norm}[1]{\lVert}{\rVert}{#1}

\newcommand{\emp}[1]{\textbf{#1}}

\usepackage[lmargin=0.7in,rmargin=0.7in,tmargin=0.65in,bmargin=0.75in,footskip=0.25in]{geometry}

\definecolor{corlinks}{RGB}{200,0,0}
\definecolor{cormenu}{RGB}{200,0,0}
\definecolor{corurl}{RGB}{200,0,0}

\hypersetup{
colorlinks=true,
urlcolor=corlinks,
linkcolor=corlinks,
menucolor=cormenu,
citecolor=corlinks,
pdfborder= 0 0 0
}

\newtheorem{theorem}{Theorem}
\numberwithin{theorem}{section}

\newtheorem{definition}[theorem]{Definition}
\newtheorem{lemma}[theorem]{Lemma}
\newtheorem{corollary}[theorem]{Corollary}
\newtheorem{proposition}[theorem]{Proposition}

\newtheorem*{proposition*}{Proposition}

\theoremstyle{definition}
\newtheorem{remark}[theorem]{Remark}

\def\01{\{0,1\}}

\DeclareDocumentCommand{\dist}{o}{%
  \idfNoValueTF{#1}{d}{d_{\mathrm{#1}}}%
}

\NewDocumentCommand{\Prob}{e{_} m}{%
  \idfNoValueTF{#1}{%
    \Pr \set*{#2}
  }{%
    \Pr_{#1} \set*{#2}
  }
}

\newcommand{\C}{\ensuremath{\mathcal{C}}}

\newcommand{\E}{\ensuremath{\mathbb{E}}}

\newcommand{\F}{\ensuremath{\mathcal{F}}}

\newcommand{\K}{\ensuremath{\mathcal{K}}}

\newcommand{\1}{\mathds{1}}

\newcommand{\beq}{\begin{equation}}
\newcommand{\beql}[1]{\begin{equation}\label{#1}}
\newcommand{\eeq}{\end{equation}}
\newcommand{\eeqp}{\,\,\,.\end{equation}}
\newcommand{\eeqc}{\,\,\,,\end{equation}}

\usepackage{algorithm}
\usepackage{algpseudocode}

\DeclareMathOperator{\poly}{poly}

\usepackage{tikz,tikz-qtree}
\usepackage{wrapfig}
\usetikzlibrary{arrows}

\usepackage{graphicx}

\makeatletter
\def\and{%
  \end{tabular}%
  \hskip 0.5em \@plus.17fil\relax
  \begin{tabular}[t]{c}}
\makeatother

\begin{document}

\title{\vspace{-20pt}\huge \bf Learning finitely correlated states:\protect\\ stability of the spectral reconstruction}

\author{Marco Fanizza\thanks{Department of Mathematical Sciences, University of Copenhagen, Universitetsparken 5, 2100 Denmark. \texttt{mf@math.ku.dk}. Previously at F\'{\i}sica Te\`{o}rica: Informaci\'{o} i Fen\`{o}mens Qu\`{a}ntics, Departament de F\'{i}sica, Universitat Aut\`{o}noma de Barcelona, ES-08193 Bellaterra (Barcelona), Spain.}
\and Niklas Galke \thanks{{F\'{\i}sica Te\`{o}rica: Informaci\'{o} i Fen\`{o}mens Qu\`{a}ntics, Departament de F\'{i}sica, Universitat Aut\`{o}noma de Barcelona, ES-08193 Bellaterra (Barcelona), Spain}. \texttt{niklas.galke@uab.cat}} 
\and Josep Lumbreras\thanks{{Centre for Quantum  Technologies,  National University of Singapore, Singapore}. \texttt{josep.lumbreras@u.nus.edu}}
\and Cambyse Rouzé\thanks{{Inria, Télécom Paris - LTCI, Institut Polytechnique de Paris, 91120 Palaiseau, France.} \texttt{rouzecambyse@gmail.com}}
\and Andreas Winter\thanks{{ICREA \&{} F\'{\i}sica Te\`{o}rica: Informaci\'{o} i Fen\`{o}mens Qu\`{a}ntics, Departament de F\'{i}sica, Universitat Aut\`{o}noma de Barcelona, ES-08193 Bellaterra (Barcelona), Spain. AW is Hans Fischer Senior Fellow with the Institute for Advanced Study, Technische Universit\"at M\"unchen, Lichtenbergstra{\ss}e 2a, D-85748 Garching, Germany}. \texttt{andreas.winter@uab.cat}}}

\date{\today} 

\maketitle
\vspace{-8mm}

\thispagestyle{empty}

\abstract{
Matrix product operators allow efficient descriptions (or realizations) of states on a 1D lattice. We consider the task of learning a realization of minimal dimension from copies of an unknown state, such that the resulting operator is close to the density matrix in trace norm. For finitely correlated translation-invariant states on an infinite chain, a realization of minimal dimension can be exactly reconstructed via linear algebra operations from the marginals of a size depending on the representation dimension. We establish a bound on the trace norm error for an algorithm that estimates a candidate realization from estimates of these marginals and outputs a matrix product operator, estimating the state of a chain of arbitrary length $t$. This bound allows us to establish an $O(t^2)$ upper bound on the sample complexity of the learning task, with an explicit dependence on the site dimension, realization dimension and spectral properties of a certain map constructed from the state. 
A refined error bound can be proven for $C^*$-finitely correlated states, which have an operational interpretation in terms of sequential quantum channels applied to the memory system. We can also obtain an analogous error bound for a class of matrix product density operators on a finite chain reconstructible by local marginals. In this case, a linear number of marginals must be estimated, obtaining a sample complexity of $\tilde{O}(t^3)$. The learning algorithm also works for states that are sufficiently close to a finitely correlated state, with the potential of providing competitive algorithms for other interesting families of states.

}

\maketitle


\section{Introduction}

Quantum state tomography is a fundamental task that allows the accurate characterization and validation of quantum devices. When provided with a source, such as a quantum circuit or experiment, that generates an unknown quantum system comprising $n$ qudits (or spins), the goal of quantum state tomography is to represent the density matrix of the quantum state through measurements and classical post-processing. It is well established that, in the absence of prior knowledge about the state, attempting to produce an estimate that closely resembles the unknown state requires resources—such as the number of copies, memory or computational time—that increase exponentially with the number of qudits $n$~\cite{ODonnell2016, Haah2017}, even when a fully functional quantum computer capable of executing any measurement allowed by quantum mechanics is available. In practice, however, a more manageable situation often arises, as the quantum states of interest typically exhibit a well-defined, physically motivated structure, requiring fewer parameters for full characterization. In these scenarios, one can aim to reconstruct efficient descriptions, with examples of such states including outputs of quantum circuits with bounded gate complexity~\cite{zhao2023learning} or thermal states~\cite{bakshi2023learning}. Alternatively, recent approaches like shadow tomography~\cite{aaronson2018shadow,buadescu2021improved} or classical shadows~\cite{huang2020predicting} have emerged, focusing on producing efficient descriptions that yield the same predictions as the density matrix on a subset of observables. 
In this work, we focus on learning \emph{tensor network} realizations for states on a one-dimensional lattice, a.k.a Matrix Product Density Operators (MPDOs), with the promise that an efficient realization exists for the otherwise unknown state. When there are no constraints on positivity, these tensor network representations are called simply matrix product operators (MPOs) or tensor-trains (TTs). Translation-invariant states on an infinite chain whose marginals are MPDOs of bounded number of parameters are known as finitely correlated states (FCSs). {Tensor network realizations} are one of the most widely studied \emph{ansatz} classes for many-body systems, even beyond the one-dimensional setting we consider. A substantial body of literature has demonstrated the effectiveness of tensor network states for computing or approximating ground and thermal properties of quantum systems. This effectiveness is supported both by theoretical results \cite{CiracMPDOGibbs,hastings2006solving,Alhambra2021} and in practice~\cite{Banuls2023}.

Our work includes the following:

\begin{itemize}
\item A quantum state tomography algorithm for finitely correlated states that can utilize both local and entangled measurements to efficiently reconstruct marginals of the original state using a Matrix Product Density Operator (MPDO) realization (see~\cref{Figintro}). Very similar algorithms appeared previously, e.g. in~\cite{Baumgratz_2013}.
\item A rigorous sample complexity bound on the trace norm error for these MPDO representations, which scales polynomially with the system size. Our results include both pure and mixed states, providing the first sample complexity guarantee for the tomography of Matrix Product Density Operators from local measurements~\cite{anshu2023survey}.
\item An extension of our approach to handle the non-translation invariant case, which is particularly relevant for learning physically motivated states in finite-size systems. The learning algorithm is also robust, meaning that it can also be applied to states that are sufficiently close to having an efficient MPDO realization. 
\end{itemize}

\begin{figure}
    \centering
\includegraphics[scale=0.4]{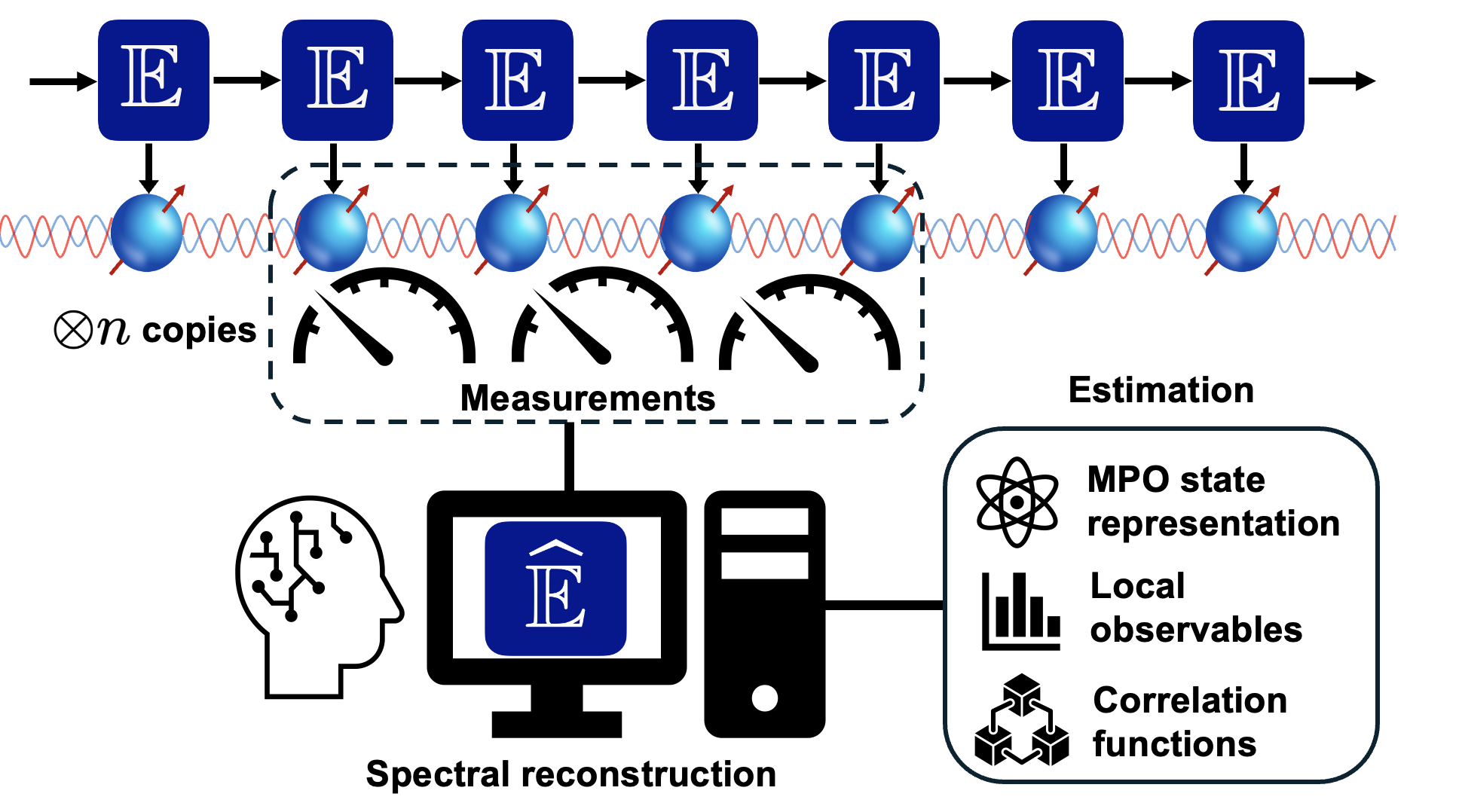}
    \caption{{Sketch for learning a matrix product operator representation of a finitely correlated state. In the measurement phase we aquire data from $n$ copies of a finite size marginal of the state through local or global measurements. Then the spectral reconstruction phase is a classical post-processing technique that allows us to obtain approximations of the parameters of a linear model that generates the state. Using this approximations we then can generate an accurate estimation of an MPO representation. This allows us to obtain estimations of expected values of local observables or correlations functions on larger marginals of the state.  }}\label{fig:cartoon_intro}
\end{figure}

FCSs were originally introduced to characterize translation-invariant states on quantum spin chains with finite correlations~\cite{Fannes1992}. The initial motivation for these studies was to provide efficient representations of ground states for finite-range interactions, 
such as the AKLT Hamiltonian~\cite{Fannes1992}.  
In fact, our approach rests on the crucial observation, elucidated in~\cite{Fannes1992}, that finitely correlated states (FCSs) can be viewed as a quantum generalization of stochastic processes which admit finite-dimensional linear models, known as quasi-realizations~\cite{Vidyasagar2014}. At a high level, a quasi-realization of a stochastic process is a collection of matrices and vectors that can reproduce the probabilities of outcomes through simple matrix multiplications and traces. In the classical case, a special subclass of these models is known as hidden Markov models~\cite{Rabiner1989, Vidyasagar2014} (also referred to as positive realizations). A hidden Markov model is a statistical model where the system being modeled follows a Markov process with unobservable (hidden) states that evolve stochastically, and observable outputs are generated based on measurements of these hidden states. A natural quantum generalization (quantum hidden Markov models) are states generated through consecutive applications of a physical map, which is a quantum channel, on a hidden quantum memory system. We call these models $C^*$-realizations, following the nomenclature of~\cite{Fannes1992}, where $C^*$-finitely correlated states are states admitting such descriptions. 
This interpretation also motivates the physical interest in learning MPDOs as they include entangled multipartite states that can be generated with a number of local physical operations that scales linearly with the size of the system. 
Note, however, that the class of finitely correlated states is not completely exhausted by states with finite-dimensional 
$C^*$-models, as a recent result by some of the present authors shows~\cite{fanizza2023quantum} in the case of the infinite chain, while similar statements had been obtained for non-translation invariant states on finite chains~\cite{Cuevas2013} and for translation-invariant states on periodic chains~\cite{Cuevas2016}. In fact, the most general class of FCSs can be interpreted as being generated via a memory system described by a general probabilistic theory (GPT), which is a mathematical characterization of plausible physical theories that admit state preparation and measurements. The physical interest of states that do not admit finite-dimensional quantum models is an important and unexplored question. In this paper, dealing also with this more general class of states is a natural choice, given the mathematical tools we use.

In the classical literature on stochastic processes, learning guarantees for hidden Markov models have been established for the so-called \textit{spectral algorithms}: a series of works starting with~\cite{Hsu2008} has shown that an estimate of the relevant marginal at precision $\epsilon^{-1}$ gives rise to a reconstruction of the process with an error in total variation distance that can be bounded as $O(\poly(t,\epsilon^{-1},m))$, where $t$ is the timespan considered, and $m$ is the size of the hidden memory. This, in turn, implies a rigorous sample complexity bound that scales quadratically with the system memory size.
In the quantum setting, reconstruction algorithms using a similar idea have been proposed, notably in the non-translation invariant setting~\cite{PhysRevLett.111.020401}, and their scalability have been demonstrated with simulations and experiments. In general these approaches involve learning a small marginal of the state using quantum state tomography, reconstructing an approximation of the hidden system, and then using that reconstruction to output a larger marginal of the state. Maximum likelihood methods have also been considered~\cite{Baumgratz_2013}. Alternative methods have been devised for Matrix Product States (MPSs), for which certifiable tomography protocols exist~\cite{cramer2010} and have been tested in experiments~\cite{lanyon2017}. Moreover, when the reconstruction can be achieved from the knowledge of local marginals of sufficiently small size, local measurements suffice, and the quantum complexity of the algorithm is much more manageable than in the general case.

 Our algorithm and proof strategy build on the classical case~\cite{Hsu2008}, generalizing the original argument to handle the presence of entangled observables and memory systems beyond classical (and even quantum) systems. At a high level, our algorithm employs what we call the \textit{spectral reconstruction} technique, in analogy to the classical case. Using this method, given as input the density matrix of marginals of large enough size (say $t^*$), 
 using linear algebra (matrix inversion and multiplications), one can reconstruct a valid MPDO representation. 
In our learning algorithm, we assume that an estimate of the density matrix of the marginal of size $t^*$ is provided; errors will naturally occur due to the finite statistics of the experiments or any kind of physical noise; using a modification of the spectral reconstruction, that can better tolerate errors in the estimate, we can then output estimates of model parameters. These can be, in turn, used to construct estimates of marginals at any size $t$, possibly larger than $t^*$. 
While classically this technique was formulated to handle translation-invariant stochastic processes, we also generalize it to address non-translation invariant case. 
At the technical level, our analysis fully exploits the fact that an FCS naturally defines an operator system~\cite{paulsen_completely_2003} on which one can speak of completely positive maps and completely bounded norms, and that the map generating the FCS is contractive. This allows us to address not only states generated by repeated quantum measurements but also models for which no finite-dimensional quantum memory explanation exists. Notably, this even provides error bounds for learning classical states that were not previously available.

\subsection{Background and related work}
As already mentioned, our approach draws inspiration from the literature on hidden Markov models, specifically focusing on what are commonly referred to as ``spectral algorithms''. These algorithms are designed to reconstruct a linear model that explains a stochastic process through the estimation of marginals and matrix algebra operations. They aim to overcome the limitations of maximum likelihood approaches, which often lack rigorous convergence guarantees. Spectral algorithms have been applied in more general graphical models as well~\cite{Anandkumar2012}. 
One challenge in ensuring the accuracy of these algorithms lies in their reliance on matrix inversion, which makes them sensitive to small singular values. In a prior study~\cite{Hsu2008} (see also~\cite{MosselRoch} for a related approach), it was shown that the spectral algorithm provides an error bound in terms of total variation distance, assuming that the process is described by a hidden Markov model satisfying certain assumptions, including that certain matrices constructed from the hidden Markov model parameters have singular values lower bounded by a known constant. This established a learning guarantee akin to PAC (Probably Approximately Correct) learning, which becomes more tractable by narrowing the class of target processes. Subsequent works, among them~\cite{Siddiqi2009} and~\cite{balle2013learning}, relaxed some of the assumptions made in~\cite{Hsu2008}, such as the requirement of reconstructing from marginals of size three and invertibility assumptions on the so-called observation matrix. 
Nevertheless, a key premise in these works is indeed the existence of an approximating hidden Markov model. To the best of our knowledge, no error bound in total variation distance has been derived for general models (as discussed in~\cite{balle2013learning} and~\cite{Balle2014}). Notably, an error bound for learning stochastic processes generated by hidden quantum Markov models with the spectral algorithm is also missing, which have been proven to be more expressive than classical ones~\cite{Monras2010,monras2016}, although they are not as expressive as the most general class of processes one could consider~\cite{fanizza2023quantum}. As a result, the idea of investigating spectral algorithms in the context of hidden quantum Markov models was suggested in~\cite{pmlr-v130-adhikary21a}. We note that for a specific type of quantum hidden Markov models it was shown that it is possible to learn efficiently~\cite{juba2012} minimizing the relative entropy between the empirical distribution and distributions in an $\epsilon$-net of the quantum hidden Markov models of fixed memory system. While it is possible that a similar argument could work for a more general class or even for genuinely quantum states, the relative entropy (or any other distinguishability measure) minimization is not computationally efficient, while the spectral algorithms are.

In the quantum information literature, a reconstruction method related to the spectral algorithm we propose is the direct tomography algorithm in~\cite{cramer2010}, which learns a circuit preparing an MPS by knowledge of marginals of sufficiently large size and uses collective operations on subchains of length scaling with the logarithm of the bond dimension. An error guarantee was provided for this algorithm, but it is not clear how to generalize this algorithm for the infinite chain or mixed states. For MPDOs, the state reconstruction scheme in~\cite{PhysRevLett.111.020401} from the exact knowledge of the marginals, is analogous to the one we present, but is lacking an error bound on the precision of the reconstruction. Moreover, our algorithm uses a slightly different prescription for the reconstruction from empirical data. Note that these algorithms were directly proposed in the non-translation-invariant (non-homogeneous) setting, in contrast to the classical ones. In terms of our error analysis, there is no substantial difference between these two cases except for the fact that, in the non-homogeneous case, one is required to learn several marginals of the state instead of a single one. The robustness of the matrix reconstruction in the bipartite case was already discussed in~\cite{PhysRevLett.111.020401} and investigated in more depth in~\cite{Holzapfel18}, where the authors obtained a bound on the accuracy in terms of the operator norm. However, an error bound in trace distance was not explicitly stated, and it is not clear how to adapt their analysis to get a bound for a chain as opposed to a bipartite state. To perform a similar step in our setting, we instead adapt the analysis carried by~\cite{Hsu2008, Siddiqi2009, balle2013learning} to the quantum case. 
Learning the marginals themselves can be done in several ways, and we take the necessary size of the marginals as a parameter of our class. In the translation invariant case, this size is at most of the order of the minimal dimension of the memory system (bond dimension), while in the non-homogenous case it can be arbitrarily large. That being said, in addition to the use of strictly local measurements to learn the marginals, one can use other methods such as constrained maximum likelihood learning algorithms based on measurement statistics. A very recent result \cite{qin2023stable} in that direction provides a bound in Hilbert-Schmidt norm for the error in the reconstruction of the state which is polynomial in the size and requires global (in fact Haar random), yet independent, single-copy measurements. It is unclear to us whether this bound directly leads to a good error bound in trace distance. Instead, we use it as a means to learn the marginals with an error in Hilbert-Schmidt distance, which serves our needs adequately as we will show in Section \ref{sec.statereconstruction}.

\begin{figure}
    \centering
\includegraphics[scale=0.4]{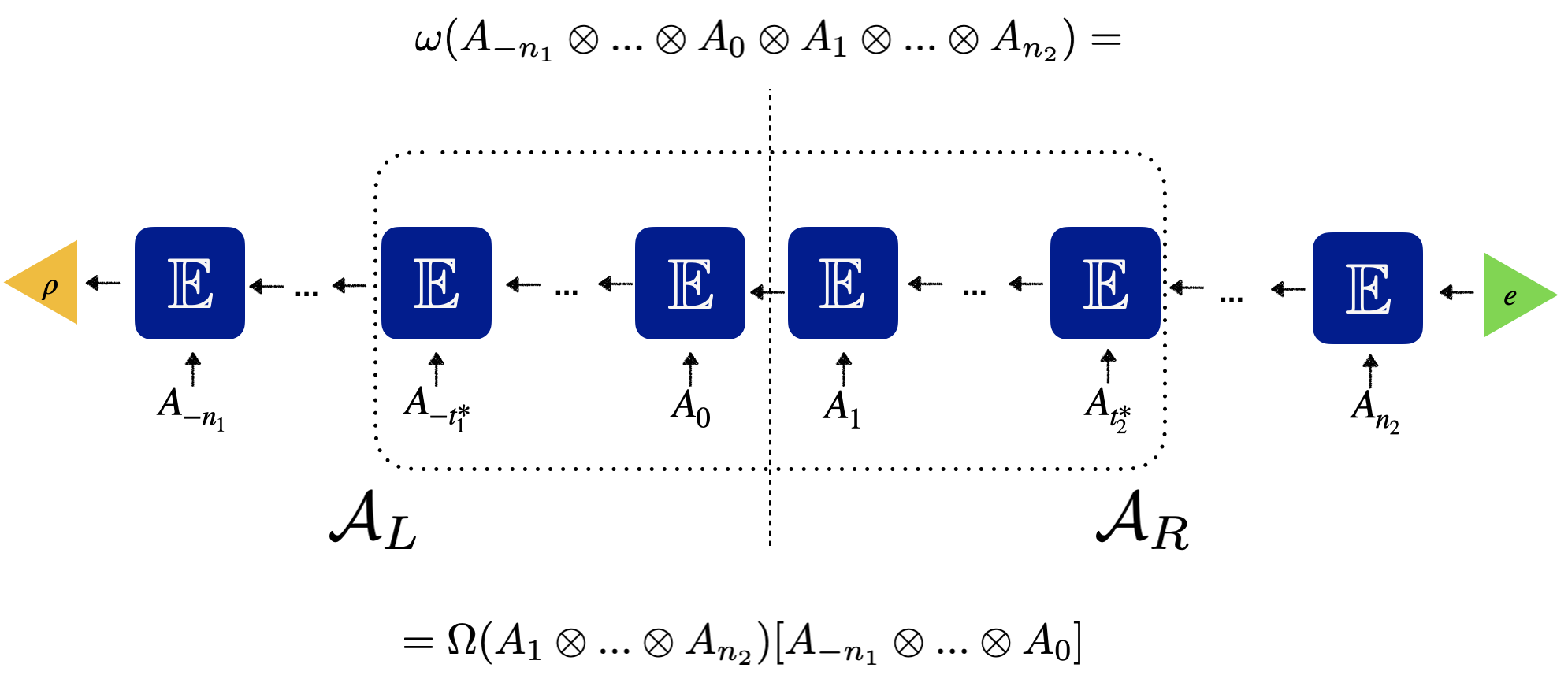}
    \caption{{Schematic of a finitely correlated state. The marginal state corresponding the dotted rectangle completely determines the state. An approximate estimate of the marginal can be used to reconstruct an estimate of the state. We bound the error in trace distance of the reconstruction up to a desired size.}}\label{Figintro}
\end{figure}

After the first version of the present work appeared,~\cite{gondolf2024conditional} showed an efficient algorithm to learn an MPO approximation of a Gibbs state, with methods different from ours and which do not apply to general FCSs.

\section{Results}
\subsection{Models for finitely correlated states}

 Let a spin system be described by a $d_{\mathcal{A}}$-dimensional Hilbert space $\mathcal{H}_{\mathcal{A}}$, and let us denote $p\times q$ complex matrices as $\mathbb{M}_{p,q}(\mathbb{C})$. A state on a finite chain with Hilbert space $\mathcal{H}_n:=\mathcal{H}_{\mathcal{A}}^{\otimes n}$ is described by a density matrix $\omega\in\mathbb{M}_{d_{\mathcal{A}}^n,d_{\mathcal{A}}^n}(\mathbb{C})$, $\omega\geq 0$, $\Tr[\omega]=1$. An MPDO representation (with open boundary condition)  of $\omega$ is specified by 
 \begin{itemize}
 \item a list of numbers $m_{i}\in \mathbb{N}$, $i\in [n+1]$, 
 \item a list of matrices $\mathbb{E}^{(i)}_{k_i,l_i}\in \mathbb{M}_{m_i,m_{i+1}}(\mathbb{C})$, $i\in[n],\, k_i,l_i\in[d_{\A}]$, \item a row vector $\rho\in \mathbb{C}^{m_1}$,
 \item and a column vector $e\in \mathbb{C}^{m_n}$ 
 \end{itemize}
 such that
 \begin{equation}
 \bra{k_1,...k_t}\omega\ket{l_1,...l_t}=\rho\mathbb{E}_{k_1,l_1}^{(1)}...\mathbb{E}_{k_t,l_t}^{(n)}e. 
 \end{equation}

Alternatively, the state of the system can be represented as a positive functional on the Hilbert space of $d_{\mathcal{A}}^n\times d_{\mathcal{A}}^n$ matrices. The functional $\tilde{\omega}$ corresponding to the density matrix $\omega$ is $\tilde{\omega}(A):=\Tr[\omega A]$. The normalization condition can be expressed as the unitality of the functional, i.e. $\tilde{\omega}(\1)=\Tr[\omega]=1$.  On the infinite chain, each spin system is associated with an index $i\in\mathbb{Z}$. While the the density matrix picture is not applicable, it is still possible to define states as normalized positive functionals on an appropriate Hilbert space describing observables that look like identity at infinity, using the formalism of $C^*$-algebras. Such states are completely determined by the sequence of their finite-size marginals, i.e. the restriction to finite subchains $[t_1,t_2]$, $t_1,t_2\in\mathbb{Z}$. 
These marginals can be of course described by density matrices $\omega_{[t_1,t_2]}$. A state is translation invariant if $\omega_{[t_1,t_2]}$=$\omega_{[t_1+t,t_2+t]}$ for every $t_1,t_2,t\in\mathbb{Z}$. Thus, we also use the notation $\omega_{s}\coloneqq \omega_{[t+1,t+s]}$. A translation invariant state is an FCS if there exist 

\begin{itemize}
 \item a number $m\in \mathbb{N}$
 \item a list of matrices $\mathbb{E}_{k_i,l_i}\in \mathbb{M}_{m,m}(\mathbb{C})$, $\, k_i,l_i\in[d_{\A}]$, 
 \item a row vector $\rho\in \mathbb{C}^{m_1}$ such that $\sum_{k_i=1}^{d_{\mathcal{A}}}\sum_{l_i=1}^{d_{\mathcal{A}}}\rho \mathbb{E}_{k_i,l_i}=\rho$,
 \item and a column vector $e\in \mathbb{C}^{m_n}$ such that $\sum_{k_i=1}^{d_{\mathcal{A}}}\sum_{l_i=1}^{d_{\mathcal{A}}} \mathbb{E}_{k_i,l_i}e=e$,
 \end{itemize}

such that the marginals of size $t$, $\omega_t:=\omega_{[1,t]}$ can be written as

\begin{equation}
 \bra{k_1,...k_t}\omega_{t}\ket{l_1,...l_t}=\rho\mathbb{E}_{k_1,l_1}...\mathbb{E}_{k_t,l_t}e. 
 \end{equation}

For the corresponding functionals, denoting $\mathbb{E}_A:=\sum_{k=1}^{d_{\mathcal{A}}}\sum_{l=1}^{d_{\mathcal{A}}} \mathbb{E}_{\bra{i}A\ket{j}}$  we have, for every $A_{i}\in \mathbb{M}_{{d_{\mathcal{A}}},{d_{\mathcal{A}}}}(\mathbb{C})$, $i\in{[t]}$

\begin{equation}
\tilde{\omega}_{t}(A_1\otimes A_2\otimes ...\otimes A_t)=\rho\mathbb{E}_{A_1}...\mathbb{E}_{A_t}e,
\end{equation}

and $\tilde{\omega}_{t}$ is extended to all matrices in $\mathbb{M}_{d_{\mathcal{A}}^t,d_{\mathcal{A}}^t}(\mathbb{C})$ by linearity. We say that $\{\mathbb{E},\rho,\tau\}$ define a realization of $\omega$.

As it is natural for states on the infinite chain, the functional representation was used in~\cite{Fannes1992}. 
In the Appendix we will also use the functional notation (and just write $\omega$ instead of $\tilde{\omega}$), as in this case it is better suited to the type of manipulations we do in the technical analysis. In the main text, we will use the density matrix notation. However, we want to stress that since we will always be interested in finite-size marginals of states on the infinite chain, we could in principle work with just density matrices, and going back and forth between the two pictures is immediate. In the following, we will focus our attention on (marginals of) translation-invariant states. Similar ideas can be applied to non-translation-invariant states on a finite chain, and we redirect to \cref{secnontrans} for the details.

Any state admits an MPDO representation with sufficiently high $m_i$. Similarly, any translation-invariant state on the infinite chain can be obtained as a limit of FCSs (in the weak-* topology)~\cite{Fannes1992}, but $m$ may need to diverge. Clearly, a realization of $\omega$ with in a vector space of dimension $m\leq m'$ can be also be seen as a realization in a vector space with dimension $m'$. It is then natural to use the dimension of the realization as an expressivity parameter: the larger the dimension, more states can be realized. We can thus consider the problem of learning states with realizations with dimension bounded by a constant $m$. As we will show, the sample complexity of this problem scales polynomially in $m$. However, this is not the only parameter appearing in our bounds. Other parameters appear, which are related to the hardness of reconstructing a realization from knowledge of the marginals. We now explain how they arise.

A realization of an FCS of minimal dimension, say $m$, is called regular. For any integers $t_1<0$ and $t_2\geq 0$ we can associate the segments $[t_1,0]$ and $[1,t_2]$ of the chain to matrices $\mathbb{M}_{d^{t_1+1},d^{t_1+1}}(\mathbb{C})$, with an orthonormal basis $\mathcal{B}_{R}$ and $\mathbb{M}_{d^{t_2},d^{t_2}}(\mathbb{C})$, with an orthonormal basis $\mathcal{B}_{L}$,
and define a linear map $\Omega$ from $\mathbb{M}_{d^{t_2},d^{t_2}}(\mathbb{C})$ matrices to functionals on $\mathbb{M}_{d^{t_1+1},d^{t_1+1}}(\mathbb{C})$, as 
\begin{equation}
\Omega(Y)[X]:=\tilde{\omega}(X\otimes Y)]=\Tr[\omega_{[t_1,t_2]}(X\otimes Y)].
\end{equation}

Identifying $\Omega$ with its corresponding matrix $\Omega$ in a self-adjoint basis (with a slight abuse of notation), one obtains a $\mathbb{M}_{d^{2(t_1+1)},d^{2t_2}}(\mathbb{C})$ real matrix. 
As from the definition, its matrix elements can be obtained simply by expectation values of $\omega$: the matrix element $\Omega_{ij}$ is equal to the expectation value $\Tr[\omega_{[t_1,t_2]}(X_i\otimes Y_j)]$, where $X_i$ and $Y_j$ are the observables in respectively $\mathcal{B}_{L}$,  with index $i$, and $\mathcal{B}_{R}$, with index $j$. In fact, the matrix elements of $\Omega$ are simply a rearrangement in a matrix form of the coefficients of $\omega_{[t_1,t_2]}$ in a product basis (across the bipartition), which means that learning $\Omega$ is equivalent to learning $\omega_{[t_1,t_2]}$. We can then construct the real singular value decomposition $\Omega=UDO$, with $U\in \mathbb{M}_{d^{2(t_1+1)},m}(\mathbb{R})$ being a matrix with orthonormal columns, $O\in \mathbb{M}_{m,d^{2t_2}}(\mathbb{R})$ being a matrix with orthonormal rows, and $D\in \mathbb{M}_{m,m}(\mathbb{R})$ diagonal and nonnegative. Similarly, we can define the maps $\Omega_A(Y)[X]:=\omega(X\otimes A\otimes Y)$ from $\mathbb{M}_{d^{t_2},d^{t_2}}(\mathbb{C})$ matrices to functionals on $\mathbb{M}_{d^{t_1+1},d^{t_1+1}}(\mathbb{C})$. One can thus form the matrices $U^{\intercal} \Omega\in \mathbb{M}_{m,d^{2t_2}}(\mathbb{C})$, $U^{\intercal} \Omega_A\in \mathbb{M}_{m,d^{2t_2}}(\mathbb{C})$, $(U^{\intercal}\Omega)^+\in \mathbb{M}_{d^{2t_2},m}(\mathbb{C})$, which is the Moore-Penrose inverse of $U^{\intercal}\Omega$. 

One can show that the rank of $\Omega$ can only be smaller or equal than $m$, and it cannot decrease by increasing $-t_1$ or $t_2$.  A condition equivalent to this rank saturation is that for a realization of minimal dimension $m$, 
\begin{itemize}
\item the column vectors $\{\mathbb{E}_{k_1,l_1}...\mathbb{E}_{k_{t_1},l_{t_1}}e\}$, $k_{i},l_{i}\in[d_{\mathcal{A}}]$, $1 \leq i\leq  t_2$ span the full space $\mathbb{C}^{m}$, and
\item the row vectors $\{\rho\mathbb{E}_{k_{-t_1},l_{-t_1}}...\mathbb{E}_{k_{t_0},l_{t_0}}\}$, $k_{i},l_{i}\in[d_{\mathcal{A}}]$, $-t_1\leq i\leq 0$ span the full space $\mathbb{C}^{m}$.
\end{itemize}
For $-t_1,t_2$ large enough it becomes equal to $m$: see Corollary~\ref{thm:Vij} for a proof, where we also show that it is enough to take $-t_1\geq m-1,t_2\geq m$, although this is just a worst-case scenario.

This motivates to include another complexity parameter for the states that we want to learn, which is the minimum $t^*$ such that the rank of $\Omega$ computed on $[-t^*+1,t^*]$ saturates to its maximum value. We also say that a state is $s$-reconstructible if $\Omega$ computed on $[-s+1,s]$ saturates to its maximum value. 

Let us now denote by $v_{\1}$ the column vector made of the coeffients of the identity operator on $[1,t_2]$, in the chosen basis, and similarly $w_{\1}$ the row vector of coeffients of the identity operator on $[-t_1,0]$ in the chosen basis. 

As soon as the rank of $\Omega$ is $m$, a regular realization (called the regular realization) can be obtained as follows:
\begin{itemize}
\item A column vector $e:=U^{\intercal} \Omega v_{\1}\in \mathbb{C}^{m}$,\label{eqe1intro}
\item A row vector $\rho:= w_{\1} U =w_{\1}(\Omega (U^{\intercal}\Omega)^{+})\in \mathbb{C}^{m}$,\label{eqrho1intro}
\item A list of matrices $\mathbb{K}_{k,l}:=U^{\intercal}\Omega_{\ketbra{k}{l}} (U^{\intercal}\Omega)^+\in \mathbb{M}_{m,m}(\mathbb{C})$, $k,l\in[d]$.
\end{itemize}

This fact suggest that an approximate realization for $\tilde{\omega}$ can be obtained by an approximate estimation of the density matrix $\omega_{[-t_1,t_2]}$. The sensitivity of this procedure to errors in the estimation of $\omega_{[-t_1,t_2]}$ is determined by the sensitivity of the linear algebra operations. The most delicate operation is the Moore-Penrose inverse of $(U^{\intercal}\Omega)^+$. This motivates another complexity parameter, i.e. a lower bound $\eta>0$ on the the smallest non-zero singular value of $\Omega$ computed on $[-s+1,s]$, denoted by $\sigma_m(\Omega,s)$. 

A more restricted class of models are those where the maps $\mathbb{E}_{k,l}$ come from an actual quantum channels: in this case $\rho$ is a state of a hidden memory system, which evolves via consecutive applications of a channel $\mathcal{E}^{\dagger}$ from the memory system to the joint system of the memory and a site of the chain. Here, the memory system has the special structure of a space of operators on a finite-dimensional Hilbert space, and we use the notation $\mathcal{L}(\mathcal{H})$ to denote the linear operators on $\mathcal{H}$. The dimension of this space is $\dim(\mathcal{H})^2$. To keep the formalism closer to that of general models, it is useful to consider the unital completely positive map $\mathcal{E}$, the adjoint of $\mathcal{E}^{\dagger}$. More in detail, a $C^*$-realization (or quantum realization) of $\tilde{\omega}$ is composed of:

\begin{itemize}
 \item a Hilbert space $\mathcal{H_{\mathcal{B}}}$ of dimension $d_{\mathcal{B}}\in \mathbb{N}$, with identity operator $\1_{\mathcal{B}}$.
 \item a unital completely positive map $\mathcal{E}:\mathcal{L}(\mathcal{H_{\mathcal{B}}}\otimes \mathcal{H_{\mathcal{A}}})\rightarrow \mathcal{L}(\mathcal{H_{\mathcal{B}}})$, 
 \item a density matrix $\rho$ on $\mathcal{H_{\mathcal{B}}}$ such that $ \Tr_{\mathcal{H_{\mathcal{A}}}}[\mathcal{E}^{\dagger}(\rho)]=\rho$,
 \end{itemize}

such that, denoting $\mathcal{E}_A:\mathcal{L}(\mathcal{H_{\mathcal{B}}})\rightarrow \mathcal{L}(\mathcal{H_{\mathcal{B}}})$ the map defined as $\mathcal{E}_A(B)=\mathcal{E}(A\otimes B)$, we have

\begin{equation}
\bra{k_1,...k_t}\omega_{t}\ket{l_1,...l_t}=\Tr_{\mathcal{H_{\mathcal{B}}}}[\rho\mathcal{E}_{\ketbra{k_1}{l_1}}... \mathcal{E}_{\ketbra{k_t}{l_t
}}]=\tilde{\rho}(\mathcal{E}_{\ketbra{k_1}{l_1}}... \mathcal{E}_{\ketbra{k_t}{l_t}}(\1_{\mathcal{B}})).
\end{equation}

In functional notation, this reads

\begin{equation}
\omega_{t}(A_1\otimes \cdots A_t)=\tilde{\rho}(\mathcal{E}_{A_1}... \mathcal{E}_{A_t}(\1_{\mathcal{B}})).
\end{equation}

We can thus consider the following class of states:

\begin{definition}\label{defrelclass}
We define $\mathcal{S}(m,s,\eta)$ as the class of translation invariant states on the infinite chain such that they admit a realization with dimension less than $m$, $s$-reconstructible, and with $\sigma_{m}(\Omega,s)\geq \eta$.
\end{definition}

\begin{definition}\label{defrelclassc}
We define $\mathcal{S}_q(d_{\mathcal{B}},s,\eta)$ as the class of translation invariant states on the infinite chain such that they admit a quantum realization with a memory Hilbert space of dimension less than $d_{\mathcal{B}}$, $s$-reconstructible, and with $\sigma_{m}(\Omega,s)\geq \eta$.
\end{definition}

From the definition, it is clear that $\mathcal{S}_q(d_{\mathcal{B}},s,\eta)\subseteq \mathcal{S}(d_{\mathcal{B}}^2,s,\eta)$. As we will now see, assuming the existence of a quantum realization can give an improved bound on the sample complexity.

\subsection{Learning algorithm for finitely correlated states}
Our goal is to learn an estimate of a realization of $\tilde\omega$ with the promise that $\tilde{\omega}$ belongs to $\mathcal{S}(m,s,2\eta)$ or $\mathcal{S}_q(d_{\mathcal{B}},s,2\eta)$. To this end, we use Algorithm~\ref{alg:learnFCS},  named \textsf{LearnFCS}, which 
\begin{itemize}
\item takes as input copies of $\hat{\omega}_{2s+1}$, $s$, $\eta$ and a choice of bases, 
\item  outputs estimated realization parameters $\hat{\rho}, \hat{e}, \hat{\mathbb{K}}_{Z_t}$,
\end{itemize}

\begin{figure}[ht]

\centering
\begin{minipage}{.9\linewidth}
\begin{algorithm}[H]
\caption{\textsf{LearnFCS}} 
	\label{alg:learnFCS}
	\begin{algorithmic}[1]
        \State Input: State $\omega_{2s+1}^{\otimes n}$, $s\in\mathbb{N}$, $\eta\in\mathbb{R}$, bases $\mathcal{B}_{\mathcal{A}} = \lbrace Z_i \rbrace_{i=1}^{d^2_{\mathcal{A}}} ,\mathcal{B}_{L} = \lbrace Y_i \rbrace_{i=1}^{d^{2s}_{\mathcal{A}}} , \mathcal{B}_{R} = \lbrace X_i \rbrace_{i=1}^{d^{2s}_{\mathcal{A}}} $
        \State Use \textsf{Tomography}$(\omega_{2s+1}^{\otimes n},s)$ and produce empirical estimates: \begin{itemize}
        \item $\widehat{\Omega(\1)}[Y_j]=\Tr[\hat{\omega}_s Y_j ]$
                \item $\widehat{\tau\Omega}[X_i]=\Tr[\hat{\omega}_s X_i]$ 
            \item 
        $\hat{\Omega}(X_i)[Y_j] = \Tr[\hat{\omega}_{2s}( Y_j \otimes X_i )] $
        \item $\hat{\Omega}_{Z_k}(X_i)[Y_j] = \Tr[\hat{\omega}_{2s+1}(Y_j\otimes Z_k \otimes X_i)]$
        \end{itemize}
        for all $Z_k\in\mathcal{B}_{\mathcal{A}}$, $Y_j\in\mathcal{B}_{\mathcal{A}_L}$ and $X_i\in\mathcal{B}_{\mathcal{A}_R}$.
        \State Compute the SVD of $\hat{\Omega} = \hat{U}\hat{D}\hat{O}^\intercal$ and truncate the columns of $\hat{U}$ such that   we keep the columns corresponding the singular values $\sigma_i ( \hat{\Omega})  \geq \eta/2$.
        \State Output the estimated realization operators:
        \begin{itemize}
            \item $\hat{e}={\hat{U}}^{\intercal}\widehat{\Omega(\1)} $
            \item $\hat{\rho}=\widehat{\tau \Omega}(\hat{U}^{\intercal}\hat\Omega )^+$
            \item $\hat{\mathbb{K}}_{Z_i}={\hat{U}}^{\intercal}(\hat\Omega_{Z_i})({\hat{U}}^{\intercal}\hat\Omega)^+$ for all basis elements $Z_i\in\mathcal{B}_\mathcal{A}$.
        \end{itemize}
	\end{algorithmic} 
\end{algorithm}
\caption{Finitely correlated state learning algorithm. \textsf{Tomography}$(\omega_{2s+1}^{\otimes n},s)$ is any subroutine that takes input $n$ copies of any state $\omega_{2s+1}$ and estimates the marginals $\hat{\omega}_{s}$, $\hat{\omega}_{2s}$ and $\hat{\omega}_{2s+1}$ on a subset of the chain of length respectively $2s$ and $2s+1$.}
\end{minipage}
\end{figure}

We then let $\hat{\omega}_t$ the operator reconstructbile from these parameters, i.e. the operator such that 
        \[  \Tr[\hat\omega_{t}(Z_1\otimes\cdots\otimes Z_t)]:=\hat{\rho}\hat{\mathbb{K}}_{Z_1}\cdots\hat{\mathbb{K}}_{Z_t}\hat{e}\] for all $Z_i\in\mathcal{B}_{\mathcal{A}}$.
We stress that the output of \textsf{LearnFCS} are a pair of vectors in $\mathbb{C}^{m}$, and $d^2_{\mathcal{A}}$ matrices $\mathbb{M}_{m,m}(\mathbb{C})$. The density matrix $\hat{\omega}_t$ requires $d_{\mathcal{A}}^{2t}$ computations of the matrix elements to be completely specified, but each matrix element can be computed in linear time once the estimated realization is obtained.

The first step of \textsf{LearnFCS} is to do tomography from copies of $\omega_{2s+1}$ to find estimates $\hat{\omega}_s,\hat{\omega}_{2s}, \hat{\omega}_{2s+1}$ of ${\omega}_s,{\omega}_{2s}, {\omega}_{2s+1}$ respectively. This can be done with any preferred quantum algorithm. Then, the realization parameters are estimated plugging the marginal estimates in the formulas for the observable realization. An important step is to remove spurios singular values (stictly smaller than $\eta$) from the SVD of the estimation of $\Omega$, which allows to obtain a realization of minimal dimension $m$ when the error in the estimation of $\omega_{2s}$ is small enough. As a matter of fact, the algorithm does not need to know a bound on $m$ in advance, while $s$ and $\eta$ are parameters of the algorithm. Once the realization parameters have been estimated, each matrix element of $\hat\omega_t$, in a product basis, can be computed in time linear in $t$. In the same way, expectation values of product observables can be computed in time linear in $t$.
A similar procedure can be used for non-translation invariant states on a finite chain, as we discuss in \cref{secnontrans}.

The parameters $s,\eta$ appear in the definition of the class because the algorithm and the sample complexity bounds we obtain depend explicitly on them. Some remarks are in order.

\begin{remark}
As we already stated, any state with realization of dimension $m$ is $m$-reconstructible, which means $\mathcal{S}(m,s,\eta)\subseteq \mathcal{S}(m,m,\eta)$. However, as we mentioned in the introduction we expect that, generically, a state $m$-dimensional realization is $O(\mathrm{poly}(\log m))$-reconstructible. This is at least true for models generated by random quantum channels, due to a generic quantum Wielandt inequality \cite{jia2024generic,KLEP201656}. An heuristic argument for the general case is the following. As we mentioned above, the rank saturation condition is equivalent to the column vectors  $\{\mathbb{E}_{k_1,l_1}...\mathbb{E}_{k_{t_1},l_{t_1}}e\}$, $k_{i},l_{i}\in[d_{\mathcal{A}}]$, and
the row vectors $\{\rho\mathbb{E}_{k_{-t_1},l_{-t_1}}...\mathbb{E}_{k_{t_0},l_{t_0}}\}$ spanning the full space $\mathbb{C}^{m}$. If the maps $\mathbb{E}_{k,l}$ are generic linear maps, $m$ row vectors and $m$ column vectors constructed as above are already generically linearly independent, meaning that we can take $s=O(\log m)$. This is not a complete argument as we have additional conditions on the maps $\mathbb{E}_{k,l}$, $\rho$ and $e$, that is, they give rise to marginals of a bona-fide translation invariant state on an infinite chain. However, there is no apparent reason for these additional conditions to be an obstacle to a complete argument.

If $s$ is of the order of $m$,  the learning algorithm requires to deal with density matrices $\hat\omega_{m}$, $\hat\omega_{2m}$, $\hat\omega_{2m+1}$ of size exponential in $m$, which may be undesirable. The argument just made shows that in practice, one may reasonably hope to target a class $\mathcal{S}(m,O(\mathrm{poly}(\log m)),\eta)$, in which case the density matrices required are only of size $O(\mathrm{poly}(m))$.
\end{remark}

\begin{remark}
    
 We could not find a compelling physical interpretation of the singular value $\sigma_{m}(\Omega,s)$ in terms of usual properties of bipartite states, such as their local spectra of the marginals.
Nonetheless, while it is clear that $\sigma_{m}(\Omega,s)$ (and thus $\eta$) can appear in the analysis of the error on the estimation of the realization parameters due to the sensitivity of the matrix inversion, it also plays a role in bounding the error propagation in computing $\hat{\omega}_t$ from approximate realization parameters, as we will explain in the overview of the proofs. This has to do with the role that $\sigma_{m}(\Omega,s)$ has in establishing inequalities between a weighted Euclidean norm and an operator norm on the memory system, viewed as an operator system. Operationally, it relates an Euclidean distance in a coordinate system for the vector space of the general probabilistic theory with a quantifier of distinguishability via measurements in the same theory.
 \end{remark}

\begin{remark}

 Let us denote $\Omega_{[1-s,s]}$ the map $\Omega$ constructed from the subchain $[1-s,s]$. An upper bound on $\sigma_{m}(\Omega,s)$ is $\sigma_{m}(\Omega,s)^2\leq \Tr[\Omega_{[1-s,s]}^2]=\Tr[\omega^2_{[1-s,s]}]$. This means that $\sigma_{m}(\Omega,s)^2$ cannot be larger than the purity of the marginals, and this upper bound can decrease exponentially fast in $s$, in the worst case. In fact, this happens if the state is a product state. In that case, the minimal realization has $m=1$, $s=1$ is sufficient for the reconstruction, and $\sigma_{m}(\Omega,s)^2=\sigma_{1}(\Omega,s)^2=\Tr[\omega^2]^{2s}\geq \frac{1}{d_{\mathcal{A}}^{2s}}$. Therefore, implementing the algorithm with arbitrary $s$ would require to have $\eta$ exponentially small in $s$ to have the guarantee to correctly reconstruct every product state. To avoid this shortcoming, one could choose the $s$ and $\eta$ to give as input to \textsf{LearnFCS} only after a pre-learning phase, where $\Omega_{[1-s,s]}$ is estimated for every $s$ larger than some fixed $s_{\max}$ that we choose beforehand. The number of copies of $\Omega_{[1-s,s]}$ must be such that the singular values are learned with error at most $\eta/3$.  We then choose the minimum $s$ such that the number of singular values larger than $\eta/2$, of the estimated $\hat\Omega_{[1-s,s]}$, is maximal. If we have the promise that there is at least one $s\leq s_{\max}$ such that the rank of $\Omega_{[1-s,s]}$ is maximal and all the singular value are larger than $\eta$, we are guaranteed to find the minimal such $s$ with high probability. In particular, $\eta\leq \frac{1}{d_{\mathcal{A}}}$ is sufficient to reconstruct product states. This preprocessing adds a subleading sample and computational complexity overhead to the analysis of \textsf{LearnFCS}.
\end{remark}

 To illustrate the effectiveness and practicality of 
 \textsf{LearnFCS}, we simulated the reconstruction of the AKLT ground state~\cite{affleck1988valence}.
 In this case, while $m=8$, $s=1$ is sufficient for the reconstruction, showing that the error behaves as predicted, see Figure~\ref{fig:results_simulation}. Technical details are reported in~\cref{sec.numer}.
 This supports the significance of our analysis for practical applications, as the {size of the necessary marginals is typically small with respect to the local dimensions}. In fact, for small $s_{\max}$, the classical post-processing can be easily implemented on a laptop.

\begin{figure}[ht]
  \begin{minipage}{0.5\textwidth}
    \centering
    \includegraphics[width=\linewidth]{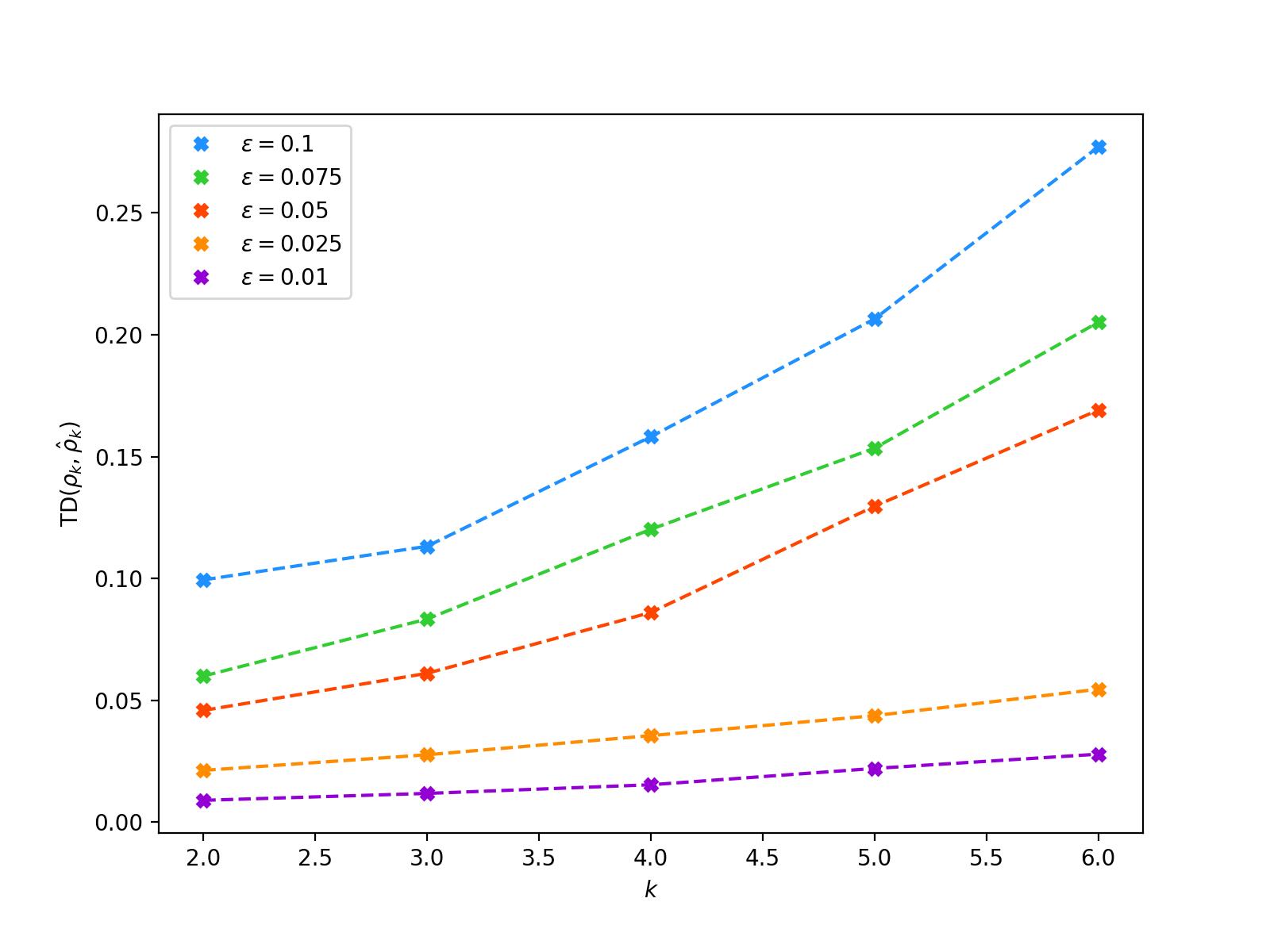}
    \label{fig:figure1}
  \end{minipage}%
  \begin{minipage}{0.5\textwidth}
    \centering
    \includegraphics[width=\linewidth]{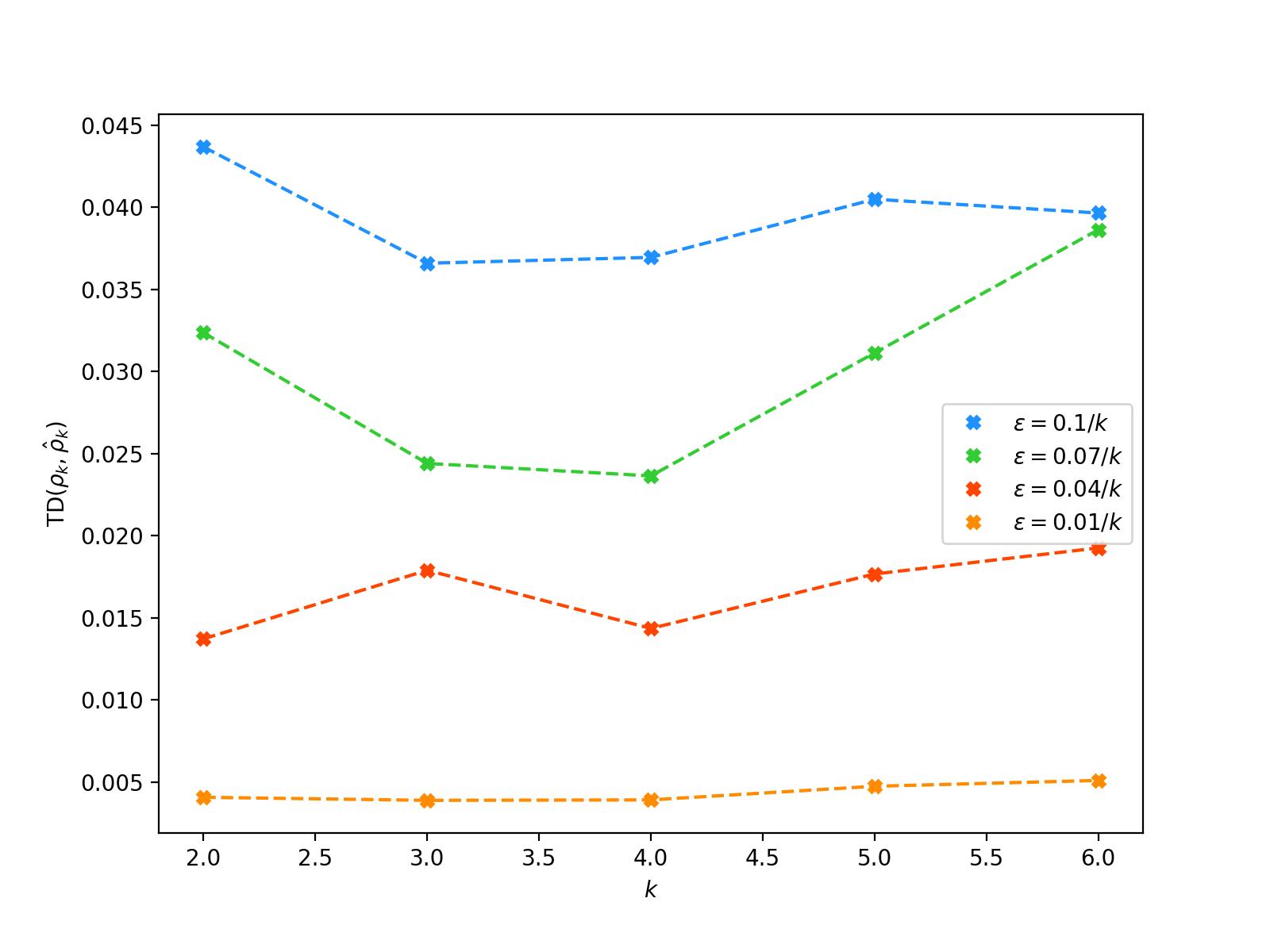}
    \label{fig:figure2}
  \end{minipage}
  \caption{Scaling of trace norm error $TD(\omega_k,\hat{\omega}_k ) = \frac{1}{2}\| \rho_k - \hat{\rho}_k \|_1$ between the true reduced density matrix $\rho_k$ of the ground state of the AKLT model and the reconstructed one $\hat{\rho}_k$ from the maps $\hat{\Omega},\hat{\Omega}_{(\cdot )}$. We study the scaling for different errors $\epsilon = \|\hat{\Omega} - \Omega \|_2 = \| \hat{\Omega}_{(\cdot)} - \Omega_{(\cdot)}  \|_2$ 
  \textbf{Left plot:} each line corresponds to a fixed $\epsilon$ and the reconstruction of $\omega_k$ is done using the same $\hat{\Omega},\hat{\Omega}_{(\cdot)}$. This plot illustrates the exponential behavior of the error propagation Theorem~\ref{thmerrorpropagation} coming from the factor $(1+\Delta)^t$.
  \textbf{Right plot:} each dot corresponds to an independent experiment and the error is normalized by the number of sites. This plot illustrates the linear behavior if we renormalize the errors since $(1+\Delta/k)^k \leq 1+2k$ (for $\Delta\leq 1/2$).}
  \label{fig:results_simulation}
\end{figure}

\subsection{Main theorems}

We are now ready to state our main results. The first theorem establishes how the error in Hilbert Schmidt distance on the marginals, $D_{HS}(A,B):=\Tr[(A-B)^\dagger(A-B)]^{1/2}$, propagates to the the trace norm error $d_{\Tr}(A,B)=\frac{\|A-B\|_1}{2}$ of the reconstructed operators:

\begin{theorem}\label{theomainintro}
Let $\tilde{\omega}\in \mathcal{S}(m,s,\eta)$. Let $\hat{\omega}_{s},\hat{\omega}_{2s}, \hat{\omega}_{2s+1}$ be estimates of ${\omega}_{s},{\omega}_{2s}, {\omega}_{2s+1}$ respectively, such that $D_{HS}(\hat{\omega}_{s},{\omega}_{s})$, $D_{HS}(\hat{\omega}_{2s},{\omega}_{2s})$, $D_{HS}(\hat{\omega}_{2s+1},{\omega}_{2s+1})$ are smaller than $\frac{\epsilon\eta^3}{20 t m \sqrt{d_{\A}}}$. Then, $\hat{\omega}_t$ constructed from the output of Algorithm~\ref{alg:learnFCS} satisfies
\begin{equation}
\frac{\|\hat{\omega}_t-\omega_t\|_1}{2}\leq \epsilon.
\end{equation}
\end{theorem}

An alternative bound can be established if we also have the promise that a quantum realization exists.

\begin{theorem}\label{theomainintroc}
Let $\tilde{\omega}\in \mathcal{S}_q(d_{\mathcal{B}},s,\eta)$. Given $\hat{\omega}_{s},\hat{\omega}_{2s}, \hat{\omega}_{2s+1}$ estimates of ${\omega}_{s},{\omega}_{2s}, {\omega}_{2s+1}$ respectively, such that $D_{HS}(\hat{\omega}_{s},{\omega}_{s})$, $D_{HS}(\hat{\omega}_{2s},{\omega}_{2s})$, $D_{HS}(\hat{\omega}_{2s+1},{\omega}_{2s+1})$ are smaller than $\frac{\epsilon\eta^3}{20 t d_{\mathcal{B}} \sqrt{d_{\A}}}$. Then, $\hat{\omega}_t$ constructed from the output of Algorithm~\ref{alg:learnFCS} satisfies
\begin{equation}
\frac{\|\hat{\omega}_t-\omega_t\|_1}{2}\leq \epsilon
\end{equation}
\end{theorem}

We note that the statement of the second theorem is not trivially implied by the first, as the minimal dimension is smaller than $d_{\mathcal{B}}^2$ but not necessarily smaller than $d_{\mathcal{B}}$.

From these error propagation bounds, sample complexity bounds are immediate. 

\begin{theorem}
 Suppose that for any $0<\epsilon_{HS}\leq 1$ there is a tomography algorithm that, for any $\tilde{\omega}\in \mathcal{S}(m,s,\eta)$, given as input $n(\epsilon_{HS},\delta)$ copies of ${\omega}_{2s+1}$ outputs estimates $\hat{\omega}_{s},\hat{\omega}_{2s}, \hat{\omega}_{2s+1}$ of ${\omega}_{s},{\omega}_{2s}, {\omega}_{2s+1}$ respectively, such that $D_{HS}(\hat{\omega}_{s},{\omega}_{s})$, $D_{HS}(\hat{\omega}_{2s},{\omega}_{2s})$, $D_{HS}(\hat{\omega}_{2s+1},{\omega}_{2s+1})$ are smaller than $\epsilon_{HS}$ with probability larger than $1-\delta$. 
Then, for any $\tilde{\omega}\in \mathcal{S}(m,s,\eta)$ otherwise unknown, $n\left(\frac{\epsilon\eta^3}{20 t m \sqrt{d_{\A}}},\delta\right)$ copies of ${\omega}_{2s+1}$ are sufficient to obtain an estimate $\hat{\omega}_t$ from Algorithm~\ref{alg:learnFCS} which satisfies
\begin{equation}
\frac{\|\hat{\omega}_t-\omega_t\|_1}{2}\leq \epsilon
\end{equation}
with probability larger than $1-\delta$.
With $\mathcal{C}_q(d_{\mathcal{B}},s,\eta)$ instead of $\mathcal{S}(d_{\mathcal{B}},s,\eta)$ above, the same is true with $n\left(\frac{\epsilon\eta^3}{20 t d_{\B} \sqrt{d_{\A}}},\delta\right)$ copies of ${\omega}_{2s+1}$.
\end{theorem}

Proofs of the above theorems are found in~\cref{sec.statereconstruction}.

\begin{remark}
In the following we set $d_{\B}=+\infty$ if a finitely correlated state does not admit a finite dimensional $C^*$-realization. With local (Pauli) measurements, one can set a number of copies of $\omega_{2s+1}$ equal to $n=O\left(\frac{d_{A}^{3(2s+1)}}{\epsilon_{HS}^2}\right)$ and learn $\omega_{2s+1}$, $\omega_{s}$, $\omega_{2s}$ up to error $\epsilon_{HS}$ and constant error probability (see Theorem 1 and Eq.~6 in~\cite{Guţă_2020} which gives an error guarantee in operator norm, which can be converted in a guarantee in Hilbert-Schmidt distance using $\|A\|_{\infty}\leq \sqrt{d}\|A\|_{2}$ for $d\times d$ matrices), obtaining 
\begin{equation}
n= O\left( \frac{t^2 \min(m^2,d_{\B}^{2}) d_{\A}^{6s+4}}{\epsilon^2\sigma_{m}(\Omega)^{6}}\right).
\end{equation}
With global but single-copy measurements, one can set $n= \tilde O\left(\frac{s^3d^2_{\A}m^2}{\epsilon_{HS}^2}\right)$ (see Theorem~1 in \cite{qin2023stable}), so that
\begin{equation}
n=\tilde O \left(\frac{s^3  t^2 m^2 \min(m^2,d_{\B}^{2}) d_{\A}^{3}}{\epsilon^2\sigma_{m}(\Omega)^{6}}\right).
\end{equation}
For simplicity, we gave the above bounds for (arbitrary) constant error probability, see the references for the specific dependence, which is anyway $O\left(\mathrm{polylog}{\frac{1}{\delta}}\right)$.
\end{remark}

We now comment on the extensions of the main result, which we further discuss in the appropriate sections below: 

\begin{itemize}

\item The non-homogeneous case can be addressed with essentially the same proof.
We provide a reconstruction algorithm in our notation, and argue that the error propagation bound has the same form as the translation invariant case. Of course, one needs that all the estimations of the marginals are  $\epsilon_{HS}$ accurate. If each marginal is estimated in an independent run of measurements, this increases the number of copies by a factor $\tilde O(t)$ in the size of the state. See \cref{secnontrans}.

\item We consider the implications for learning one-dimensional Gibbs states in terms of their approximations by matrix product operators. A sufficiently good finitely correlated approximation, reconstructible form marginals of small size, and with sufficiently large small singular falues of $\Omega$,
would imply a polynomial sample complexity for the task of learning the reduced states of a Gibbs state at any positive temperature, similarly to that of 
\cite{anshu2020sample}. However, we were not able to find approximations which satisfy this conditions, and the applicability of this idea is currently unclear,  see \cref{sec.learningclose}. This motivates further study towards improving the learning algorithm to be more robust, better controlling MPDO approximations of thermal states, and understanding which properties allow states to have MPO approximations which are sufficiently close for the learning algorithm to work.

\end{itemize}

\section{Mathematical tools and overview of proofs}

\subsection{Heuristic: error propagation for quantum models}

The main rationale behind the error propagation bound is to generalize a well-known method to bound distances between a pair of states to which repeated quantum operations are applied, using the telescopic trick. We recall the definition of the diamond norm of a map $\Phi:\mathcal{L}(\mathcal{H}_{\A})\rightarrow \mathcal{L}(\mathcal{H}_{\B})$:

\begin{equation}
\|\Phi\|_{\diamond}:=\sup_{n\in \mathbb{N}}\sup_{A\in \mathbb{M}_{n,n}(\mathbb{C})\otimes \mathcal{L}(\mathcal{H}_{\A}),\|A\|_1\leq 1}\|(\id_{\mathbb{M}_{n,n}(\mathbb{C})}\otimes \Phi)(A)\|_1,
\end{equation}

and of the completely bounded norm, for maps $\E:\mathcal{L}(\mathcal{H}_{\B})\rightarrow \mathcal{L}(\mathcal{H}_{\A})$

\begin{equation}
\|\E\|_{cb}:=\sup_{n\in \mathbb{N}}\sup_{A\in \mathbb{M}_{n,n}(\mathbb{C})\otimes \mathcal{L}(\mathcal{H}_{\B}),\|A\|_{\infty}\leq 1}\|(\id_{\mathbb{M}_{n,n}(\mathbb{C})}\otimes \E)(A)\|_{\infty}.
\end{equation}

They are dual to each other, meaning $\|\Phi\|_{\diamond}=\|\Phi^{\dagger}\|_{cb}$. If $\Phi$ is a channel we have $\|\Phi\|_{\diamond}=1$.

A well known result (see, e.g., \cite{Watrous_2018}) is that, for a collection of channels and Hilbert spaces $\{\Phi_{i}:\mathcal{L}(\mathcal{H}_i)\rightarrow \mathcal{L}(\mathcal{H}_{i+1})\}$, $\{\tilde{\Phi}_{i}:\mathcal{L}(\mathcal{H}_i)\rightarrow \mathcal{L}(\mathcal{H}_{i+1})\}$, $i\in[t]$, and defining $\Phi:=\circ_{i=1}^t\Phi_i$, $\tilde{\Phi}:=\circ_{i=1}^t \tilde{\Phi}_{i}$

\begin{equation}\label{eq:datasumproc}
\|\Phi-\tilde{\Phi}\|_{\diamond}\leq \sum_{i=1}^t \|\Phi_i-\tilde{\Phi}_i\|_{\diamond}.
\end{equation}

The marginal of a quantum realization with parameters $\rho,\mathcal{E}$ can be expressed in terms of the channels $\Phi_{i}: \mathcal{L}(\mathcal{H}_{\B}\otimes \mathcal{H}^{\otimes {i-1}}_{\A})\rightarrow \mathcal{L}(\mathcal{H}_{\B}\otimes \mathcal{H}^{\otimes {i}}_{\A})$ as $\Phi_{i}=\mathcal{E}^{\dagger}\otimes \id_{\mathcal{H}^{\otimes {i-1}}_{\A}}$, as $\omega_{t}=\Phi(\rho)$ using the notation introduced above. Let $\hat\omega_{t}$ be another state obtained from a different initial memory state $\hat\rho$ but different channel $\hat{\mathcal{E}}^{\dagger}$, so that $\hat\omega_{t}=\hat\Phi(\hat\rho)$. Then we have

\begin{equation}\label{eq:dataproc}
d_{\Tr}(\omega_{t},\hat\omega_{t})=\frac{\|\Phi(\rho)-\hat\Phi(\hat\rho)\|_{1}}{2}\leq \sum_{i=1}^t \frac{\|\Phi_i-\hat{\Phi}_i\|_{\diamond}}{2}+\frac{\|\rho-\hat{\rho}\|_{1}}{2}=t \frac{\|\mathcal{E}^{\dagger}-\hat{\mathcal{E}}^{\dagger}\|_{\diamond}}{2}+\frac{\|\rho-\hat{\rho}\|_{1}}{2},
\end{equation}
which tells us that the error grows linearly in $t$. Conversely, to obtain $d_{\Tr}(\omega_{t},\hat\omega_{t})\leq \epsilon$ it is sufficient to estimate $\mathcal{E}$ at precision $\epsilon/t$ and $\rho$ at precision $\epsilon$. If this was possible with $O(\poly(\eps^{-1}))$ copies of $\omega_t$, it would imply a $O(\poly(t))$ sample complexity for learning $\omega_t$. However, already under the promise that a quantum realization exists, it is not immediately clear how to estimate $\mathcal{E}$ and $\rho$ efficiently without resorting to global measurements on $\omega_t$. Instead, \textsf{LearnFCS} gives us a chance to learn an estimate of the observable realization, which is not necessarily a quantum realization. This means we cannot use the bound on the diamond norm. The main contribution of this paper is to realize how to get an error propagation bound using a different norm, which analogously to what the diamond norm does for quantum states, encodes distinguishability between maps in a general probabilistic theory that describes the states of the joint system of the memory system and the chain.

\subsection{General probabilistic theories and operator systems}

A general probabilistic theory is defined by a (proper, spanning and closed) cone $C$ in a real vector space $V$, a vector $e\in C$ called \textit{unit}, and a proper cone $C'$ in $V^*$ (the space of linear functionals on $V$), such that $f(v)\geq 0$ for all $v\in C$ and $f\in C'$. This simple mathematical structure is a very simple blueprint of a theory describing systems and measurements: 
\begin{itemize}
\item $C\cap (e-C)$ represent the set of events, 
\item $S=\{f\in C':f(e)=1\}$ is the set of states, 
\item a binary measurement of the pair of events $(v,e-v)$ on the state $f$ is described by the probability distribution $(f(v),f(e-v))=(f(v),1-f(v))$,
\item physical transformations between two GPTs with state spaces $S_1$ and $S_2$ and units $e_1$, $e_2$ are linear maps sending $S_1$ into $S_2$. Their adjoint maps coincides with unital maps $\E$, ($\E(e_1)=e_2$) that send $C_1$ into $C_2$.
\end{itemize}
The order norm of the cone on $V$ is denoted as $\|\cdot\|_e$ and defined as $\|v\|_e=\{\min \lambda\geq 0: \pm v \leq \lambda e\}$. The corresponding dual norm on $V'$ is $\|f\|_{e,*}=\{\max |f(v)|: v\in V, \|v\|_e\leq 1\}$. 

For $d$-dimensional quantum theory the above structure corresponds to familiar objects. $V$ is the set of self-adjoint matrices $\mathbb{M}_{d,d}(\mathbb{C})^{sa}$, and $V^*$ as the set of normal functionals $f_Y[X]=\Tr[XY]$, $Y\in \mathbb{M}_{d,d}(\mathbb{C})^{sa}$, $C$ is the set of positive-semidefinite matrices, $e$ is the identity matrix $\1$, and $C'$ can be identified with the set of positive-semidefinite matrices, and $S$ can be identified with the set of density matrices. Physical maps are completely positive trace-preserving superoperators. The order norm and the dual norm simply correspond to the $\infty$ and $1$-Schatten norms respectively,
$\|\cdot\|_{e}=\|\cdot\|_{\infty}$, $\|\cdot\|_{e,*}=\|\cdot\|_{1}$. 

The operational significance of the distances induced by the order norm translates to the GPT setting:

\begin{itemize}
\item The optimal probability of binary state discrimination of two states $f,g$ is $\frac{1+\frac{\|f-g\|_{e,*}}{2}}{2}$
\item Physical maps sending states of a GPT with units $e$, into states of a GPT with unit $e'$ are contractive: $\|\Phi(f)-\Phi(g)\|_{e',*}\leq \|f-g\|_{e,*}$. Thus, distinguishability is decresead by the action of a physical map.
\end{itemize}

An interesting problem is to characterize the possible ways of combining a pair of GPTs in a single one, describing the joint system~\cite{Aubrun2021}.
In quantum theory, we use a simple recipe to build joint systems: given two quantum systems with Hilbert spaces $\H_{1}$, $\H_{2}$, the joint system is described by the tensor product $\H_{1}\otimes \H_{2}$, and the states are identified with density matrices of $\H_{1}\otimes \H_{2}$, which include entangled states. 

Operator systems (which we review in more detail in~\cref{appendixopsys}), arise when we define a collections of GPTs on the joint system of some GPT and finite-dimensional quantum systems, such that completely positive maps connecting the quantum systems are also physical maps for the joint GPTs. Mathematically, an operator system is a collection of closed cones $C_n\subseteq (\mathbb{M}_{n,n}(\mathbb{C})\otimes V)^{sa}$\footnote{here $V$ is a complex vector space, and the adjoint is given by the complex conjugate}, with units $e_n=\1_{n}\otimes e$, and such that for any $v\in C_k$, $v=\sum_{i=1}^k A_i \otimes v_i$, and any matrix $M\in \mathbb{M}_{k',k}(\mathbb{C})$, the vector $v_M=\sum_{i=1}^k MA_iM^{\dagger}\otimes v_i $ satisfies $v_M\in C_{k'}$. 
$C_1$ corresponds to an isolated GPT. Thanks to the cone structure, we can also define order norms $\|\cdot\|_{e_n}$ and their dual norms on $(V\otimes \mathbb{M}_{n,n}(\mathbb{C}))^*$, and also completely bounded norm for a map $\E:V\rightarrow V'$ between two GPTs, with units $e$ and $e'$:

\begin{equation}
\|\E\|_{e\rightarrow e',cb}:=\sup_{n\in \mathbb{N}}\sup_{A\in \mathbb{M}_{n,n}(\mathbb{C})\otimes V,\|A\|_{e_n}\leq 1}\|(\id_{\mathbb{M}_{n,n}(\mathbb{C})}\otimes \E)(A)\|_{e'_{n}},
\end{equation}

and a diamond norm form maps  $\E:V'^*\rightarrow V^*$

\begin{equation}
\|\Phi\|_{\diamond}:=\sup_{n\in \mathbb{N}}\sup_{A\in \mathbb{M}_{n,n}(\mathbb{C})\otimes V'^*,\|A\|_{e'_n,*}\leq 1}\|(\id_{\mathbb{M}_{n,n}(\mathbb{C})}\otimes \E)(A)\|_{e_{n},*}.
\end{equation}

Both norms satisfy triangle inequalities and are submultiplicative under composition $\|\E_1\circ \E_2\|_{cb}\leq \|\E_1\|_{cb}\|\E_2\|_{cb}$, $\|\Phi_1\circ \Phi_1\|_{\diamond}\leq \|\Phi_1\|_{\diamond}\|\Phi_2\|_{\diamond}$, and for completely positive maps that preserve the set of states $\|\Phi\|_{\diamond}=1$. Their adjoints coincides with unital completely positive maps, which satisfy $\|\E\|_{cb}=1$.
This is enough to obtain an analogue of~\eqref{eq:dataproc} when the memory system is a GPT and when $\Phi_i$ are not quantum channels but completely positive maps in the sense of an operator system constructed on the GPT.

Finally, from a state with regular realization given by $\rho, e, \mathbb{E}$, a natural operator system with units $e_n$ arise from the state itself, such that $\mathbb{E}:\mathcal{L}(\mathcal{H}_{A}) \otimes V\rightarrow V $ is a unital completely positive map. Let $\mathbb{E}_i=\id_{\mathcal{L}(\mathcal{H}_{A})^{\otimes i-1}}\otimes \mathbb{E}$ and $\mathbb{E}^t:=\circ_{i=1}^t \mathbb{E}_i$.
The cones of the operator system are

\begin{equation}
C_n=\overline{\{\id_{\mathbb{M}_{n,n}(\mathbb{C})}\otimes \mathbb{E}^t(v\otimes e): v\in \mathbb{M}_{n,n}(\mathbb{C})\otimes \mathcal{L}(\mathcal{H}_{A}^{\otimes t}), v\geq 0, t\in\mathbb{N}\}}.
\end{equation}

We also emphasize here that may be other choices of operator systems. In fact, if a quantum model exists, a choice of operator system exists for the regular realization, that derives from the usual order on matrices as a quotient operator system, see~\cref{sec.opsysfcs}. This is used for the bound which is valid only for states admitting a quantum realization, see~\cref{sec.statereconstructionc}.

Unitality and complete positivity of $\mathbb{E}_i$ are immediate. The following section shows how these tools allow to prove an error propagation bound for our learning algorithm. 


\section{Error propagation for \textbf{LearnFCS}}

The proof of the error propagation bound requires four main ingredients.

\begin{itemize}
\item Since the regular realizations are not unique, it is important to identify one that can be reliably compared with the estimated realization parameters $\hat{e}$,  $\hat{\rho}$, $\hat{\mathbb{K}}$. 
The observable realization is not the ideal choice, essentially because the explicit matrices in the SVD of $\Omega$ may be very sensitive to errors. 
\item An error propagation bound can be obtained on by a {telescoping trick} and a recursive argument, based on {submultiplicativity} of completely bounded (cb) operator norms for maps between operator systems. A propagation bounds for unital completely positive maps as in~\cref{eq:dataproc} is not possible as we cannot guarantee that $\hat{\mathbb{K}}$ is unital and completely positive. In particular, \textbf{LearnFCS} does not take as input a valid operator system for the state.
\item The error propagation bound is expressed in terms of errors measured by the order norms of the operator system, and we need to convert them into usual error measures for the estimation of the density matrix of the marginals, in our case the Hilbert-Schmidt distance.
\end{itemize}

We comment on the solution in the following subsections.

\subsection{Empirical realization}

The idea of~\cite{Hsu2008} is to compare the empirical estimate with an exact {realization which depends on the empirical estimate} itself. This can be done in the quantum case too, where an exact {empirical realization} can be defined, with parameters
$\tilde{e},\tilde{\rho},\tilde{\mathbb{K}}$, defined as
\begin{align}
\tilde{e}&:={\hat{U}}^{\intercal}\Omega v_{\1}\in \mathbb{C}^{m},\label{tildee2intro}\\
\tilde{\rho}&:=w_{\1} \Omega(\hat{U}^{\intercal}\Omega)^+ \in \mathbb{C}^{m},\label{tilderhointro}\\
\tilde{\mathbb{K}}_{A}&:={\hat{U}}^{\intercal}(\Omega_{A})({\hat{U}}^{\intercal}\Omega)^+\in \mathbb{M}_{m,m}(\mathbb{C})\,.\label{tildeKAintro}
\end{align}

This is a valid realization as soon as $\hat{M}=\hat{U}^{\intercal}U$ is invertible, indeed
\begin{align}
\tilde{e}&=\hat{U}^{\intercal}U e,\label{tildeeeintro}\\
\tilde{\rho}&=\rho (\hat{U}^{\intercal}U)^{-1},\label{tilderhoointro}\\
\tilde{\mathbb{K}}_A&=\hat{U}^{\intercal}U\mathbb{K}_{A}(\hat{U}^{\intercal}U)^{-1}\label{tildekintro},
\end{align}

so that $\tilde{e},\tilde{\rho},\tilde{\mathbb{K}}$ and ${e},{\rho},{\mathbb{K}}$ realize the same state (we used that $UU^{\intercal}$ is the orthogonal projector on the image of $\Omega$ and $\Omega_{(\cdot)}$). To ensure that $\hat{U}^{\intercal}U$ is invertible, results from matrix perturbation theory can be used, requiring that the error in the estimation in Hilbert-Schmidt norm of $\omega_{2s}$ is small with respect to the sensitivity parameter $\eta$. In the same way, the truncation step in \textsf{LearnFCS} ensures that $\hat{e},\hat{\rho},\hat{\mathbb{K}}$ and $\tilde{e},\tilde{\rho},\tilde{\mathbb{K}}$ are also a row vector in $\mathbb{C}^m$, a column vector in $\mathbb{C}^m$, and a linear map from $\mathcal{L}(\mathcal{H}_A)$ to $\mathbb{M}_{m,m}(\mathbb{C})$, respectively. See~\cref{appendixpert} for the details.

The following error parameters can be introduced, using the order norms of the operator system:
\begin{align}\label{errorpardef1intro}
\delta_1&:=\|(\hat\rho-\tilde{\rho})\hat{M}\|_{e',*},\\
\delta_{\infty}&:=\|\hat{M}^{-1}(\tilde{e}-\hat{e})\|_{e'},\label{errorpardef2intro}\\
\Delta&:=\|\hat{M}^{-1}(\tilde{\mathbb{K}}-\hat{\mathbb{K}})\hat{M}\|_{\1\otimes e'\rightarrow e',cb}.\label{errorpardef3intro}
\end{align}

In case a quantum realization exists, an analogous expression can be used to consider errors according to the order inherited from the quotient operator system, adapting the definition of $\hat{M}$.

\subsection{Error propagation from completely bounded norms}

The error propagation bound takes the following form
\begin{align}
\|\hat\omega_{t}-\omega_{t}\|_1\leq (1+\delta_1)(1+\delta_{\infty})(1+\Delta)^t-1.
\end{align}

The full argument for the error propagation bounds is detailed in~\cref{subsec:errorprop}. The idea is to use triangle inequality to obtain
\begin{align}
\|\hat\omega_{t}-\omega_{t}\|_1&\leq\underbrace{\|(\hat\rho-\tilde\rho)\tilde{\mathbb{K}}^{t}\tilde{e}\|_1}_{(\operatorname{I})}+\underbrace{\|(\hat\rho-\tilde\rho)(\hat{\mathbb{K}}^{t}\hat{e}-\tilde{\mathbb{K}}^{t}\tilde{e})\|_1}_{(\operatorname{II})}+\underbrace{\|\tilde{\rho}(\tilde{\mathbb{K}}^{t}\tilde{e}-\hat{\mathbb{K}}^{t}\hat{e})\|_1}_{(\operatorname{III})}\label{eqtriangleineqintro}
\end{align}

and bound each term by induction. The crucial step is to bound
$\|\hat{M}^{-1}(\tilde{\mathbb{K}}^{t}-\hat{\mathbb{K}}^{t})\hat{M}\|_{e_t\rightarrow e, cb}$, which can be done as follows

\begin{align}
\|\hat{M}^{-1}(\hat{\mathbb{K}}^{t}-\tilde{\mathbb{K}}^{t})\hat{M}\|_{e_t\rightarrow e, cb}&\leq \|\hat{M}^{-1}\tilde{\mathbb{K}}\hat{M} \hat{M}^{-1}(\hat{\mathbb{K}}^{t-1}-\tilde{\mathbb{K}}^{t-1})\hat{M}\|_{e_t\rightarrow e, cb}\\
&+\|\hat{M}^{-1}(\hat{\mathbb{K}}-\tilde{\mathbb{K}})\hat{M} \hat{M}^{-1}\tilde{\mathbb{K}}^{t-1}\hat{M}\|_{e_t\rightarrow e, cb}\\
&+\|\hat{M}^{-1}(\hat{\mathbb{K}}-\tilde{\mathbb{K}})\hat{M} \hat{M}^{-1}(\hat{\mathbb{K}}^{t-1}-\tilde{\mathbb{K}}^{t-1})\hat{M}\|_{e_t\rightarrow e, cb}
\end{align}
by triangular inequalities, and, by submultiplicativity,

\begin{align}
\|\hat{M}^{-1}(\hat{\mathbb{K}}^{t}-\tilde{\mathbb{K}}^{t})\hat{M}\|_{e_t\rightarrow e, cb}&\leq (\Delta+\|\hat{M}^{-1}\tilde{\mathbb{K}}\hat{M}\|_{e_t\rightarrow e, cb})\|\hat{M}^{-1}(\hat{\mathbb{K}}^{t}-\tilde{\mathbb{K}}^{t})\hat{M}\|_{e_t\rightarrow e, cb}\\&+\Delta\|\hat{M}^{-1}\tilde{\mathbb{K}}\hat{M}\|_{e_1\rightarrow e, cb}^t. 
\end{align}

By induction, we have

\begin{align}
\|\hat{M}^{-1}(\hat{\mathbb{K}}^{t}-\tilde{\mathbb{K}}^{t})\hat{M}\|_{e_t\rightarrow e, cb}&\leq (\Delta+\|\hat{M}^{-1}\tilde{\mathbb{K}}\hat{M}\|_{e_1\rightarrow e, cb})^t-\|\hat{M}^{-1}\tilde{\mathbb{K}}\hat{M}\|_{e_1\rightarrow e, cb}^t. 
\end{align}

This expression grows exponentially in $t$, in general. However, in our case we can use that $\hat{M}^{-1}\tilde{\mathbb{K}}^{t}\hat{M}=\tilde{\mathbb{K}}^{t}$ is unital completely positive, therefore $\|\hat{M}^{-1}\tilde{\mathbb{K}}\hat{M}\|_{e_1\rightarrow e, cb}=1$ and the bound is $O(\epsilon)$ if $\Delta=O(t/\epsilon)$. If a quantum realization exists, this argument is also valid for the quotient operator system.

\subsection{Connection with tomography error}

What is left is to express $\delta_1,\delta_{\infty},\Delta$ in terms of $d_{HS}(\omega_s,\hat\omega_s)$, $d_{HS}(\omega_{2s},\hat\omega_{2s})$, $d_{HS}(\omega_{2s+1},\hat\omega_{2s+1})$, which can be done as follows.

First of all, the canonical Euclidean norm $\|\cdot\|$ on $\mathbb{C}^{m}$ let us to measure the errors also using 
$\|\hat\rho-\tilde{\rho}\|_2$, $\|\tilde{e}-\hat{e}\|_{2}$, $\|\hat{M}\|_{2\rightarrow 2}$, $\|\tilde{\mathbb{K}}-\hat{\mathbb{K}}\|_{2\rightarrow 2}$, where $\|A\|_{2\rightarrow 2}=\max_{\|v\|_2\leq 1}\|Xv\|_2$. Since $\tilde{\rho},\tilde{e},\tilde{K}$, $\hat{\rho},\hat{e},\hat{K}$, are obtained from the marginals and their estimates with linear algebra, it is possible to connect these error measures to the tomography errors in Hilbert-Schmidt distance using standard matrix perturbation theory, see~\cref{appendixpert}.

The completely bounded norm in the definition of $\Delta$ can be bounded with the norm without amplification, since the {cb norm} of a finite rank map 
{is} {bounded by its norm times the rank}. Alternatively, for states with quantum models we 
use that the cb norm of a linear map to $d_{\mathcal{B}}\times d_{\mathcal{B}}$ matrices 
{equals} the norm of the $d_{\mathcal{B}}$-th amplification. 
These changes give rise to the alternative bound.

The final piece is to relate the Euclidean norm to the order norm. To express the error parameters in terms of the Hilbert-Schmidt error of the marginal estimates we introduce a suitable {{weighted} HS norm} $\|\cdot\|_\Omega$  constructed from the map $\Omega$ and relate it to the {order norm} $\|\cdot\|_e$. We 
obtain {the bounds}
$\sigma_{m}(\Omega)\|\cdot\|_{e_n}\leq \|\cdot\|_{n,\Omega} \leq \sqrt{n}\|\cdot\|_{e_n}$ for the {$n$-th amplification norms}, whose proof 
{might} be of independent interest, as it applies to general bipartite states. This passage is crucial but technical, and the details can be found in Appendix~\ref{sec.opsysfcs}. To get the final result we also use inequalities between Schatten norms 
and an adaptation of the matrix perturbation analysis of~\cite{Hsu2008,Siddiqi2009,balle2013learning}, with key lemmas collected in Appendix~\ref{appendixpert}.

\section{Conclusions}

We established rigorous bounds for the problem of learning a realization of finitely correlated state from data, showing that under certain conditions there is an efficient solution. The performance of the learning algorithm depends critically on the parameter $\eta$, and we leave as an open question to understand how to interpret this parameter in terms of operationally motivated physical quantities.
Another important question, that we did not directly address, is how to learn and certify a good quantum model for the unknown state. This is important if we wish to prepare the unknown state with sequential operations. Under the promise that a good quantum model exists with some bounded realization dimension, in the translation invariant case, a possible approach would be to first learn an estimate of the realization with our algorithm, and then optimize over quantum models such that their marginals are close to the estimated realization. If the size of these marginals is taken large enough, such that they completely determine the state, one can also control the error on larger marginals. For the non-translation invariant case, ensuring that the marginals completely determine the state is a condition which is more difficult to implement.
Finally, the techniques we introduce could be applied to other graphical models with 1D structure, and to other setting such as online learning. It is still an important open question to understand if efficient learning is possible for tensor networks on higher dimensional lattices.

\section*{Acknowledgements}
The authors thank Ivan Todorov for discussions on the operator systems associated to finitely correlated states,  Matteo Rosati and Farzin Salek for early conversations on the reconstruction of quasi-realizations of stochastic processes, Samuel Scalet for discussions on learning Gibbs states, Marco Tomamichel for comments on the manuscript. MF thanks Steven Flammia, Yunchao Liu, Angus Lowe and Antonio Anna Mele for discussions on learning matrix product states, and Roy Araiza for discussions on operator systems.

JL is supported by the National Research Foundation, Prime Minister’s Office, Singapore and the Ministry of Education, Singapore under the Research Centres of Excellence programme and the Quantum Engineering Programme grant NRF2021-QEP2-02-P05.
C.R. has been supported by ANR project QTraj (ANR-20-CE40-0024-01) of the French
National Research Agency (ANR).
MF and AW were supported by the Baidu Co.~Ltd.~collaborative project ``Learning of Quantum Hidden Markov Models''. MF, NG and AW are supported by the Spanish MCIN (project PID2022-141283NB-I00) with the support of FEDER funds, by the Spanish MCIN with funding from European Union NextGenerationEU (grant PRTR-C17.I1) and the Generalitat de Catalunya, as well as the Ministry of Economic Affairs and Digital Transformation of the Spanish Government through the QUANTUM ENIA ``Quantum Spain'' project with funds from the European Union through the Recovery, Transformation and Resilience Plan - NextGenerationEU within the framework of the "Digital Spain 2026 Agenda". MF is also supported by a Juan de la Cierva Formaci\'on fellowship (Spanish MCIN project FJC2021-047404-I), with funding from MCIN/AEI/10.13039/501100011033 and European Union NextGenerationEU/PRTR. AW acknowledges furthermore support by the European Commission QuantERA grant ExTRaQT (Spanish MCIN project PCI2022-132965), by the Alexander von Humboldt Foundation, and by the Institute for Advanced Study of the Technical University Munich.

\bibliography{learning_hidden_markov.bib, biblio.bib} 
\bibliographystyle{unsrt}

\begin{appendix}
 \section*{Appendix}



The following appendices are structured as follows.

\begin{itemize}
\item In \cref{secprel} we set up the notation and present the objects that we use in the analysis, such as several particular realizations of a finitely correlated state and how they are connected to each other. 
\item In~\cref{sec:opsysqr} we discuss the operator systems associated with realizations. \item In~\cref{sec.statereconstruction} we show the main result for general states, and in~\cref{sec.statereconstructionc} we treat the special case of $C^*$-finitely correlated states.
\item In~\cref{secnontrans} we extend the arguments to the non-translation invariant case, for states on a finite chain. 
\item In~\cref{sec.learningclose} we consider how the algorithm can be used to learn states which are only close to finitely correlated, including Gibbs states. \item In~\cref{sec.numer} we report some numerical experiments.
\item In~\cref{appendixpert} we collect useful facts on matrix perturbation theory.
\end{itemize}

\section{Preliminaries}\label{secprel}

\subsection{Notation}
In the following we think of states as non-negative linear functionals from a $C^*$-algebra to $\mathbb{C}$. For finite dimensional systems, we can simply think of a density matrix $\rho\in\mathbb{M}_{n}$ as a functional on $\mathbb{M}_{n}$, $\rho(X)=\Tr[\rho X]$. Let us fix the notations for states on an infinite chain following~\cite{Fannes1992}.
We consider states on an infinite chain with sites labelled by $n\in\mathbb{Z}$, where to each site $n$ a finite dimensional $C^*$-algebra $\mathcal A_n \equiv \A$ is associated. Without loss of generality, $\mathcal A$ can be embedded into the algebra of $d_{\A}\times d_{\A}$ matrices $\mathbb{M}_{d_{\A}}$. For each finite subset $X$ of $\mathbb{Z}$ we can define an algebra $\mathcal{A}_X:=\bigotimes_{n\in X}\mathcal A_{n}$. The algebra of infinite subsets is defined through the inductive limit, with the identifications $\mathcal{A}_{X_2} \equiv \mathcal{A}_{X_1}\otimes \1_{X_2\setminus X_1}$ for $X_1\subseteq X_2$ finite subsets. The chain algebra $\mathcal{A}_{\mathbb{Z}}$ is defined in this way, for example, as the inductive limit for sets $[-n,n]:=\{-n,\dots,n\}$. The right chain algebra $\mathcal{A}_{R}$ is defined as the inductive limit of $\mathcal{A}_{[1,n]}$, and the left chain algebra $\mathcal{A}_{L}$ is defined as the inductive limit of $\mathcal{A}_{[-n,0]}$. The translation operators $\alpha_r$ act on the algebra by sending $\mathcal{A}_{X}$ into $\mathcal{A}_{X+r}$. We denote the set of translation invariant states by $\mathcal{T}$. 

For linear maps $L:V_1\rightarrow V_2$ and $M:V_2 \rightarrow V_3$ we write the composition of $M$ and $L$ simply as $ML$. The identity map on a vector space $V$ is denoted as $\id_{V}$, the identity element of a $C^*$-algebra $\mathcal{A}$ is denoted as $\1_{\mathcal{A}}$. Given a linear map $\mathbb{E}:\mathcal{A}\ni A \mapsto \mathbb{E}_A\in \mathrm{Hom}(V)$, we use the notation $\mathbb{E}^t: \mathcal{A}^{\otimes t}\ni A \mapsto \mathbb{E}_A\in \mathrm{Hom}(V)$ for the linear map that acts on tensor product vectors as $\mathbb{E}^t_{X_1\otimes\cdots\otimes X_t}=\mathbb{E}_{X_1}\cdots\mathbb{E}_{X_t}$. Similarly, given a linear map $\E:\mathcal{A}\otimes V \rightarrow V$, we use the notation $\E^t: \mathcal{A}^{\otimes t}\otimes V \rightarrow V$ for the linear map that acts on product vectors as $\mathbb{E}^t(X_1\otimes\cdots\otimes X_t\otimes x)=\mathbb{E}_{X_1}\cdots\mathbb{E}_{X_t}(x)$.

Additionally, we use the following notation for linear maps. We denote by $\|T\|_p=(\Tr[(T^\dagger T)^{{p}/{2}}])^{1/p}$ the Schatten $p$-norm for $T:H_1\rightarrow H_2$ being a linear map between finite dimensional Hilbert spaces $H_1,\,H_2$, where $T^\dagger$ is the adjoint map (we also use the notation $T^{\intercal}$ if we deal with real maps). The singular values of $T$ are denoted and ordered as $\sigma_1(T)\ge \sigma_2(T)\ge \dots\ge \sigma_r(T)$, so that $\|T\|_p=\left(\sum_{i=1}^{r}\sigma_{i}(T)^p\right)^{1/p}$. We denote by $T^{+}$ the Moore-Penrose pseudoinverse of $T$.  For a map $L: B_1\rightarrow B_2$, where $B_1$ and $B_2$ are spaces of finite dimensional linear maps equipped with Schatten norms, we have $\|L\|_{p\rightarrow p}:=\sup_{X\in B_1, \|X\|_p\leq 1} \|L(X)\|_p$, and $\|L\|_{p\rightarrow p, cb}:=\sup_{d\in \mathbb{N}}\sup_{X\in \mathbb{M}_{d}\otimes B_1,\|X\|_p\leq 1} \|(\id_{\mathbb{M}_d}\otimes L)(X)\|_p$. 
By slight abuse of notations, for $p\in[1,\infty]$, we will also denote by $\|\cdot\|_p$ the usual $\ell_p$ norm on $\mathbb{R}^m$, as well as the $p$-norm of linear functionals: given a functional $\varphi$ on a (matrix) space $\mathcal{C}$, $\|\varphi\|_p:=\sup_{\|X\|_{p'}\le 1}|\varphi(X)|$, where $1/p'+1/p=1$. For $p=2$, this is also equal to $\|\varphi\|_2=(\sum_i|\varphi(X_i)|^2)^{1/2}$ for some orthonormal basis $\{X_i\}$ of $\mathcal{C}$.

\subsection{Finitely correlated states}

A subclass of translation invariant states are \textit{finitely correlated states}, which admit the following description.

\begin{definition}
A \emp{(linear) realization} of a translation invariant state $\omega$ is
a quadruple $(V,e,\mathbb{E},\rho)$, where $V$ is a finite-dimensional vector space serving as a memory, $\mathbb{E}$ a linear map $\mathbb{E}:A\in \mathcal{A} \rightarrow \mathbb{E}_A\in \mathrm{Hom}(V)$, $e$ an element of $ V$, and $\rho$ a linear functional in $ V^*$, such that 
\begin{equation}
\omega(A_1\otimes \cdots\otimes A_n)=\rho\mathbb{E}_{A_1} \cdots \mathbb{E}_{A_n}(e),
\end{equation}
and
\begin{equation}
\rho \mathbb{E}_{\1}=\rho \qquad  \mathbb{E}_{\1}(e)=e.
\end{equation}
A translation invariant state $\omega$ admitting a realization is called a \textbf{finitely correlated state}.
\end{definition}

The following characterization holds, from Proposition 2.1 of~\cite{Fannes1992}.

\begin{proposition}\label{prop} Let $\mathcal A$ be a $C^*$-algebra with unit, and let $\omega$ be a translation invariant state on the chain algebra $\mathcal{A}_{\mathbb{Z}}$. Then the following are equivalent:
\begin{itemize}
\item The set of functionals $\Psi_X:\mathcal{A}_L\rightarrow\mathbb{C}$ of the form
\begin{equation}
\Psi_X(A_{-n}\otimes\cdots\otimes A_0)=\omega(A_{-n}\otimes \cdots\otimes A_0\otimes X)
\end{equation}
with $X\in\mathcal{A}_R$ and $n\in\mathbb N$ generates a finite-dimensional linear subspace in the dual of $\mathcal{A}_L$, of dimension $m$, say.
\item There exists a realization $(V,e,\mathbb{E},\rho)$ of $\omega$.
\end{itemize}
Moreover, the minimal dimension of $V$ is $m$. 
\end{proposition}

Furthermore, also as part of \cite[Prop.~2.1]{Fannes1992} we have
\begin{proposition}\label{propsim} Realizations of minimal dimension, called \emp{regular} realizations, are determined up to linear isomorphisms: if $(V,e,\mathbb{E},\rho)$ and $(V',e',\mathbb{E}',\rho')$ are two regular realizations of $\omega$, then $V$ and $V'$ are isomorphic via an invertible linear map $M:V\rightarrow V'$, such that $e'=Me$, $\rho'=\rho  M^{-1}$, $\mathbb{E}'_A=M \mathbb{E}_{A} M^{-1}$.
\end{proposition}

In the proof of the above statements, it is used that a realization of minimal dimension can be obtained as a \emp{quotient realization} of any realization.
Note, that if $\tilde W\subset V^*$ is a set of linear functionals on $V$, we denote the subspace on which they all vanish simultaneously by $\tilde W^\perp$ (that is, $\tilde W^\perp$ is just the intersection of their kernels). Then we hsve the following.

\begin{proposition}\label{thm:quotient}
Let $(V,e,\mathbb{E},\rho)$ be a realization of a translation invariant state $\omega$, whose minimal realizations are of dimension $m$. Define $W=\{\mathbb{E}^k_Ae\ |\  A\in\mathcal{A}_{[1,k]},k\in\mathbb{N}_{>0}\}$, and $\tilde{W}=\{\rho\mathbb{E}^k_A\ |\  A\in\mathcal{A}_{[1-k,0]},k\in\mathbb{N}_{>0}\}$, set $V':=W/({W\cap\tilde{W}^{\perp}})$ and $L:=W\rightarrow V'$ the canonical projection. Then $\dim{V'}=m$, and there exists a (thus regular) realization $(V',e',\mathbb{K}',\rho')$ such that, for the restrictions $\E_A |_{W}$ and $\rho|_{W}$ to $W$ we have:
\begin{align}
\mathbb{K}'_A L &= L \mathbb{E}_A|_{W}, \label{eqquot1} \\
\rho'L\label{eqquot2} &= \rho|_{W}, \\
 e'&= L e.
\end{align}

\end{proposition}

\subsubsection{Observable regular realization}\label{sec.observablereal}

The proof of Proposition~\ref{prop} gives also a way to construct a realization from the knowledge of a marginal of the state on a finite subset of the chain.
We call this realization \emph{observable regular realization}, as it is expressed only in terms of quantities which can be estimated from experiments.

Let us select two finite-dimensional unital subalgebras $\mathcal{C}_R\subseteq \mathcal{A}_R$ and $\mathcal{C}_L\subseteq \mathcal{A}_L$ such that the vector space generated by functionals $\Psi_X:\mathcal{C}_{L}\rightarrow\mathbb{C}$ of the form 
\begin{equation}\label{PsiX}
\Psi_X(Y):=\omega(Y\otimes X),
\end{equation}
with $X\in\mathcal{C}_{R}$, has dimension $m$. We call this vector space $V$, which can be thought of as a subspace of $
{\mathcal{A}_{\Lambda}^*\cong \mathcal{A}_\Lambda}$, for some finite subset $\Lambda$ of the chain, say $\Lambda = [1, t^*]$ for some $t^*\le m$ (see \cref{appendix.tstar}), and define $L$ to be a projection from $\mathcal{A}_\Lambda^*$ to $V$. 

We can thus define the map $\mathbb{E}:\mathcal{A}\ni A \mapsto \mathbb{E}_A\in \mathrm{Hom}(\mathcal{A}^{*}_{\Lambda})$ as
\begin{equation}
\mathbb{E}_A(\Psi_X):=\Psi_{A\otimes X},
\end{equation}
and by setting $\mathbb{E}_A$ to be $0$ on the orthogonal complement of $V$.  We define the functional $\tau\in\mathcal{A}_\Lambda^{**}$ as 
\begin{equation}\label{eqidstarstar}
\tau(Y)=Y(\1),\,\, \forall\,Y\in\cA_{\Lambda}^*\,.
\end{equation}
Note that identifying $\cA_{\Lambda}^*$ with $\cA_\Lambda$ through the Hilbert-Schmidt inner product, $\tau$ acts as the trace functional.

Define the maps $\Omega:\C_{R}\rightarrow\C_{L}^*$ and $\Omega_{(\cdot)}: \cA\otimes \C_{R}\rightarrow\C_{L}^*$ through
\begin{align}\label{def.OmegaOmega()}
\Omega(X)[Y]&:=\omega(Y\otimes X)\\
\Omega_{A}(X)[Y]&:=\omega(Y\otimes A\otimes X),
\end{align}
where $\Omega_{(\cdot)}$ is extended linearly to all of $\cA\otimes \C_{R}$.
Fix a self-adjoint basis $\{X_i\}_{i=1}^{m_1}$ of $\mathcal{C}_R$, such that $\Tr[X_iX_j]=\delta_{ij}$, and a self-adjoint basis $\{Y_i\}_{i=1}^{m_2}$ of $\mathcal{C}_L$, such that $\Tr[Y_iY_j]=\delta_{ij}$.
By singular value decomposition, we can write
\begin{align}\label{eq:omega_svd}
\Omega=U D O,
\end{align}
with $D$ diagonal in $\mathbb{R}_{\geq 0}^{m\times m}$, $U:\mathbb{C}^{m}\rightarrow \C_{L}^*$ and $O:\C_{R}\rightarrow \mathbb{C}^{m}$ partial isometries, with real entries in the chosen basis and satisfying $U^\intercal U = OO^\intercal = \1_{\mathbb C^m}$.

\begin{proposition}[Observable regular realization]\label{prop.observablequasireal}
$(\mathbb{C}^m, e,\mathbb{K},\rho)$ is a regular realization, where
 \begin{align}
e&:=U^{\intercal} \Omega(\1),\label{eqe1}\\
\rho&:=\tau  U =\tau \Omega (U^{\intercal}\Omega)^{+},\label{eqrho1}\\
\mathbb{K}_{A}&:=U^{\intercal}\Omega_{A} (U^{\intercal}\Omega)^+.\label{eqK1}
\end{align}

\end{proposition}

\begin{proof}
We have that
\begin{align}
{\mathbb{E}_A \Omega(X)=\mathbb{E}_A \Psi_{X}=\Psi_{A\otimes X}=\Omega_{A}}(X)\,,
\end{align}
and therefore
\begin{align}
\mathbb{E}_A U=\Omega_{A}(D O)^{+}=\Omega_{A} (U^{\intercal} \Omega)^{+}\,,
\end{align}
{where we have also used that $U^\intercal  U=\1_{\mathbb{C}^m}$ so that $U^\intercal  \Omega=D O$.} 
Indeed we have
\begin{align}
\mathbb{K}_A U^\intercal &=U^{\intercal} \Omega_{A} (U^{\intercal} \Omega)^+ U^{\intercal}= U^{\intercal} \mathbb{E}_A  U U^{\intercal},
\end{align}
which lets us conclude $\mathbb{K}_{A} U^{\intercal}=U^{\intercal} \mathbb{E}_A$, since $P_U = UU^{\intercal}$ is the projector on the space identified with $V$ in $\mathcal{A}^*_{\Lambda}$, and both $U^{\intercal}$ and $\mathbb{E}_A$ are zero on the orthogonal complement of $V$ (so in particular $\mathbb E_A P_U = \mathbb E_A$).
Hence, we have proved that
\begin{align*}
\rho\mathbb{K}_{A_n}\dots \mathbb{K}_{A_1}e=\tau P_U
\mathbb{E}_{A_n}P_U\dots \mathbb{E}_{A_1}P_U\Psi_{\1}=\tau
\mathbb{E}_{A_n}\dots \mathbb{E}_{A_1}\Psi_{\1}\,,
\end{align*}
and therefore $(\mathbb{C}^m,e,\mathbb{K},\rho)$ is a regular realization of $\omega$. 
\end{proof}

\subsubsection{Guarantees on marginal size for translation invariant states}\label{appendix.tstar}
We observe that if the rank $m$ of the finitely correlated state viewed as a bilinear form is known, the construction in the previous subsection can be made completely explicit. First, we have the following lemma.

\begin{lemma}\label{lem:finitestop}
Define $V_{i}$ the vector space generated by functionals $\Phi_X:\mathcal{A}_{R}\rightarrow\mathbb{C}$ of the form 
\begin{equation}
\Phi_X(A_1\otimes\cdots\otimes A_{n}):=\omega(X\otimes A_1\otimes \cdots\otimes A_{n}),
\end{equation}
with $X\in\mathcal{A}_{[-i+1,0]}$ and $n\in\mathbb N$.  If $\omega$ has a minimal realization of dimension $m$, then there exists $i\le m$ such that $V_j = V_i = \Phi(\A_L)$ for each $j\ge i$.
In particular, this holds for $V_{m}$.
\end{lemma}

\begin{proof}
Let $i$ be such that $\dim V_i = \dim V_{i+1}$, i.e.\ $V_i = V_{i+1}$.
Let $\Phi_{Y_k}$, $k=1,\dots, \dim V_i$, be a basis of $V_i$, so $Y_k\in \A_{[-i+1,0]}$.
If $X = \sum_j X_j\otimes A_j\in \A_{[-i-1,-1]}\otimes\A = \A_{[-i-1,0]}$, i.e.\ $\Phi_X\in V_{i+2}$, then for $A\in\A_R$ using translation invariance of the FCS $\omega$:
$$\Phi_X(A) = \sum_j \Phi_{X_j}(A_j\otimes A) = \sum_{j,k}c_{jk}\Phi_{Y_k}(A_j\otimes A) = \sum_{j,k} c_{jk}\Phi_{Y_k\otimes A_j}(A) = \Phi_Y(A)$$
where $Y = \sum_{j,k} c_{jk} Y_k\otimes A_j \in \A_{[-i,0]}$.
Therefore, $V_{i+2} = V_{i+1}$ and by induction $V_i = V_j = \Phi(\A_L)$ for all $j\ge i$.
Because of $\dim(\Phi(\A_L)) = m$ by assumption, we have $\dim V_i = m$ and thus $i\le m$ since $\dim V_j$ has to increase monotonically in $j$.
\end{proof}
Let $\Psi:\A_R\to\A_{L}^*$ be the map such that $\Psi_X(Y) = \Phi_Y(X)$ where $X\in\A_{R}, Y\in\A_{L}$. 
We have $\Psi(\A_R)\cong \Phi(\A_{L})^*$ since both are finite-dimensional of the same dimension and since $\Psi_X(\Phi_Y) = \Phi_Y(X)$ is well-defined.
By the above Lemma we can replace $\A_L$ by a finite subchain $\Psi(\A_R)\cong \Phi(\A_{[-m+1,0]})^*$.
The argument in the proof of the Lemma can also be applied to the analogous spaces $W_i$ for $\Psi$ showing that $\Psi(\A_R)\cong \Psi(\A_{[1,m]})$.
Therefore, it suffices to parametrize $\Psi_X$ by $X\in \A_{[1,m]}$ and restrict inputs to $\A_{[-m+1, 0]}$ so one has:
\begin{corollary}\label{thm:Vij}
Define $V_{i,j}$ the vector space generated by functionals $\Psi_X:\mathcal{A}_{[-j+1, 0]}\rightarrow\mathbb{C}$ of the form 
\begin{equation}
\Psi_X(A):=\omega(X\otimes A),
\end{equation}
with $X\in\mathcal{A}_{[1,i]}$.
If $\omega$ has a minimal realization of dimension $m$, $V_{m,m}$ 
has dimension $m$.
The same result holds for the analogously defined spaces $W_{i,j}$ of functionals on the left-hand side of the chain.
\end{corollary}

\subsubsection{Quotient realizations}
Let $(V, u, \mathbb E, \sigma)$ be any realization and $(V', e',\mathbb{K}',\rho')$ its quotient realization.
Since $(V', e',\mathbb{K}',\rho')$ and the observable realization $(\mathbb C^{m}, e,\mathbb{K},\rho)$ constructed in Proposition \ref{prop.observablequasireal} are both regular, there is a unique invertible linear transformation $M:V'\rightarrow \mathbb{C}^m$ such that
\begin{align}
e&=Me',\\
{\mathbb{K}}_A&=M\mathbb{K}'_{A}M^{-1},\\
{\rho}&=\rho' M^{-1}.
\end{align}

A specific map $M$ with desirable properties can be obtained explicitly as follows:
Recall that $\C_{R}\subseteq \mathcal{A}_{[1,t^*]}$ and $\C_{L}\subseteq \mathcal{A}_{[1-t^*,0]}$ (there is always such $t^*$, in fact we can take $t^*\leq m$, see Appendix \ref{appendix.tstar}).
We define the map  $\mathcal{F}:V\rightarrow \mathcal{A}_{[1-t^*,0]}^*$ as
\begin{equation}
\mathcal{F}(Z)[A]:=\sigma \mathbb{E}^{t^*}(A\otimes Z)\,.
\end{equation}

By definition, we have, for any $A\in\C_{R}$ and $A'\in\C_{L}$,
\begin{equation}\label{eqFEO}
 \mathcal{F} \E^{t^*}_A(u)[A']=\sigma\mathbb{E}^{t^*}_{A'} \E^{t^*}_A(u)=\omega(A'\otimes A)=\Omega(A)[A'],
\end{equation}
so $\mathcal{F}\E^{t^*}_A(u)=\Omega(A)$.
Similarly, for all $A''\in\mathcal{A}$,$A\in\C_{R}$,
\begin{equation}\label{eq.OmegaAAF}
\mathcal{F}\E^{t^*+1}_{A''\otimes A}(u)=\Omega_{A''}(A)\,.
\end{equation}
With this:
\begin{lemma}\label{thm:quotientM}
    It holds that $W\cap\ker(U^\intercal\mathcal F) = W\cap\tilde W^\perp$.
    In particular, $U^\intercal \mathcal F|_W$ factorizes through the quotient realization as $U^\intercal \mathcal F|_W = ML$ with an invertible map $M:V'\to \mathbb C^m$ and $L:W\to V'$ being the projection.
\end{lemma}
\begin{proof}
    Note, that $\mathcal F\mathbb E_A(u) = \Psi_A$ and thus $\ker(U^\intercal) = \mathcal F(W)^\perp$.
    Now, $Z\in\ker(\mathcal F)$ if and only if for each $A$ we have
    $$0 = \mathcal F(Z)[A] = \sigma\mathbb E^{t^*}_A(Z)$$
    which, by definition, holds if and only if $Z \in \tilde W^\perp$.
\end{proof}
We arrive at
\begin{proposition}\label{propquotientc} 
It holds that
\begin{align}
{e}&={U}^{\intercal}\F(u),\\
{\rho}{U}^{\intercal}\F|_W&=\sigma|_W\\
{\mathbb{K}}_A{U}^{\intercal}\F|_W&={U}^{\intercal}\F\E_A|_W,
\end{align}
In particular, $(\mathbb{C}^m,e,\mathbb{K},\rho)$ is similar to the quotient realization of $(V', e',\E,\rho')$ through the invertible map $M:V'\rightarrow \mathbb{C}^m$ defined by $ML =U^{\intercal}\mathcal{F}|_W$ (see \cref{thm:quotient}).
\end{proposition}
\begin{proof}
    It holds that $\F(u) = \Omega(\1)$ so $U^\intercal\F(u) = e$ is immediate.
    Next, note that by definition $\rho U^\intercal\F = \tau\F$.
    Hence, if $w=\E_A(u)\in W$, $A\in \ker(\Omega)^{\perp
}$, then
    $$\rho U^\intercal\F(w) = \tau\F(w) = \Omega(A)[\1] = \omega(A) = \sigma(w).$$
    Furthermore,
    $$U^\intercal \F \E_{A'} w = U^\intercal\Omega_{A'}(A) = \KK_{A'}U^\intercal\Omega(A) = \KK_{A'}U^\intercal\F(w).$$
    Writing $U^\intercal\F = ML$ by \cref{thm:quotientM} the proof is complete.
\end{proof}

\subsubsection{$C^*$-realizations}
\label{secquotientcp}

A special kind of finitely correlated states are those that can be generated by consecutive applications of a quantum operation on a genuinely quantum memory system. They admit $C^*$-realizations, which we now define, and are also known as $C^*$-finitely correlated states.

\begin{definition}[$C^*$-realization]
\label{def:C*real}
A \textbf{$C^*$-realization} of a translation invariant state $\omega$ is a realization $(\B,\1_{\B},\cE,\rho_0)$ of $\omega$, where $\B$ is a $C^*$-algebra, $\1_{\B}$ is the identity of $\B$, $\cE:\A\otimes\B\rightarrow \B$ is completely positive and unital, and $\rho_0$ is a positive functional on $\B$, such that $\rho_0(\1)=1$. Whenever $\omega$ admits a $C^*$-realization it is called a \textbf{$C^*$-finitely correlated state}.  
\end{definition}

\begin{figure}
    \centering
\includegraphics[scale=0.27]{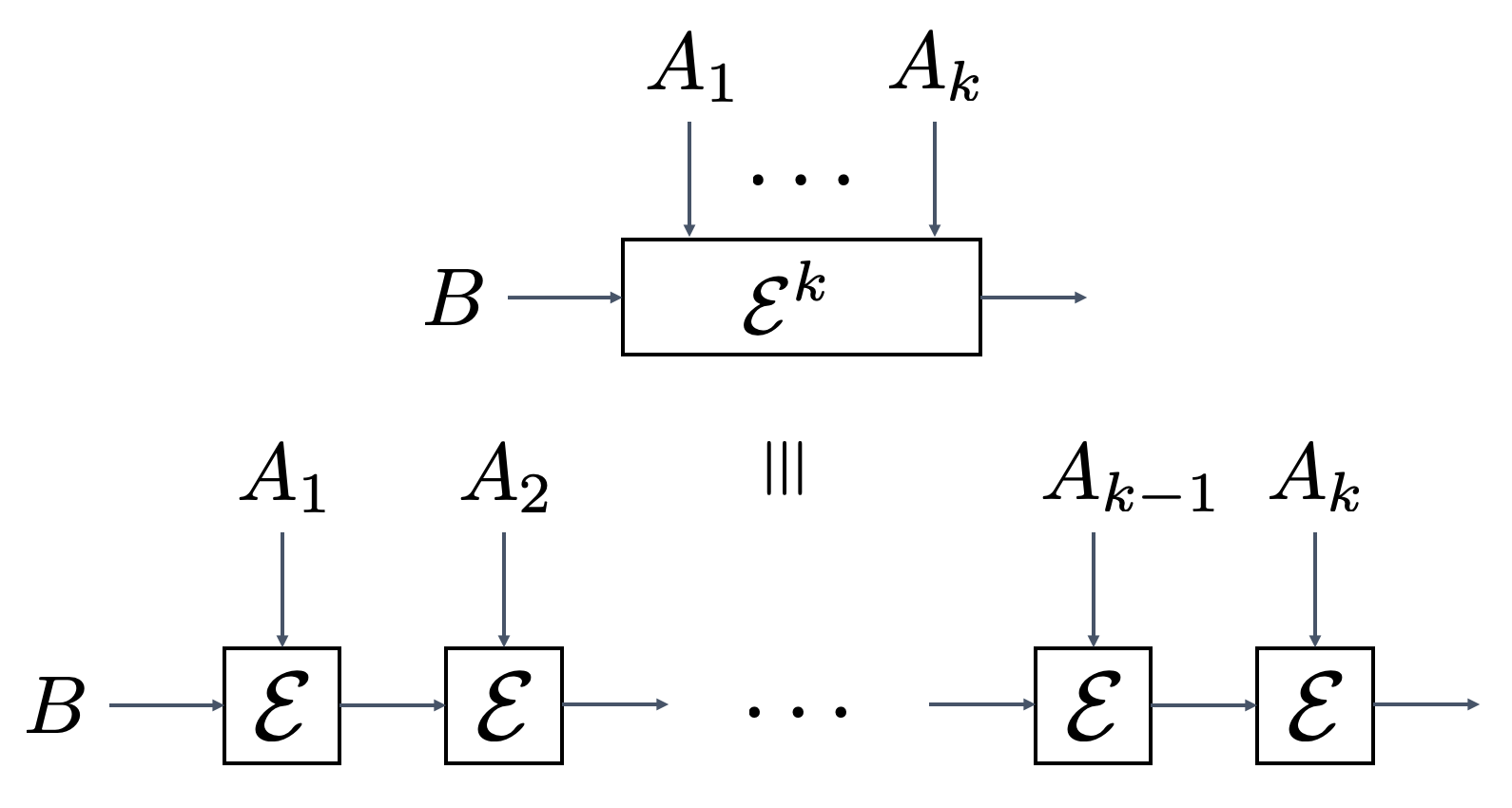}
    \caption{Representation of a $C^*$-finitely correlated state.}\label{Fig1}
\end{figure}

We use the symbol $\mathcal{E}$ instead of $\mathbb{E}$ when dealing with $C^*$ realizations to distinguish them among realizations.
Note, that whenever we find a realization with memory system being a $C^*$-algebra and generating map $\E$ being completely positive in the natural $C^*$ order, we can always find a realization satisfying the same that is of the form described in \cref{def:C*real}, even with $\rho_0$ being of full support, see \cite[Lem.~2.5]{Fannes1992}.

As before we will also use the notation $\mathcal{E}_A(B) = \cE(A\otimes B)$. We can apply the construction of the quotient realization \cref{thm:quotient}, that is we set
\begin{align}
\W=\{\mathcal{E}^{t^*}_A(\1_{\mathcal B})\ |\  A\in\mathcal{C}_{R}\}\,,\qquad \text{ and } \qquad \tilde{\W}=\{\rho_0\mathcal{E}^{t^*}_A\ |\ A\in\mathcal{C}_{L}\}
\end{align} 
 where we recall the notation $\mathcal{E}^k(A_1\otimes\dots\otimes A_k\otimes B)=\mathcal{E}_{A_1}\mathcal{E}_{A_2}\cdots\mathcal{E}_{A_k}(B)$ (see Figure \ref{Fig1}).
We also set 
\begin{align}
\V:=\W/({\W\cap\tilde{\W}^{\perp}})\subseteq \W
\end{align}
and define the projection 
$L:\W\rightarrow\V$.

Other than in the general case, the Hilbert Schmidt inner product lets us identify the quotient $\V$ with the orthogonal complement $\V'$ of $\W\cap\tilde\W^\perp$ in $\W$ and denote the unique map doing so as $R:\mathcal{\V\rightarrow \V'}$.
It satisfies $RL = \Pi_{\mathcal V'}$ where $\Pi_{\V'}$ is the orthogonal projector onto $\mathcal V'$.

\cref{thm:quotientM} then takes the following form:
\begin{lemma}\label{finvert}
    $U^\intercal\F$ is invertible on $\V'$ with inverse $(U^{\intercal}\F|_{\V'})^{-1} = RM^{-1}$
\end{lemma}

\subsubsection{Empirical realizations}
It will be useful to also consider the following realizations, obtained by a similarity transformation of the observable realization. In the following $\hat{U}$ can be thought as a truncated left orthogonal map in the SVD of $\hat{\Omega}$, which is an estimate of $\Omega$ (cf.\  \cref{sec.statereconstruction}). While the realization constructed from $\hat{\Omega}$ will not be an exact realization of $\omega$, the \emph{empirical realization} constructed here is, and will be easier to compare with the estimate than the observable realization.

\begin{proposition}[Empirical realizations]\label{propempquasi}
Let $\omega$ be a translation invariant state with observable realization $(\mathbb{C}^m,e,\mathbb{K},\rho)$. For any map $\hat U:\mathbb{C}^{m}\rightarrow \C_{L}^*$ (real in the basis of self-adjoint elements) such that $\hat{U}^{\intercal}U$ is invertible, the quadruple $(\mathbb C^{m}, \tilde{e},\tilde{\mathbb{K}},\tilde{\rho})$ given by

\begin{align}
\tilde{e}&={\hat{U}}^{\intercal}\Omega(\1),\label{tildee2}\\
\tilde{\rho}&=\tau \Omega(\hat{U}^{\intercal}\Omega)^+,\label{tilderho}\\
\tilde{\mathbb{K}}_{A}&={\hat{U}}^{\intercal}(\Omega_{A})({\hat{U}}^{\intercal}\Omega)^+\,.\label{tildeKA}
\end{align}

is a realization of $\omega$, in fact

\begin{align}
\tilde{e}&=\hat{U}^{\intercal}U e,\label{tildeee}\\
\tilde{\rho}&=\rho (\hat{U}^{\intercal}U)^{-1},\label{tilderhoo}\\
\tilde{\mathbb{K}}_A&=\hat{U}^{\intercal}U\mathbb{K}_{A}(\hat{U}^{\intercal}U)^{-1}.\label{tildek}
\end{align}

Any such realization is called \textbf{empirical realization}.
\end{proposition}

\begin{proof}
The proof is by explicit substitution, noting that $({\hat{U}}^{\intercal}\Omega)^+=({{U}}^{\intercal}\Omega)^+(\hat{U}^{\intercal}U)^{-1}$ and that $UU^{\intercal}\Omega_A=\Omega_A$ because the images of $\Omega$ and $\Omega_{(\cdot)}$ coincide.
\end{proof}

The following proposition links empirical realizations with quotient realizations, and in particular with quotient realizations of completely positive ones. 
Let us define the map
\begin{align}\label{hatM}
\hat{M}:=\hat{U}^\intercal UM\,,
\end{align}
where $M$ is the linear map connecting a quotient realization with the observable realization, as in Proposition~\ref{propquotientc}. The following proposition easily follows.

\begin{proposition}[Empirical realizations and quotient realization]
For an empirical realization as in Proposition \ref{propempquasi}, and a quotient realization as in Proposition \ref{propquotientc}, we have

\begin{align}
\tilde{e}&=\hat{M} e',\label{tildee}\\
\tilde{\mathbb{K}}_A&=\hat{U}^{\intercal}U\mathbb{K}_{A}(\hat{U}^{\intercal}U)^{-1}
=\hat{M}\mathbb{K}'_{A}\hat{M}^{-1}
,\label{tildeK}\\
\tilde{\rho}&=\rho' \hat{M}^{-1}\label{tilderh}.
\end{align}
so that
\begin{align}
\tilde{e}&=\hat{M} L u,\label{tildee'}\\
\tilde{\mathbb{K}}_A\hat{M} L&=\hat{M}\mathbb{K}'_{A} L=\hat{M}L \E_{A}|_{W},\label{KAtilde}\\
\tilde{\rho}\hat{M}L&=\sigma|_{W},\label{eqrhorel}
\end{align}
and, hence,
\begin{align}\label{eqKerelation}
\tilde{\mathbb{K}}^t_{X_1\otimes\cdots\otimes X_t}\tilde{e}=\tilde{\mathbb{K}}_{X_1}\cdots\tilde{\mathbb{K}}_{X_t}\hat MLu=\hat{M}
L\mathbb{E}^t_{X_1\otimes\cdots\otimes X_t}u\,.
\end{align}
\end{proposition}

\section{Operator systems and realizations}\label{sec:opsysqr}

\subsection{Operator systems and completely bounded norms}
\label{appendixopsys}
The memory system $V$ of a finitely correlated state $\omega$ on $\A$ with regular realization $(V, e,\E, \rho)$ inherits the structure of an operator system \cite[Lem.~A.1]{Fannes1992}.
In brief, it is constructed from the positivity structure on $\mathbb M_n(\A)$ by push-forward with the map $\Psi:\A_R\to V$ and its amplifications.

We will summarize some facts about the theory of operator systems in this section.
For more details see \cite{paulsen_completely_2003}.

\begin{definition}\label{def:opsys}
    Let $V$ be a complex vector space with a (conjugate-linear) involution $*:V\to V$ and a family $C$ of cones $C_n\subset \mathbb M_n(V) \cong \mathbb M_n\otimes V$ defining an order on $\mathbb M_n(V)$.
    We then obtain a natural involution on $\mathbb M_n(V)$.
    Let $e\in V$.
    We call $V \equiv (V, C, e)$ an \emph{(abstract) operator system} if the following are satisfied
    \begin{enumerate}
        \item 
       Each $X\in C_n$ is self-adjoint.
        \item
        $e$ is an Archimedean matrix order unit for the cones $C_n$.
        \item
        For each $k\times n$-matrix $M\in\mathbb M_{k,n}$ it holds that $M^*C_k M\subset C_n$.
    \end{enumerate}
\end{definition}
Since $e$ is a matrix order unit, by definition for every $n$ and every self-adjoint $X\in\mathbb M_n(V)$ there exists $r>0$ such that $r e^{(n)} \ge X$ where $e^{(n)} = \1_{\mathbb M_n}\otimes e$.
For $X\in\mathbb M_n(V)$ let
$$S(X) = \begin{pmatrix}
    0 & X \\
    X^* & 0
\end{pmatrix}.$$
Note that $S(X) = S(X)^*$.
The \emph{order norm} of $X\in \mathbb M_n(V)$ is then defined as
$$\|X\|_{e^{(n)}} = \inf\{r\in\RR_+\ |\ r e^{(2n)}\ge S(X)\}\ge 0.$$
$e$ being Archimedean implies that the order norms are, in fact, norms and that the corresponding cones are closed.
For self-adjoint $X$ the order norm is also given as $\|X\|_{e^{(n)}} = \inf\{r>0\ |\ re^{(n)} \ge X, -X\}$.
It always holds that $\|e^{(n)}\|_{e^{(n)}} = 1$.
For a linear map $\alpha:V\to W$ between operator systems we denote its $n$-th amplification by $\alpha_n = \id_{\mathbb{M}_n}\otimes\ \alpha:\mathbb M_n(V)\to \mathbb M_n(W)$.
$\alpha$ is called \emph{positive} if it maps positive elements of $V$ to positive elements of $W$ and \emph{completely positive} if each of its amplifications is positive.
We then write $\alpha\in\CP(V,W)$.
It is \emph{unital} if it maps the order unit of $V$ to the order unit of $W$.
If $\alpha$ is completely positive then $\alpha$ is bounded with its operator norm satisfying $\|\alpha\| = \|\alpha(e_V)\|_{e_W}$.
The \emph{Choi-Effros theorem} \cite[Thm.\ 13.1]{paulsen_completely_2003} ensures that every abstract operator system can be represented faithfully as a \emph{concrete operator system}, meaning a $*$-invariant subspace of $\B(\H)$ for some Hilbert space $\H$.

The estimates in \cref{sec.statereconstruction} 
also require us to use norms of complete boundedness.
The natural setting for talking about these norms is that of an \emph{operator space}, that is a closed subspace of $\B(\H)$.
For more information see \cite{pisierIntroductionOperatorSpace2003}.
There is an abstract characterization of this structure, similar to the abstract characterization of operator systems above, such that every abstract operator space can be represented as a concrete operator space by \emph{Ruan's theorem} \cite[Thm.\ 13.4]{paulsen_completely_2003}.
The structure of an (abstract) operator system induces the structure of an (abstract) operator space.
A linear map $\alpha:V\to W$ between operator spaces is called \emph{completely bounded} if all its amplifications $\alpha_n$ are bounded and the sequence of operator norms $\|\alpha_n\|$ is uniformly bounded.
The corresponding \emph{norm of complete boundedness} is then
$$\|\alpha\|_{cb} \equiv \|\alpha\|_{\CB(V,W)} = \sup_n \|\alpha_n\|_{\B(\mathbb M_n(V), \mathbb M_n(W))}.$$
If $\alpha$ is a completely positive map between operator systems then $\|\alpha\|_{cb} = \|\alpha\| = \|\alpha(e_V)\|_{e_W}$.
If $\alpha$ is unital it is thus automatically completely contractive.
Furthermore, the following bound is true  \cite[Thm.~3.8]{pisierIntroductionOperatorSpace2003}:
\begin{lemma}\label{lemmaamp}
    If $\alpha$ has rank $m$ then $\|{\alpha}\|_{\CB(V,W)} \le m\|\alpha\|_{\B(V,W)}$.
    In particular, this is the case whenever $V$ or $W$ is finite-dimensional.
\end{lemma}
We also have the following result. It is basically \cite[Prop.~1.12]{pisierIntroductionOperatorSpace2003} which is valid also for abstract operator spaces by Ruan's theorem.
\begin{lemma}\label{lemmaamp2}
    If $\alpha$ is a map from an (abstract) operator space $V$ to $\mathbb{M}_n$, then $\|{\alpha}\|_{cb} = \|\alpha_n\|_{\B(\mathbb{M}_n(V),\mathbb{M}_{n^2})}$.
\end{lemma}

\begin{remark}
    While there is a clear link between the structures of operator systems and operator spaces, the corresponding categories do not behave exactly the same.
    For example, the natural operator system structure induced on the quotient of an operator system does not in general coincide with the quotient operator space structure.
    A second subtlety concerns conventions in the literature.
    Concrete operator systems and, in particular, spaces are often assumed to be closed (hence complete) in norm.
    For the abstract structures, in particular operator systems, this is often not required (and then also not true when represented faithfully).
    This is certainly not a real problem in practice but should be kept in mind to avoid confusion.
    Both subtleties are illustrated by \cite[Prop.~4.5]{KavrukEtAl2013}.
\end{remark}

Let $\H$ be a Hilbert space.
Then there exists a natural operator space structure on $\H$ satisfying that the  (conjugate linear) identification $\mathbb M_n(\H) \cong \mathbb M_n(\H)^*$ is completely isometric \cite[Ch.~7]{pisierIntroductionOperatorSpace2003}.
We will implicitly use this operator system structure whenever we talk about complete boundedness of maps on Hilbert spaces.
It holds \cite[Prop.~7.2]{pisierIntroductionOperatorSpace2003}:
\begin{lemma}\label{eq:euclidan_diamond}
    If $\H$ and $\K$ are Hilbert spaces and $T\in\B(\H,\K)$ then $\|T\|_{cb} = \|T\|$.
\end{lemma}
\hide{\begin{corollary}\label{eq:U_euclideannorm}
    If $T\in\B(\H,\K)$ is a partial isometry (or, for that matter, any other contraction) then $\|T\|_{cb}\le 1$.
\end{corollary}}

\subsection{Operator systems and finitely correlated states}\label{sec.opsysfcs}

The memory system of a finitely correlated state $\omega$ is by definition a finite-dimensional operator system (cf.\ \cref{appendixopsys}), which means that it can be represented as a finite-dimensional subspace of some algebra $\B(\H)$.
However, the underlying Hilbert space $\H$ does not need to be finite-dimensional, and indeed there are examples of operator systems that can only be embedded into infinite-dimensional $C^*$-algebras (see \cite[Thm.~2.4]{FarenickPaulsen2012} together with the fact that the group $C^*$-algebra of a free group is infinite-dimensional).
Moreover, an arbitrary finitely correlated state is not necessarily $C^*$-finitely correlated (in fact, already in the abelian case~\cite{fanizza2023quantum} there are counterexamples, but we don't know of properly non-commutative examples).
Now, it is true that $C^*$-finitely correlated states are weak-$*$-dense in all translation invariant states \cite[Prop.~2.6]{Fannes1992}.
Thus, we cannot distinguish between $C^*$-finitely correlated states and a generic translation invariant state by measuring only finitely many observables up to precision $\epsilon_{HS} >0$.
But since the memory size needed to realize these approximating $C^*$-finitely correlated states will usually diverge with the size of the region where the observables are supported, we are also interested in extending the algorithm directly to the case of a general finitely correlated state.

Let $\Psi_{n}\coloneqq \id_{\mathbb M_n}\otimes\Psi$, and $V_n\coloneqq\{\Psi_{n}(X),X\in \mathbb M_n(\mathcal{A}_R)\}$. The order we consider in the following is given by the cones 
$$C_n\coloneqq\overline{\{\Psi_{n}(X)|X\in \mathbb M_n(\mathcal{A}_R),X\geq 0\}},$$
where the closure is in weak-* topology, with units $e_n=\Psi_{n}(\1)$. It is immediate to check that this construction gives a bona-fide operator system.
The construction can also be dualized which gives us a minimal and a maximal order. This will implicitly be used in the proof of \cref{lemnormineq} and is commented on in \cref{remminmaxorder}.
The order gives rise to an order norm $\|\Psi_X\|_{\Psi_{\1}}$ (see Appendix~\ref{appendixopsys}) that in the case of a $C^*$ order yields back the $C^*$ norm. In addition to this generalized operator norm, we will also need an analogue of the the Hilbert-Schmidt norm for the estimates in \cref{secparest}.
Since we cannot guarantee the memory systems to be represented on a finite-dimensional Hilbert space
we cannot directly use the Hilbert-Schmidt norm.
However, we can use the maps $\Psi_X$ defined in Equation \eqref{PsiX} to define an inner product. For $X\in\A_{[1,t^*]}$ set $\xi_X(i) = \tr{XX_i}$ and $X' = \sum_i \Psi_X(X_i) X_i$.
Then clearly $\Psi_X(Y) = \tr{X' Y}$ and $\xi_{X'}(i) = \Psi_X(X_i)$, so $\xi_{X'} = \ddot{\Omega}\xi_X$ where $\ddot{\Omega}_{ij} = \omega(X_i\otimes X_j)$ (`Ömega') is $\Omega$ written in coordinates.
Furthermore, $\Psi_{X} = \Psi_Z$ if and only if $\xi_{X'} = \xi_{Z'}$ if and only if $X' = Z'$.
We then set
$$(\Psi_X, \Psi_Y)_{\Omega} := \tr{(X')^* Y'} = \xi_{X'}\cdot\xi_{Y'} = \sum_i \overline{\Psi_X(X_i)}\Psi_Y(X_i)$$
where $x\cdot y = \sum_I \overline{x(I)}y(I)$ is the inner product on $\mathbb C^{d^{t^*}}$.
$(\cdot,\cdot)_\Omega$ thus defines an inner product on $V$.
We denote the corresponding norm by $\|\cdot\|_\Omega$ and observe the following inequalities:
\begin{lemma}\label{lemnormineq}
    Let $X\in\A_{[1,t^*]}=\C_R$.
    It holds that $\|{\Psi_X}\|_{\Psi_{\1}} \le \|P(X)\|_{2}$, where $P$ is the Hilbert-Schmidt orthogonal projection onto $\supp(\Omega)$, and
    \begin{align}
        \sigma_m(\Omega)\|{\Psi_X}\|_{\Psi_{\1}} \le\|\Psi_X\|_\Omega
        \le \|\Psi_X\|_{\Psi_{\1}}.
    \end{align}
   
\end{lemma}
\begin{proof}
First, we observe that $\Psi_{P(X)}= \Psi_X$ by definition of $P$, and that $\Psi$ is completely positive and unital, thus contractive. 
Then
\begin{align*}
\|\Psi_X\|_{\Psi_{\1}}=\|\Psi_{P(X)}\|_{\Psi_{\1}}
&\le \|P(X)\|_\infty\le \|P(X)\|_2\,.
\end{align*}
        
        \hide{
        Next, it holds that
        \begin{align*}
            0\le \sigma_m(\Omega)^2\|\Psi_X\|_\Omega^2
             =\omega((X')^*\otimes X)
             \le\|{X}\|_{\A_{[1,N]}}\|{X'}\|_{\A_{[1,N]}}
             \le \|{X}\|_{\A_{[1,N]}}\|{X'}\|_{2}
            = \sigma_m(\Omega)\|{X}\|_{\A_{[1,N]}}\|{\Psi_X}\|_\Omega
        \end{align*}}
        Next, we have 
        \begin{align*}
             \|{\Psi_X}\|_\Omega^2= \|{\Psi_{P(X)}}\|_\Omega^2 = (\ddot\Omega\xi_{P(X)})\cdot(\ddot\Omega\xi_{P(X)})\ge \sigma_m(\Omega)^{2}\xi_{P(X)}\cdot \xi_{P(X)}= \sigma_m(\Omega)^{2}\|{P(X)}\|^2_{2}\,.
        \end{align*}
                For the last inequality, consider the matrix amplification $\Psi_{2}$ of $\Psi$ on $\mathbb M_2(\C_R)$ mapping to $\mathbb M_2(\C_L^*)\cong \mathbb M_2(\C_L)^*$ (the isomorphism is implemented by letting $A\otimes \varphi$ act on $B\otimes X$ as $(A\otimes\varphi). (B\otimes X) := \tr{AB}\varphi(X)$).
        In particular, $\Psi_{2}$ is completely positive.
        For $X\in \C_R$ let
        $$S(X) = 
        \begin{pmatrix}
            0 & X\\
            X^* & 0
        \end{pmatrix}
        = \ketbra{0}{1}\otimes X + \ketbra 10 \otimes X^*
        $$
        be the embedding of $\C_R$ as self-adjoint elements in $\mathbb M_2(\C_R)$.
        Note that $\Psi_{2, S(X)} = S(\Psi_X)$.
        Let $\lambda\ge\|X'\|_\infty$ so that $\lambda\1 - S(X')$ is a positive matrix.
        Furthermore, let $r> 0$.
        We then have that
        \begin{align*}
            \Psi_{2,r\1 + S(X)}(\lambda\1 - S(X'))
            =& r\lambda (\id_{\mathbb M_2}\otimes\Psi_{\1}).\1 + \lambda S(\Psi_X).\1\\
            & - r(\id_{\mathbb M_2}\otimes\Psi_{\1}).S(X') - S(\Psi_X).S(X')\\
            =& 2r\lambda - S(\Psi_X).S(X').
        \end{align*}
        Here, we used $(\id_{\mathbb M_2}\otimes\Psi_{\1}).\1 = 2$ and that if $A$ is block-diagonal and $B$ is `block-off-diagonal' then $(A\otimes\varphi).(B\otimes X) = 0$. 
        Now, we have that
        \begin{align*}
            S(\Psi_X).S(X') = \Psi_X((X')^*) + \Psi_{X^*}(X') = 2\|{\Psi_X}\|_\Omega^2.
        \end{align*}
We find that 
\begin{align*}
\Psi_{2,r\1+S(X)}(\lambda\1-S(X'))=2(r\lambda-\|\Psi_X\|_\Omega^2)\,.
\end{align*}
        Thus, choosing $\lambda = \|\Psi_X\|_\Omega = \|X'\|_{2} \ge \|{X'}\|_{\infty}$ and $r<  \|\Psi_X\|_\Omega$ we find
        $$r\|\Psi_X\|_\Omega < \|\Psi_X\|_\Omega^2$$
        showing that $\Psi_{2, r\1 + S(X)}$ is not a positive functional.
        By definition of the order norm this shows $r\le\|\Psi_X\|_{\Psi_{\1}}$ and thus the desired norm inequality (cf.\ also \cref{remminmaxorder}). 
\end{proof}
\begin{remark}\label{rem:ineq}
    Replacing $\C_R$ by $\mathbb M_n(\C_R)\cong \mathbb M_n\otimes\C_R$ (similarly for $\C_L$) and $\Psi$ by its $n$-th amplification $\Psi_n$, we can apply the above proof mutatis mutandis.
    This yields the inequalities
    $$\sigma_m(\Omega)\|\cdot\|_{{\Psi_n}(\1)}\le\|\cdot\|_{n,\Omega} \le \sqrt{n}\|\cdot\|_{{\Psi_n}(\1)}$$
    for the norms on $\mathbb M_n(V)$.
\end{remark}

The inequalities in \cref{lemnormineq} hold for the norms defined on the memory system of the abstract (minimal) realization.
Consider the restriction of $U^\intercal$ to $V$, $\Psi_X\mapsto U^{\intercal}\Psi_X$, which is invertible. It induces an operator system structure on $\mathbb{C}^{m}$ that is completely order isomorphic to the order structure of the regular realization and it holds that $U^\intercal\Psi_\1 = e$ as defined in \cref{sec.observablereal}. Also note that $\norm{(\id_{\mathbb{M}_n}\otimes U^{\intercal})\Psi_{n,X}}_2 = \norm{\Psi_X}_{n,\Omega}$.

In particular, the following inequalities hold for $x\in \mathbb{M}_n(\mathbb{C}^{m})$: 

\begin{equation}\label{eq:normineqobserv}
    \sigma_m(\Omega)\norm{x}_{\1_{\mathbb{M}_n}\otimes e} \le \norm{x}_2 \le \sqrt{n}\norm{x}_{\1_{\mathbb{M}_n}\otimes e}
\end{equation}
where $\norm{\cdot}_2$ is just the standard euclidean norm.

\begin{remark}\label{remminmaxorder}
    In fact, the matrix order on the memory system $V$ is not unique.
    The construction in \cite{Fannes1992} mentioned in the beginning of this section can be applied to both sides of the chain.
    Dualizing the order on one side then gives a matrix order on the other side and the two need not coincide.
    However, every positive element in the directly constructed one is also positive in the one obtained through dualizing and in fact every admissible matrix order making the generating map $\E$ completely positive has to lie between the two orders.
    We may thus call them minimal and maximal and since the cones of the minimal matrix order are smaller than those of the maximal matrix order, the corresponding order norms satisfy the reverse inequality.
    This has also been observed by I. Todorov \cite{notesIvan} (and possibly others).
    The proof of \cref{lemnormineq} works for the minimal matrix order for the leftmost inequality and implicitly uses the maximal matrix order for the rightmost.
    Therefore, the bounds are valid for the order norm of any admissible matrix order.
\end{remark}
\begin{remark}
    Every bipartite state between two $C^*$-algebras gives rise to operator system structures analogous to \cref{remminmaxorder}.
    Assuming that this operator system is finite-dimensional, it is also sufficient to consider a finite-dimensional subspace of the algebra, which can be assumed to be an operator (sub) system.
    If this operator subsytem can be represented on a finite-dimensional Hilbert space, we can again use the Hilbert-Schmidt inner product to define a norm on the space of `correlation functionals' that satisfies the bounds as in \cref{lemnormineq}.
    This might be useful to obtain bounds in other settings.
    We are not aware of this being mentioned in the literature explicitly in this context. However, the results of \cite{LANCE1973157, notesIvan, vanLuijkEtAl2023} are very close to this.
\end{remark}

\begin{remark}
We also note that a quotient realization of a $C^*$-finitely correlated state can be given the structure of an abstract operator system, which is not necessarily isomorphic to the the minimal or the maximal order, but is inherited from the usual matrix order on $\B$.  
$\V=\W/({\W\cap \tilde{\W}^{\perp}})$ can also be given the structure of an operator system, with unit $e'=L\1_{\B}$ and cones $C_n:= L(\mathbb M_n(\W)^+)$. By ``push-forward'' via the invertible map $M$, the observable realization also gets an operator system structure and $M$ automatically becomes unital and completely positive. This operator system, given by the cones $MC_n$ is admissible in the sense above. Indeed, the minimal operator system
in the observable realization is given by the cones $$C^{\min}_n\coloneqq\overline{\{\mathbb{K}^{n}_Xe|X\in \mathbb M_n(\mathcal{A}_{[1,n]}),X\geq 0,n\geq 1\}},$$
and the units $e_n$. The maximal operator system is still given by the units $e_n$, but the cones are  
$$C^{\max}_n\coloneqq\{\rho\mathbb{K}^{n}_Y|Y\in \mathbb M_n(\mathcal{A}_{[-n,0]}),Y\geq 0,n\geq 0\}^*.$$
Using the notation of Proposition~\ref{thm:quotient}, we have $\mathbb{K}^{t}_{X}e=M\mathbb{K}'^{t}_{X}e'=L\mathcal{E}^{t}_{X}(\1_{\B})$, and since $\mathcal{E}^{t}_{X}(\1_{\B})\geq 0$ if $X\geq 0$, $\mathbb{K}^{t}_{X}e\in C^{\min}_{n}$ implies $\mathbb{K}^{t}_{X}e\in MC_{n}$. On the other hand, if $Y\geq 0$ and $\mathcal{E}^{t_1}_{X}(\1_{\B})\geq 0$ (so that $\mathbb{K}^{t_2}_{X}e\in MC_n$),
\begin{align}
\rho\mathbb{K}^{t_1}_{Y}\mathbb{K}^{t_2}_{X}e=\rho'\mathbb{K}'^{t_1}_{Y}\mathbb{K}'^{t_2}_{X}e'=\rho_0\mathcal{E}^{t_1}_{Y}|_{W}\mathcal{E}^{t_2}_{X}(\1_{\B})\geq 0,
\end{align}

and thus $\mathbb{K}^{t_2}_{X}e\in C_n^{\max}$.

The map $R$ identifying $\V$ with $\V'\subset\W$, see \cref{finvert}, then is a (unital) complete order isomorphism.
$\V$ inherits the Hilbert-Schmidt norm from $\B$, which we denote as $\norm{\cdot}_{2}$ without ambiguity.
For $x\in \mathbb{M}_n(\V)$ we have $\norm{x}_{(e')^{(n)}}\leq \norm{Rx}_{\1_{\mathbb{M}_n(\B)}}\leq \norm{x}_{2}$ by contractivity of $L$ and the standard inequality between Schatten norms. We also have that $U^{\intercal}\F R e'=e$, and for the norms induced by the cones $C_n$, 
\begin{align}
&\norm{\id_{\mathbb{M}_n} \otimes U^{\intercal}\F R}_{\1_{\mathbb{M}_n}\otimes e'\rightarrow 2}\leq \norm{\id_{\mathbb{M}_n} \otimes \id_{\mathbb{C}^m}}_{\1_{\mathbb{M}_n} \otimes e\rightarrow 2}\norm{\id_{\mathbb{M}_n} \otimes  U^{\intercal}\F R}_{\1_{\mathbb{M}_n}\otimes e'\rightarrow {\1_{\mathbb{M}_n} \otimes e}} \leq \sqrt{n},\label{contractC}
\end{align}

where we used that the \eqref{eq:normineqobserv} (in particular, the second inequality) holds for the norms of any admissible matrix order on the memory system and, thus $\norm{\id_{\mathbb{M}_n} \otimes \id_{\mathbb{C}^m}}_{\1_{\mathbb{M}_n} \otimes e\rightarrow 2}\leq \sqrt{n}$, and that $U^{\intercal}\F R$ is a unital complete order isomorphism.

\end{remark}

\section{State reconstruction for FCS}\label{sec.statereconstruction}

Given a finitely-correlated state $\omega$ on the infinite chain, we denote its restriction to $\mathcal{A}_{[1,t]}$ as $\omega_t$, i.e.
\begin{align}
\omega_t(X_1\otimes \cdots\otimes X_t):=\rho\mathbb{K}_{X_1}\cdots \mathbb{K}_{X_t}e\,.
\end{align}
Consider the corresponding map $\Omega_{(\cdot)}:(A,A')\mapsto \Omega_A(A')$ as defined in Equation \eqref{def.OmegaOmega()} and its observable realization $(\mathbb{C}^m,e,\mathbb{K},\rho)$ as in \cref{prop.observablequasireal}.
We can reconstruct an approximation of the realization parameters from estimates $\hat{\Omega},\hat{\Omega}_{(\cdot)}, \widehat{\Omega(\1)}, \widehat{\tau{\Omega}}$ obtained from estimating the expectation values of $\omega_t$. Here, we need to impose that the estimate $\hat{\Omega}$ has at least rank $m$, and we take its $m$-truncated singular value decomposition, meaning that we only retain the first $m$ singular values and singular vectors. This means that $\hat{U}$ is the matrix of the first $m$ left singular vectors of $\hat\Omega$.

\begin{definition}[Spectral state reconstruction]\label{defstaterecon}

\begin{equation}
\hat\omega_{t}(X_1\otimes\cdots\otimes X_t):=\hat{\rho}\hat{\mathbb{K}}_{X_1}\cdots\hat{\mathbb{K}}_{X_t}\hat{e},
\end{equation}
where
\begin{align}
\hat{e}&={\hat{U}}^{\intercal}\widehat{\Omega(\1)},\label{hate}\\
\hat{\rho}&=\widehat{\tau{\Omega}}(\hat{U}^\intercal \hat{\Omega})^+,\label{hatrho}\\
\hat{\mathbb{K}}_{A}&={\hat{U}}^{\intercal}(\hat\Omega_{A})({\hat{U}}^{\intercal}\hat\Omega)^+\,.\label{hatKA}
\end{align}
\end{definition}

To bound the accuracy of the estimate $\hat\omega_t$ it is convenient to compare the above empirical reconstruction with the empirical realization, rather than the observable realization.

\subsection{Empirical estimates of realization parameters}
In the following, we obtain bounds on the estimates of realization parameters, in terms of the accuracy of the estimation of the marginals.

\begin{lemma}If $\|\Omega - \hat{\Omega} \|_{2\rightarrow 2}\leq \frac{\sigma_m(\Omega)}{3}$, it holds that:
\begin{align}
    \| \tilde{\mathbb{K}} - \hat{\mathbb{K}} \|_{2\rightarrow 2} &\leq \frac{1+\sqrt{5}}{2} \frac{\|\Omega - \hat{\Omega} \|_{2\rightarrow2}}{\min \lbrace \sigma_m (\hat{\Omega}) , \sigma_m (\hat{U}^{\intercal}\Omega) \rbrace^2} 
    + \frac{\| \Omega_{(\cdot)}-\hat{\Omega}_{(\cdot)}\|_{2\rightarrow 2}}{\sigma_m (\hat{U}^{\intercal}\hat\Omega)}\label{tildeKhatK},
    \\&\leq 4\left( \frac{\|\Omega - \hat{\Omega} \|_{2}}{\sigma_m ({\Omega})^2} + \frac{\| \Omega_{(\cdot)}-\hat{\Omega}_{(\cdot)}\|_{2}}{3\sigma_m (\Omega)}\right),\label{boundDelta12}\\
    \| \hat e-\tilde e \|_{2} &\leq \| \Omega(\1) - \widehat{\Omega(\1)} \|_{2},\label{bounde}
    \end{align}
    \begin{align}
\| \hat\rho-\tilde\rho \|_{2} & \leq \frac{1+\sqrt{5}}{2}\frac{\| \hat\Omega - \Omega \|_{2}}{ \min \lbrace \sigma_m(\hat\Omega) , \sigma_m (\hat{U}^{\intercal} \Omega ) \rbrace^2 } + \frac{\| \widehat{\tau\Omega} -\tau\Omega\|_2} {\sigma_m (\hat{U}^\intercal \hat\Omega )}\,\label{bounddelta12}\\
&\leq 4 \left(\frac{\| \hat\Omega - \Omega \|_{2}}{\sigma_m(\Omega)^2} + \frac{\| \widehat{\tau\Omega} -\tau\Omega\|_2} {3\sigma_m (\Omega)}\right)
 .\label{eathohattilde1}\\
 \| (\hat U^{\intercal} U)^{-1} \|_{2\rightarrow 2}&\leq \frac{2}{\sqrt{3}}\label{uupert}
\end{align}

\end{lemma}
\begin{proof}
In order to bound
\begin{align}
\| \tilde{\mathbb{K}}-\hat{\mathbb{K}} \|_{2\rightarrow 2} = \sup_{X\in\mathcal{A}\otimes \mathbb{C}^m,\, \| X \|_2 \leq 1} \| (\tilde{\mathbb{K}}-\hat{\mathbb{K}}) (X ) \|_2,
\end{align}
note that, by Equations \eqref{tildeKA} and \eqref{hatKA}, 
\begin{align}
  \|\tilde{\mathbb{K}}-\hat{\mathbb{K}}\|_{2\rightarrow 2}= \|\hat{U}^{\intercal}(\Omega_{(\cdot)})(\hat{U}^{\intercal}\Omega)^+-\hat{U}^{\intercal}(\hat{\Omega}_{(\cdot)})(\hat{U}^{\intercal}\hat{\Omega})^+\|_{2\rightarrow 2}
\end{align}
By triangle inequalities,
\begin{align}
    \| \tilde{\mathbb{K}} - \hat{\mathbb{K}}\|_{2\rightarrow 2} &\leq \| \hat{U}^{\intercal}(\Omega_{(\cdot)})\left( (\hat{U}^{\intercal}\Omega)^+ - (\hat{U}^{\intercal}\hat{\Omega})^+ \right) \|_{2\rightarrow 2}  +\| \hat{U}^{\intercal}(\Omega_{(\cdot)} - \hat{\Omega}_{(\cdot)})(\hat{U}^{\intercal}\hat\Omega)^+ \|_{2\rightarrow 2} \,.
\end{align}
Therefore using Lemma~\ref{lem:error_pseudoinverses} and the inequality $\|A\|_{2\rightarrow 2} = \norm{A}_\infty\leq \|A\|_{2} $ and the fact that $\sigma_m(\hat{\Omega})=\sigma_m(\hat{U}^\intercal\hat{\Omega})$,

\begin{align}
    \| \tilde{\mathbb{K}} - \hat{\mathbb{K}} \|_{2\rightarrow 2} &\leq \| \Omega_{(\cdot)}\|_{2} \frac{1+\sqrt{5}}{2} \frac{\|\Omega - \hat{\Omega} \|_{2\rightarrow 2}}{\min \lbrace \sigma_m (\hat{\Omega}) , \sigma_m (\hat{U}^{\intercal}\Omega) \rbrace^2} + \frac{\| \Omega_{(\cdot)}-\hat{\Omega}_{(\cdot)}\|_{2\rightarrow 2}}{\sigma_m (\hat{U}^{\intercal}\hat\Omega)}.
\end{align}
Moreover, given bases $\{X_\ell\}_{\ell=1}^{d_{\A}^{2s}}$, $\{A_j\}_{j=1}^{d_{\A}^2}$, $\{Y_i\}_{i=1}^{d_{\A}^{2s}}$ of $\mathcal{C}_R$, $\mathcal{A}$ and $\mathcal{C}_L$, respectively,
\begin{align}
   \| \Omega_{(\cdot)}\|_{2} &\leq \sqrt{\sum_{i=1}^{d_{\A}^{2s}}\sum_{j=1}^{d_{\A}^{2}}\sum_{l=1}^{d_{\A}^{2s}} | \omega(Y_i\otimes A_j \otimes X_\ell ) |^2}\le \sqrt{\Tr[\omega_{2t^*+1}^2]}\leq 1\,.
\end{align}
Therefore,
\begin{align}
    \| \tilde{\mathbb{K}} - \hat{\mathbb{K}} \|_{2\rightarrow 2} &\leq \frac{1+\sqrt{5}}{2} \frac{\|\Omega - \hat{\Omega} \|_{2\rightarrow2}}{\min \lbrace \sigma_m (\hat{\Omega}) , \sigma_m (\hat{U}^{\intercal}\Omega) \rbrace^2} 
    + \frac{\| \Omega_{(\cdot)}-\hat{\Omega}_{(\cdot)}\|_{2\rightarrow 2}}{\sigma_m (\hat{U}^{\intercal}\hat\Omega)}.\label{tildeKhatK1}
\end{align}

For the second inequality,  using the expressions for $\tilde e$ and $\hat e$ given in Equations 
\eqref{tildee2} and \eqref{hate} respectively
\begin{align}
     \| \hat e-\tilde e \|_{2} &= \|\hat{U}^{\intercal}(\Omega(\1) - \widehat{\Omega(\1)})\|_{2} \\
    &\leq \| \hat{U}^{\intercal} \|_{2\rightarrow 2} \| \Omega(\1) - \widehat{\Omega(\1)}\|_{2} \\
    &\leq \| \Omega(\1) - \widehat{\Omega(\1)}\|_{2}
\end{align}

For the third inequality, using the expressions for $\tilde\rho$ and $\hat\rho$ given in Equations 
\eqref{tilderho} and \eqref{hatrho} respectively,
\begin{align}
    \| \hat\rho-\tilde\rho \|_{2} & = \|\widehat{\tau\Omega}(\hat{U}^{\intercal}\hat\Omega )^+- \tau\Omega(\hat{U}^{\intercal}\Omega )^+\|_2 \\
    & \leq \|\tau\Omega \big((\hat{U}^{\intercal}\hat\Omega )^+- (\hat{U}^{\intercal}\Omega )^+ \big)\|_2 + \|\widehat{\tau\Omega}(\hat{U}^{\intercal}\hat\Omega )^+ - \tau\Omega (\hat{U}^{\intercal}\hat\Omega )^+\|_2 \\
   & \leq  \| (\hat{U}^{\intercal}\hat\Omega )^+ - (\hat{U}^{\intercal}\Omega )^+ \|_{2\rightarrow 2} \|\tau\Omega\|_2 + \| (\hat{U}^{\intercal}\hat\Omega )^+ \|_{2\rightarrow 2} \| \widehat{\tau\Omega} -\tau\Omega\|_2 \,.
\end{align}

Then using $\|\tau\Omega\|_2 = \sqrt{\sum_{i=1}^{d_{\A}^{2s}} {\omega}(X_i)^2 }\leq 1$ and Lemma~\ref{lem:error_pseudoinverses},
\begin{align}
\| \hat\rho-\tilde\rho \|_{2} & \leq \frac{1+\sqrt{5}}{2}\frac{\| \hat\Omega - \Omega \|_{2}}{ \min \lbrace \sigma_m(\hat\Omega) , \sigma_m (\hat{U}^{\intercal} \Omega ) \rbrace^2 } + \frac{\| \widehat{\tau\Omega} -\tau\Omega\|_2} {\sigma_m (\hat{U}^\intercal \hat\Omega )}\,.
\end{align}

Then, ~\eqref{boundDelta12}, \eqref{eathohattilde1},~\eqref{uupert} directly follow from perturbation bounds for singular values, which are collected in Appendix~\ref{appendixpert}, and in particular Lemma~\ref{lemma:subspace}. For these bounds we closely follow the treatment of previous works in the classical setting~\cite{Hsu2008,Siddiqi2009,balle2013learning}.
\end{proof}

\subsection{Error propagation bound}
\label{subsec:errorprop}
With $(\mathbb C^m, \K, e, \rho)$ being the observable realization from \cref{prop.observablequasireal}, in the following we will assume that $\V$ is an operator system with order unit $e'$ and $M:\V\rightarrow \mathbb{C}^m$ is an invertible linear map such that $Me'=e$ and that $\rho'=\rho M$ is a positive functional. It holds that $\rho'(e') = 1$ and that the map ${M}^{-1}{\mathbb{K}}{M}$ is unital, in the sense that ${M}^{-1}{\mathbb{K}}_{\1_{\A}}{M}e'=e'$. We further assume that it is completely positive. Of course, we can always take $\V$ to be the vector space of functionals $\Psi_A$, $A\in \A_{R}$, whose positive elements at level $n$ are $\Psi_n(A)$, $A\in\mathbb{M}_n(A)^+$.
The following error parameters will be used in the analysis: given $\widetilde{e}$, $\widetilde{\rho}$, $\widetilde{\mathbb{K}}_A$ and $\hat{M}$ as in Equations  
\eqref{tildeee}, \eqref{tilderhoo}, \eqref{tildek} and \eqref{hatM}.  
respectively, where $\hat{U}$ and is chosen as in Definition \ref{defstaterecon},
\begin{align}\label{errorpardef1}
\delta_1&:=\|(\hat\rho-\tilde{\rho})\hat{M}\|_{e',*},\\
\delta_{\infty}&:=\|\hat{M}^{-1}(\tilde{e}-\hat{e})\|_{e'},\label{errorpardef2}\\
\Delta&:=\|\hat{M}^{-1}(\tilde{\mathbb{K}}-\hat{\mathbb{K}})\hat{M}\|_{\1\otimes e'\rightarrow e',cb}.\label{errorpardef3}
\end{align}
Above, we denote by $\tilde{\mathbb{K}}:\A\otimes\mathbb{C}^m\to \mathbb{C}^m$ and $\hat{\mathbb{K}}:\A\otimes\mathbb{C}^m\to \mathbb{C}^m$ the maps $A\otimes x\mapsto \tilde{\mathbb{K}}_A(x)$ and $A\otimes x\mapsto \hat{\mathbb{K}}_A(x)$, respectively. 

Our first main result is the following error propagation bound.
We write $\|\cdot\|_1$ since finite-size marginals are essentially density operators.
Furthermore, if $\|\cdot\|$ is any norm, we denote the corresponding norm on the dual space by $\|\cdot\|_*$, so $\|\cdot\|_1 \equiv \|\cdot\|_{\mathscr A,*}$ if $\mathscr A$ denotes an arbitrary $C^*$-algebra.
We recall that for maps between two spaces with norms $\|\cdot\|_a$ and $\|\cdot\|_b$, say, we denote the corresponding operator norm as $\|\cdot\|_{a\to b}$, so, e.g., $\|\cdot\|_{a,*} = \|\cdot\|_{a\rightarrow p}$ for any $p$-norm on $\mathbb C$.

\begin{theorem}[Error propagation]\label{theoerrprop}
Let $\omega$ be a finitely correlated state. 
For the state reconstruction in Definition~\ref{defstaterecon}, and the error parameters defined in Eqs.~\eqref{errorpardef1},~\eqref{errorpardef2} and \eqref{errorpardef3}, if $\hat{U}^{\intercal}U$ is invertible, then
\begin{align}
\|\hat\omega_{t}-\omega_{t}\|_1\leq (1+\delta_1)(1+\delta_{\infty})(1+\Delta)^t-1.
\end{align}
\end{theorem}
\begin{proof}[Proof of Theorem~\ref{theoerrprop}]
By the triangle inequality, we have the following bound
\begin{align}
\|\hat\omega_{t}-\omega_{t}\|_1&\leq\underbrace{\|(\hat\rho-\tilde\rho)\tilde{\mathbb{K}}^{t}\tilde{e}\|_1}_{(\operatorname{I})}+\underbrace{\|(\hat\rho-\tilde\rho)(\hat{\mathbb{K}}^{t}\hat{e}-\tilde{\mathbb{K}}^{t}\tilde{e})\|_1}_{(\operatorname{II})}+\underbrace{\|\tilde{\rho}(\tilde{\mathbb{K}}^{t}\tilde{e}-\hat{\mathbb{K}}^{t}\hat{e})\|_1}_{(\operatorname{III})}\label{eqtriangleineq}
\end{align}

Let us start with $(\operatorname{I})$. By Equations \eqref{tildeee} and \eqref{tildek},
\begin{align}
\|(\hat\rho-\tilde\rho)\tilde{\mathbb{K}}^{t}\tilde{e}\|_1&=\|(\hat\rho-\tilde\rho)\hat M M^{-1} \mathbb{K}^{t}M(e')\|_1\nonumber\\
&\le \norm{(\hat\rho-\tilde{\rho})\hat M}_{e',*} \norm{M^{-1} \mathbb{K}^{t}M(e')}_{\infty \rightarrow e'}\\
&=\norm{(\hat\rho-\tilde{\rho})\hat M}_{e',*}\,,\label{eq(I)}
\end{align}
where we used that $\norm{M^{-1} \mathbb{K}^{t}M(e')}_{\infty\rightarrow e'}=1$ because $M^{-1} \mathbb{K}^{t}M$ is unital. By similar arguments, $(\operatorname{II})$ is bounded as
\begin{align}
\|(\hat\rho-\tilde\rho)(\tilde{\mathbb{K}}^{t}\tilde{e}-\hat{\mathbb{K}}^{t}\hat{e})\|_1
&\leq \norm{(\hat\rho-\tilde{\rho})\hat M}_{e',*}\|(\hat M)^{-1}(\tilde{\mathbb{K}}^{t}\tilde{e}-\hat{\mathbb{K}}^{t}\hat{e})\|_{\infty\rightarrow e'},\label{eq(II)}
\end{align}
using invertibility of $\hat U^{\intercal} U$. Finally, we bound $(\operatorname{III})$:
\begin{align}
\|\tilde\rho(\tilde{\mathbb{K}}^{t}\tilde{e}-\hat{\mathbb{K}}^{t}\hat{e})\|_1
&\leq \|\tilde\rho \hat M\|_{e',*}\|\hat {M}^{-1}(\tilde{\mathbb{K}}^{t}\tilde{e}-\hat{\mathbb{K}}^{t}\hat{e})\|_{\infty\rightarrow e'}\nonumber\\
&\leq\|\rho M\|_{e',*}\|\hat {M}^{-1}(\tilde{\mathbb{K}}^{t}\tilde{e}-\hat{\mathbb{K}}^{t}\hat{e})\|_{\infty\rightarrow e'}\label{eqIII}\\
&=\|\hat {M}^{-1}(\tilde{\mathbb{K}}^{t}\tilde{e}-\hat{\mathbb{K}}^{t}\hat{e})\|_{\infty\rightarrow e'},\label{eq(III)}
\end{align}
where we also used that $\rho'=\rho M$ is a non-negative functional and $\rho'(e')=1$. 
By definition of the cb norm, the triangle inequality, and denoting $\hat{e}_t:=\hat{\mathbb{K}}^t\hat{e}$ and $\tilde{e}_t:=\tilde{\mathbb{K}}^t\tilde{e}$, 

\begin{align}
\|\hat{M}^{-1}(\tilde{\mathbb{K}}^{t}\tilde{e}-\hat{\mathbb{K}}^{t}\hat{e})\|_{\infty\to e'} &\le
\|\hat{M}^{-1}(\tilde{\mathbb{K}}^{t}\tilde{e}-\hat{\mathbb{K}}^{t}\hat{e})\|_{\infty\rightarrow e', cb}\\
&\leq
\underbrace{\|\hat{M}^{-1}(\tilde{\mathbb{K}}-\hat{\mathbb{K}})\tilde{e}_{t-1}\|_{\infty\rightarrow e', cb}}_{(\operatorname{I}')}\\
&+\underbrace{\|\hat{M}^{-1}(\tilde{\mathbb{K}}-\hat{\mathbb{K}})(\hat{e}_{t-1}-\tilde{e}_{t-1})\|_{\infty\rightarrow e', cb}}_{(\operatorname{II}')}\nonumber\\&+\underbrace{\|\hat{M}^{-1}\tilde{\mathbb{K}}^{t}(\tilde{e}_{t-1}-\hat{e}_{t-1})\|_{\infty\rightarrow e', cb}}_{(\operatorname{III}')}\,.
\end{align}

For the first piece $(\operatorname{I}')$, we have
\begin{align}
\|\hat{M}^{-1}(\tilde{\mathbb{K}}-\hat{\mathbb{K}})\tilde{e}_{t-1}\|_{\infty\rightarrow e', cb}&=\|\hat{M}^{-1}(\tilde{\mathbb{K}}-\hat{\mathbb{K}})\hat{M}\hat{M}^{-1}\tilde{e}_{t-1}\|_{\infty\rightarrow e', cb}\nonumber\\
&\leq \Delta \|\hat{M}^{-1}\tilde{e}_{t-1}\|_{\infty\rightarrow e',cb}=\Delta\,,
\end{align}
where we recall that $\Delta$ is defined in Equation \eqref{errorpardef3}, because $\|\hat{M}^{-1}\tilde{e}_{t-1}\|_{\infty\rightarrow e',cb}=\|M^{-1}\mathbb{K}^{t-1}M(e')\|_{\infty\rightarrow e',cb}=1$ since $M^{-1}\mathbb{K}^{t-1}M$ is unital and completely positive. For the second piece $(\operatorname{II}')$, we have 
\begin{align}
\|\hat{M}^{-1}(\tilde{\mathbb{K}}-\hat{\mathbb{K}})(\hat{e}_{t-1}-\tilde{e}_{t-1})\|_{\infty\rightarrow e', cb}&=\|\hat{M}^{-1}(\tilde{\mathbb{K}}-\hat{\mathbb{K}})\hat {M}\hat{M}^{-1}(\hat{e}_{t-1}-\tilde{e}_{t-1})\|_{\infty\rightarrow e', cb}\nonumber \\
&\leq\|\hat{M}^{-1}(\tilde{\mathbb{K}}-\hat{\mathbb{K}})\hat{M}\|_{\1\otimes e'\rightarrow e', cb}\|\hat{M}^{-1}(\hat{e}_{t-1}-\tilde{e}_{t-1})\|_{\infty\rightarrow e', cb}\nonumber\\
&\leq \Delta \|\hat{M}^{-1}(\hat{e}_{t-1}-\tilde{e}_{t-1})\|_{\infty\rightarrow e', cb}\,.
\end{align}
Finally, for $(\operatorname{III}')$ we have
\begin{align}
\|\hat{M}^{-1}\tilde{\mathbb{K}}(\tilde{e}_{t-1}-\hat{e}_{t-1})\|_{\infty\rightarrow e', cb}&=\|M^{-1}\mathbb{K}^{t-1}M[\hat{M}^{-1}(\tilde{e}_{t-1}-\hat{e}_{t-1})]\|_{\infty\rightarrow e', cb}\label{III'}\\
&\leq \|M^{-1}\mathbb{K}^{t-1}M\|_{\1 \otimes e'\rightarrow e', cb}\|\hat{M}^{-1}(\tilde{e}_{t-1}-\hat{e}_{t-1})\|_{\infty\rightarrow e', cb}\nonumber\\
&=\|\hat{M}^{-1}(\tilde{e}_{t-1}-\hat{e}_{t-1})\|_{\infty\rightarrow e', cb}\,,\label{III'2}
\end{align}
where \eqref{III'} follows from Equations \eqref{tildek} and \eqref{tildeK}, and \eqref{III'2} holds since $M^{-1}\mathbb{K}^{t-1}M$ is unital. By combining the bounds we just found for $(\operatorname{I}')$, $(\operatorname{II}')$ and $(\operatorname{III}')$, we get
\begin{equation}
\|\hat{M}^{-1}(\tilde{e}_{t}-\hat{e}_{t})\|_{\infty\rightarrow e', cb}\leq \Delta+\Delta \|\hat{M}^{-1}(\tilde{e}_{t-1}-\hat{e}_{t-1})\|_{\infty\rightarrow e', cb}+\|\hat{M}^{-1}(\tilde{e}_{t-1}-\hat{e}_{t-1})\|_{\infty\rightarrow e', cb}\,.
\end{equation}
Recalling the definition of $\delta_\infty$ from Equation \eqref{errorpardef2}, we get by induction that
\begin{equation}
\|\hat{M}^{-1}(\tilde{e}_{t}-\hat{e}_{t})\|_{\infty\rightarrow e', cb}\leq (1+\Delta)^t\delta_\infty+(1+\Delta)^t-1\,.\label{eq.induction}
\end{equation}
Finally, combining Equations \eqref{eq.induction}, \eqref{eq(I)}, \eqref{eq(II)}, \eqref{eq(III)} and \eqref{eqtriangleineq}, and using the definition \eqref{errorpardef1} for $\delta_1$,

\begin{align}
\|\hat\omega_{t}-\omega_{t}\|_1&\leq \delta_1+(1+\delta_1)((1+\Delta)^t\delta_\infty+(1+\Delta)^t-1)\nonumber\\
&=(1+\delta_1)(1+\delta_\infty)(1+\Delta)^t-1\,.\nonumber
\end{align}
\end{proof}

\subsection{General bound}
\label{secparest}
With the notation of Section~\ref{sec:opsysqr}, let us consider the order norm on $\mathbb{C}^m$ denoted as $\norm{\cdot}_{e}$. Then the error propagation bound follows from Theorem~\ref{theoerrprop}, with $M$ being the identity map.

\begin{lemma}\label{parambounds}
If $\|\Omega - \hat{\Omega} \|_{2\rightarrow 2}\leq \frac{\sigma_m(\Omega)}{3}$, we have
\begin{align}
    \delta_\infty &\le \tfrac 2{\sqrt{3}\sigma_m(\Omega)}\| \Omega(\1) - \widehat{\Omega(\1)}\|_{2} \\ 
    \Delta&\le \frac{8 m\sqrt{d_\A}}{\sqrt{3}
    \sigma_m(\Omega)}\left( \frac{\|\Omega - \hat{\Omega} \|_{2}}{\sigma_m ({\Omega})^2} + \frac{\| \Omega_{(\cdot)}-\hat{\Omega}_{(\cdot)}\|_{2}}{3\sigma_m (\Omega)}\right),\label{Deltaboundgen}\\
    \delta_1 &\le  4 \left(\frac{\| \hat\Omega - \Omega \|_{2}}{\sigma_m(\Omega)^2} + \frac{\| \widehat{\tau\Omega} -\tau\Omega\|_2} {3\sigma_m (\Omega)}\right)
\end{align}
\end{lemma}
\begin{proof}
From~\eqref{errorpardef1}, 
\begin{align}
    \delta_{\infty} &= \|(\hat U^{\intercal} U)^{-1}(\tilde e-\hat e)\|_{e} \\
    &\leq \| (\hat U^{\intercal} U)^{-1} \|_{2\rightarrow e}\| \tilde e-\hat e \|_{2} \\
    &\leq \tfrac 1{\sigma_m(\Omega)}\| (\hat U^{\intercal} U)^{-1} \|_{2\rightarrow 2}\| \Omega(\1) - \widehat{\Omega(\1)}\|_{2}\\
    &\leq\tfrac 2{\sqrt{3}\sigma_m(\Omega)}\| \Omega(\1) - \widehat{\Omega(\1)}\|_{2} ,
\end{align}
where we used~\eqref{eq:normineqobserv} and~\eqref{bounde} for the second inequality, and~\eqref{uupert} for the last one.  
For $\Delta$ we have, from~\eqref{errorpardef3}
\begin{align}
    \Delta&=\|(\hat U^{\intercal} U)^{-1}(\tilde{\mathbb{K}}-\hat{\mathbb{K}})\hat U^{\intercal} U\|_{\1_\A\otimes e\rightarrow e,cb}\\ 
     &\leq m \|(\hat U^{\intercal} U)^{-1}(\tilde{\mathbb{K}}-\hat{\mathbb{K}})\hat U^{\intercal} U\|_{\1_\A\otimes e\rightarrow e}\\
     &\leq  m\|(\hat U^{\intercal} U)^{-1}\|_{2 \rightarrow e}
     \|\tilde{\mathbb{K}}-\hat{\mathbb{K}}\|_{2\rightarrow 2}
     \|\hat{U}^{\intercal}U\|_{2\rightarrow 2}
     \|\id_{\A\otimes \mathbb{C}^m}\|_{\1_\A\otimes e\rightarrow 2,2} \\
      &\leq m \tfrac{\sqrt{d_{\A}}}{\sigma_m(\Omega)} \|(\hat U^{\intercal} U)^{-1}\|_{2\rightarrow 2} \|\tilde{\mathbb{K}}-\hat{\mathbb{K}}\|_{2\rightarrow 2}\label{equivSchatten2}\\
      &\le \frac{8 m\sqrt{d_\A}}{\sqrt{3}
    \sigma_m(\Omega)}\left( \frac{\|\Omega - \hat{\Omega} \|_{2}}{\sigma_m ({\Omega})^2} + \frac{\| \Omega_{(\cdot)}-\hat{\Omega}_{(\cdot)}\|_{2}}{3\sigma_m (\Omega)}\right),
\end{align}
where we used Lemma~\ref{lemmaamp} for the first inequatity,~\eqref{eq:normineqobserv} for the second inequality, and~\eqref{boundDelta12} for the last one.
Lastly, for $\delta_1$, from~\eqref{errorpardef1}
\begin{align}
    \delta_{1} &= \|(\hat\rho-\tilde\rho)\hat U^{\intercal} U\|_{e,*}\\
    &\leq \|(\hat\rho-\tilde\rho)\hat U^{\intercal} U\|_{2}\\
    &\leq \|\hat\rho-\tilde\rho\|_{2}\label{deltabound1}\,,\\
    &\le  4 \left(\frac{\| \hat\Omega - \Omega \|_{2}}{\sigma_m(\Omega)^2} + \frac{\| \widehat{\tau\Omega} -\tau\Omega\|_2} {3\sigma_m (\Omega)}\right), 
\end{align}
where we used~\eqref{eathohattilde1} for the last inequality.
\end{proof}

We can thus prove the following. 
\begin{theorem}[Reconstruction error from local estimates]\label{thmerrorpropagation} Assuming that, for $\epsilon<1$, 
\begin{align}
 \|\tau\Omega - \widehat{\tau\Omega}) \|_{2}&\leq \frac{3\sigma_m(\Omega)}{4}\,\epsilon\,,\\ 
 \|\Omega(\1) - \widehat{\Omega(\1)} \|_{2}&\leq \frac{\sqrt{3}\sigma_m( \Omega)}{2}
 \epsilon\,,\\
 \|\Omega - \hat{\Omega} \|_{2}&\leq \frac{1}{3}\min\left(\frac{\sigma_m(\Omega)^2}{4}
 \epsilon, \frac{\sqrt{3}\sigma_m(\Omega)^3}{8 t m\sqrt{d_{\A}}
 }\epsilon\right),\label{eqmala}\\  
 \|\Omega_{(\cdot)} - \hat{\Omega}_{(\cdot)} \|_{2}&\leq \frac{3\sqrt{3}\sigma_m(\Omega)^2}{8 tm \sqrt{d_{\A}}}\epsilon\,,
 \end{align}
 
we have 
\begin{align}
\|\hat\omega_{t}-\omega_{t}\|_1\leq \frac{145}{9} \epsilon\,.
\end{align}

 \begin{proof}
 It is immediate to verify that the above inequalities are such that $\|\Omega - \hat{\Omega} \|_{2\rightarrow 2}\leq \|\Omega - \hat{\Omega} \|_{2}\leq \frac{\sigma_m(\Omega)^2}{3}\leq \frac{\sigma_m(\Omega)}{3}$, because $\sigma_m(\Omega)\leq \|\Omega\|_2\leq 1$.  We can thus insert the bounds from Proposition~\ref{parambounds} in the bound of Theorem~\ref{theoerrprop}, which now reads $\|\hat\omega_{t}-\omega_{t}\|_1\leq (1+\epsilon)(1+\frac{4}{3}\epsilon)(1+\frac{4}{3t}\epsilon)^t-1,$ and use $(1+a/t)^t\leq 1+2a$ for $0\leq a\leq 1$.
 \end{proof}

\end{theorem}

Thanks to these bounds, for general finitely-correlated states, we can prove (see~\cref{defrelclass} for the definiton of the class $\mathcal{S}(m,s,\eta)$):

\begin{theorem}[Theorem~\ref{theomainintro}, restated]\label{theomain}
Let $\tilde{\omega}\in \mathcal{S}(m,s,\eta)$. Let $\hat{\omega}_{s},\hat{\omega}_{2s}, \hat{\omega}_{2s+1}$ be estimates of ${\omega}_{s},{\omega}_{2s}, {\omega}_{2s+1}$ respectively, such that $D_{HS}(\hat{\omega}_{s},{\omega}_{s})$, $D_{HS}(\hat{\omega}_{2s},{\omega}_{2s})$, $D_{HS}(\hat{\omega}_{2s+1},{\omega}_{2s+1})$ are smaller than $\frac{\epsilon\eta^3}{20 t m \sqrt{d_{\A}}}$. Then, $\hat{\omega}_t$ constructed from the output of Algorithm~\ref{alg:learnFCS} satisfies
\begin{equation}
\frac{\|\hat{\omega}_t-\omega_t\|_1}{2}\leq \epsilon.
\end{equation}
\end{theorem}

\begin{proof}
According to the error propagation bound in Theorem~\ref{thmerrorpropagation}, the bottleneck is the second term in Eq.~\eqref{eqmala}, which is $ O\left(\frac{\epsilon\sigma_m(\Omega)^3}{t m\sqrt{d_{\A}}}\right)$. All the relevant marginals can be estimated at that precision with high probability with any preferred method. Then, the first $m$ singular values of the estimated marginals will be closer than $O\left(\frac{\epsilon\sigma_m(\Omega)^3}{t m\sqrt{d_{\A}}}\right)$ to the true values, via Lemma~\ref{lem:error_perturbation}, meaning that singular value decomposition truncated to singular values larger than $\sigma_m(\Omega)-O\left(\frac{\epsilon\sigma_m(\Omega)^3}{t m\sqrt{d_{\A}}}\right)$ selects a map $\hat{U}$ of rank $m$  and the error propagation bound applies.
\end{proof}
\subsection{Bound for $C^*$-FCS}\label{sec.statereconstructionc}

If the state admits a $C^*$-realization, we can make use of the operator system inherited by the usual $C^*$ matrix order to get an alternative error propagation bound, which can be better than the one obtained for general states. In fact, the new sample complexity bound is obtained by exchanging $m$ with $d_{\B}$, and we have $m\leq d^2_{\B}$, making the particular bound for $C^*$ states non trivial. The most relevant difference in the proof is that we can explicitly compute the cb norm of a map into $\B$ by considering a single amplification which is determined by $d_{\B}$, see \cref{lemmaamp2}.
With the notation of Section~\ref{sec:opsysqr}, let us consider the order norm on $\V$ denoted as $\norm{\cdot}_{e'}$, $e'=L\1_{\B}$. Then the error propagation bound from Theorem~\ref{theoerrprop} applies, where we can choose the map $M$ to be $M=U^{\intercal} \F R$.

\begin{lemma}
We have that
\begin{align}
    \delta_\infty &\le \tfrac {2\sqrt{d_{\B}}}{\sqrt{3}\sigma_m(\Omega)}\| \Omega(\1) - \widehat{\Omega(\1)}\|_{2} \\ 
    \Delta&\le \frac{8 d_{\B}\sqrt{d_\A}}{\sqrt{3}
    \sigma_m(\Omega)}\left( \frac{\|\Omega - \hat{\Omega} \|_{2}}{\sigma_m ({\Omega})^2} + \frac{\| \Omega_{(\cdot)}-\hat{\Omega}_{(\cdot)}\|_{2}}{3\sigma_m (\Omega)}\right)\\
    \delta_1 &\le  4 \left(\frac{\| \hat\Omega - \Omega \|_{2}}{\sigma_m(\Omega)^2} + \frac{\| \widehat{\tau\Omega} -\tau\Omega\|_2} {3\sigma_m (\Omega)}\right)
\end{align}
\end{lemma}
\begin{proof}

We will use the Euclidean norm on $\V$ induced by the identification of $\V$ with a subspace $\V'$ of $\W$ via the map $R$. Via this identification we can also consider the Schatten norms of $x\in \mathcal{V}$. For $x\in\V$, we thus have $\norm{x}_{\infty}= \norm{R x}_{\infty}\leq \norm{R x}_{2}=\norm{x}_{2}$. This implies $\norm{RM^{-1}}_{2\rightarrow \infty}\leq \norm{RM^{-1}}_{2\rightarrow 2} $. We also have $\|x\|_{(Re')^{(n)}}\leq \|x\|_{\infty}$ for $x\in \mathbb{M}_n(\V)$ as an immediate consequence of the definition of the order norm on $\V$.
Note that the singular values of $U^{\intercal}\F R$  and $U^{\intercal}\F\Pi_{\V'}$ are the same. On the other hand, we can observe that the unnormalized $l$-th right singular vector of $U^{\intercal}\F\Pi_{\V'}$ can be written as $\cE^{t^*}_{X_l}(\1_{\B})$ for some $X_l\in\A_{[1,t^*]}$ and orthogonal to the kernel of 
$\Omega$, and $\|X_{l}\|_2=1$. But then 
\begin{equation}
\sqrt{d_\cB}\,\sigma_l(U^{\intercal}\F\Pi_{\V'})\ge\sigma_l(U^{\intercal}\F\Pi_{\V'})\|\cE^{(t^*)}_{X_l}(\1_{\B})\|_2 =\|U^{\intercal}\F\cE^{(t^*)}_{X_l}(\1_{\B})\|_2=\|U^{\intercal}\Omega(X_l)\|_2\geq \sigma_{m}(\Omega),\label{boundsingval}
\end{equation}
using that $\|\cE^{(t^*)}_{X_l}(\1_{\B})\|_2\le \sqrt{d_\cB}\|\cE^{(t^*)}_{X_l}(\1_{\B})\|_\infty\le \sqrt{d_\cB}$,
which implies 
\begin{equation}\label{Mminus1norm}
\norm{M^{-1}}_{2\rightarrow e'}=\norm{RM^{-1}}_{2\rightarrow Re'}\leq \norm{RM^{-1}}_{2\rightarrow \infty}\leq\norm{RM^{-1}}_{2\rightarrow 2}\leq \frac{\sqrt{d_{\B}}}{\sigma_{m}(\Omega)}.
\end{equation}

In the following, we repeat the steps of Lemma~\ref{parambounds}, explaining only what is different in this case. 
We have
\begin{align}
    \delta_{\infty} &= \|M^{-1}(\hat U^{\intercal} U)^{-1}(\tilde e-\hat e)\|_{e'} \\
    &\leq \norm{M^{-1}}_{2\rightarrow e'} \| (\hat U^{\intercal} U)^{-1} \|_{2\rightarrow 2}\| \tilde e-\hat e \|_{2} \\
    &\leq \tfrac {\sqrt{d_B}}{\sigma_m(\Omega)}\| (\hat U^{\intercal} U)^{-1} \|_{2\rightarrow 2}\| \Omega(\1) - \widehat{\Omega(\1)}\|_{2}\\
    &\leq\tfrac {2\sqrt{d_B}}{\sqrt{3}\sigma_m(\Omega)}\| \Omega(\1) - \widehat{\Omega(\1)}\|_{2},
\end{align}

where we used~\eqref{Mminus1norm} in the second inequality.

For $\Delta$ we have
\begin{align}
    \Delta&=\|M^{-1}(\hat U^{\intercal} U)^{-1}(\tilde{\mathbb{K}}-\hat{\mathbb{K}})\hat U^{\intercal} U M\|_{\1_\A\otimes e'\rightarrow e',cb}\\ 
    &=\|R M^{-1}(\hat U^{\intercal} U)^{-1}(\tilde{\mathbb{K}}-\hat{\mathbb{K}})\hat U^{\intercal} U M\|_{\1_\A\otimes e'\rightarrow Re',cb}\\ 
    &\leq\| R M^{-1}(\hat U^{\intercal} U)^{-1}(\tilde{\mathbb{K}}-\hat{\mathbb{K}})\hat U^{\intercal} U M\|_{\1_\A\otimes e'\rightarrow \infty,cb}\\ 
    &=\|\id_{\B}\otimes R M^{-1}(\hat U^{\intercal} U)^{-1}(\tilde{\mathbb{K}}-\hat{\mathbb{K}})\hat U^{\intercal} U M\|_{\1_\A\otimes \1_\B\otimes e'\rightarrow \infty}\\
     &\leq  \norm{\id_{\B} \otimes R M^{-1}}_{2\rightarrow \infty} \|(\hat U^{\intercal} U)^{-1}\|_{2 \rightarrow 2}
     \|\tilde{\mathbb{K}}-\hat{\mathbb{K}}\|_{2\rightarrow 2}
     \|\hat{U}^{\intercal}U\|_{2\rightarrow 2}
     \|\id_{\A}\otimes \id_{\B}\otimes M \|_{\1_\A\otimes e'\rightarrow 2} \\
      &\leq  \tfrac{\sqrt{d_{\A}}d_{\B}}{\sigma_m(\Omega)} \|(\hat U^{\intercal} U)^{-1}\|_{2\rightarrow 2} \|\tilde{\mathbb{K}}-\hat{\mathbb{K}}\|_{2\rightarrow 2}\label{equivSchatten22}\\
      &\le \frac{8 \sqrt{d_\A}d_{\B}}{\sqrt{3}
    \sigma_m(\Omega)}\left( \frac{\|\Omega - \hat{\Omega} \|_{2}}{\sigma_m ({\Omega})^2} + \frac{\| \Omega_{(\cdot)}-\hat{\Omega}_{(\cdot)}\|_{2}}{3\sigma_m (\Omega)}\right),
\end{align}
where we used that $\|x\|_{(Re')^{(n)}}\leq \|x\|_{\infty}$ in the first inequality, Lemma~\ref{lemmaamp2} in the third equality and \eqref{Mminus1norm} and \eqref{contractC} in the third inequality.
Lastly, for $\delta_1$:
\begin{align}
    \delta_{1} &= \|(\hat\rho-\tilde\rho)\hat U^{\intercal} U M\|_{e',*}\\
    &\leq \|(\hat\rho-\tilde\rho)\hat U^{\intercal} U\|_{2}\\
    &\leq \|\hat\rho-\tilde\rho\|_{2}\label{deltabound1c}\,\\
    &\le  4 \left(\frac{\| \hat\Omega - \Omega \|_{2}}{\sigma_m(\Omega)^2} + \frac{\| \widehat{\tau\Omega} -\tau\Omega\|_2} {3\sigma_m (\Omega)}\right) 
\end{align}
\end{proof}

We can thus prove
\begin{theorem}[Reconstruction error from local estimates]\label{thmerrorpropagationc} Assuming that, for $\epsilon<1$, 
 \begin{align}
 \|\tau\Omega - \widehat{\tau\Omega}) \|_{2}&\leq \frac{3\sigma_m(\Omega)}{4}\,\epsilon\,,\\ 
 \|\Omega (\1)- \widehat{\Omega(\1)} \|_{2}&\leq \frac{\sqrt{3}\sigma_m( \Omega)}{2\sqrt{d_B}}
 \epsilon\,,\\
 \|\Omega - \hat{\Omega} \|_{2}&\leq \frac{1}{3}\min\left(\frac{\sigma_m(\Omega)^2}{4}
 \epsilon, \frac{\sqrt{3}\sigma_m(\Omega)^3}{8 t d_{\B}\sqrt{d_{\A}}
 }\epsilon\right),\label{eqmala2}\\  
 \|\Omega_{(\cdot)} - \hat{\Omega}_{(\cdot)} \|_{2}&\leq \frac{3\sqrt{3}\sigma_m(\Omega)^2}{8 t d_{\B} \sqrt{d_{\A}}}\epsilon\,,
 \end{align}
we have
\begin{align}
\|\hat\omega_{t}-\omega_{t}\|_1\leq \frac{145}{9} \epsilon\,.
\end{align}

 \begin{proof}

 It is immediate to verify that the above inequalities are such that $\|\Omega - \hat{\Omega} \|_{2\rightarrow 2}\leq \|\Omega - \hat{\Omega} \|_{2}\leq \frac{\sigma_m(\Omega)^2}{3}\leq \frac{\sigma_m(\Omega)}{3}$, because $\sigma_m(\Omega)\leq \|\Omega\|_2\leq 1$.  We can thus insert the bounds from Proposition~\ref{parambounds} in the bound of Theorem~\ref{theoerrprop}, and use $(1+a/t)^t\leq 1+2a$ for $0\leq a\leq 1$.
 \end{proof}

\end{theorem}

The proof of the main theorem is analogous to the one of Theorem~\ref{theomain}. Note that the only difference is the dependence of the Hilbert-Schmidt error $\epsilon_{HS}$ on the parameters of the model, and see~\cref{defrelclassc} for the definiton of the class $\mathcal{S}_q(m,s,\eta)$.

\begin{theorem}\label{theomain2}
Let $\tilde{\omega}\in \mathcal{S}_q(d_{\mathcal{B}},s,\eta)$. Given $\hat{\omega}_{s},\hat{\omega}_{2s}, \hat{\omega}_{2s+1}$ estimates of ${\omega}_{s},{\omega}_{2s}, {\omega}_{2s+1}$ respectively, such that $D_{HS}(\hat{\omega}_{s},{\omega}_{s})$, $D_{HS}(\hat{\omega}_{2s},{\omega}_{2s})$, $D_{HS}(\hat{\omega}_{2s+1},{\omega}_{2s+1})$ are smaller than $\frac{\epsilon\eta^3}{20 t d_{\mathcal{B}} \sqrt{d_{\A}}}$. Then, $\hat{\omega}_t$ constructed from the output of Algorithm~\ref{alg:learnFCS} satisfies
\begin{equation}
\frac{\|\hat{\omega}_t-\omega_t\|_1}{2}\leq \epsilon
\end{equation}
\end{theorem}

\section{Non-translation-invariant FCS on a finite chain}
\label{secnontrans}

The techniques and the formalism of the translation invariant case can be generalized to the non-translation invariant case on a finite chain. We consider states given by the following definition

\begin{definition}[Finitely correlated states on a finite chain]
For a state $\omega\in\mathcal{A}_{[1,N]}^*$, let $\Omega^{[i,j,k]}:\A_{[j+1,k]}\rightarrow \A_{[i,j]}^*$ be defined as
\begin{align}
\Omega^{[i,j,k]}[X](Y)=\omega(\1^{\otimes i-1}\otimes Y\otimes X\otimes \1^{\otimes N-k}),\,\,  1\leq i\leq j\leq k\leq N,
\end{align}
and $\Omega^{[i,j,k]}=\Omega^{[\max\{1,i\},j,\min\{k,N\}]}$ for all the other values of $i$ and $k$, with $0\leq j\leq N$.
A state $\omega\in\mathcal{A}^*_{[1,N]}$ is $(r,l,m)$ finitely correlated if $\mathrm{rank}\,\Omega^{[1,j,N]}= \mathrm{rank}\,\Omega^{[j-l+1,j,j+r]}=:m_j\leq m$ for $0\leq j\leq N$. We also use the notation $\Omega^{[i,j,k]}_{X}:=\Omega^{[i,j,k]}(X\otimes \cdot)$.
\end{definition}

 Let also $\Omega^{[j-l+1,j,j+r]}=U_jD_jO_j$ be singular value decompositions. A notion of realizations for states on a finite chain can be obtained analogously to the translation invariant case, as well as a notion of regular realization, but we will not report on this in depth. 
\begin{proposition}[Observable realization]\label{reconstructionnontrans}Defining
\begin{align}
{\mathbb{K}}^{(1)}_{A}&=\Omega_{A}^{[1,0,1+r]}({U}_{1}^{\intercal}\Omega^{[1,1,1+r]})^+, \\
{\mathbb{K}}^{(N)}_{A}&={{U}}^{\intercal}_{N-1}\Omega_{A}^{[N-l,N-1,N]},\\
\mathbb{K}^{(j)}_{X_{j}}&:=U_{j-1}^{\intercal}\Omega^{[j-l,j-1,j+r]}_{X_{j}}(U_j^{\intercal}\Omega^{[j-l+1,j,j+r]})^+, \mathrm{for}\,2\leq j\leq N-1,
\end{align}
we have
\begin{align}
\omega(X_1\otimes \cdots\otimes X_N)
&=\mathbb{K}_{X_{1}}^{(1)}\cdots\mathbb{K}_{X_{N}}^{(N)}.
\end{align}
In analogy to the translation-invariant case it will be called \textbf{observable realization}. 
\end{proposition}
\begin{proof}
Let $V_{[i+1,j]}$ be the space generated by the span of the functionals $\Psi^{(i+1,j)}_{X}\in \A_{[\max(1,i-l+1),\min(i,N)]}^*$, $X\in\A_{[\max(1,i+1),\min(j,N)]}$, $\Psi^{(i+1,j)}_{X}(Y):=\omega(Y\otimes X)$ (where $V_{[N+1,N]}$ is one-dimensional and it is the span of the functional $\Psi(Y)=\omega(\1\otimes Y)$, $Y\in \A_{[N-l,N]}$, and $V_{[1,j]}$ is one-dimensional as it is the span of functionals $\Phi(X)=\omega(X\otimes \1)$, $X\in \A_{[1,j]}$ also a one-dimensional vector space). For $i=1,...,N$
we can define the linear maps $\mathbb{E}_A^{(i)}:V_{[i+1,i+r+1]}\rightarrow V_{[i,i+r]}$ as $\mathbb{E}^{(i)}_A\Psi^{(i+1,i+r+1)}_{X}=\Psi^{(i,i+r+1)}_{A\otimes X}=\Psi^{(i,i+r)}_{X'}$, where the last equality holds for some $X'\in A_{[\max(1,i+1),\min(N,i+r)]}$ because of the definition of $(l,r,m)$-finitely correlated state.
We obtain

\begin{equation}
\omega(X_1\otimes \cdots\otimes X_N)=\mathbb{E}_{X_{1}}^{(1)}\cdots\mathbb{E}_{X_{N}}^{(N)}.
\end{equation}
The image of $\Omega^{[j-l+1,j,j+r]}$ lies in $V_{[j+1,j+r+1]}$, so we can write
\begin{align}
\mathbb{E}_{X_{j}}^{(j)}\Omega^{[j-l+1,j,j+r]}=\Omega^{[j-l,j-1,j+r]}_{X_{j}},
\end{align}
where $\Omega^{[j-l,j-1,j+r]}_{X_{j}}:\A_{[j,\min(j+r,N)]}\rightarrow V_{[j,j+r]}=V_{[j,j+r-1]}$, $\Omega^{[j-l,j-1,j+r]}_{X_{j}}(X)=\Psi^{(j,j+r)}_{X_{j}\otimes X}$. Using the singular value decomposition
\begin{align}
\Omega^{[j-l+1,j,j+r]}={{U}}_{j}D_jO_{j},
\end{align}
with $U_{j}:\mathbb{R}^{m_j}\rightarrow V_{[j+1,j+r+1]}$, $D:\mathbb{R}^{m_j}\rightarrow \mathbb{R}^{m_j}$, $O_{j}: \A_{[j+1,j+r]}\rightarrow \mathbb{R}^{m_j}$. 
We get
\begin{align}
\mathbb{E}_{X_{N}}^{(N)}&=\Omega^{[N-l,N-1,N+r]}_{X_{N}}\\
\mathbb{E}_{X_{j}}^{(j)}{U}_{j}&=\Omega^{[j-l,j-1,j+r]}_{X_{j}}({U}_j^{\intercal}\Omega^{[j-l+1,j,j+r]})^+
\end{align}
Note that for $X\in A_{[j+1,j+r]}$, ${U}_j{U}_j^{\intercal}\Psi^{[j+1,j+r]}_{X}={U}_j{U}_j^{\intercal}\Omega^{[j-l+1,j,j+r]}(X)=\Omega^{[j-l+1,j,j+r]}(X)=\Psi^{[j+1,j+r]}_X$, therefore 
we have
\begin{align}
\omega(X_1\otimes \cdots\otimes X_n)&=
\mathbb{K}_{X_{1}}^{(1)}\cdots\mathbb{K}_{X_{N}}^{(N)}.
\end{align}
\end{proof}

 Note that $m_N=m_0=1$ and thus we do not need the analogue of $e$ and $\rho$ in this case because they are essentially included in $\mathbb{K}^{(N)}:\mathbb C \otimes \A\to \mathbb{C}^{m_{N-1}}$ and $\mathbb{K}^{(1)}:\mathbb{C}^{m_{1}}\otimes \A\to\mathbb C$ respectively.
In the following, we use the notation $\mathbb{K}^{[i,j]}:=\mathbb{K}^{(i)}\mathbb{K}^{(i+1)}\cdots\mathbb{K}^{(j)}$.

\begin{proposition}[Empirical realizations]\label{propempquasinonhom}
For any collection of maps $\hat U_j:\mathbb{C}^{m_j}\rightarrow \A_{[j-l+1,j]}^*$ (with real coefficients in a self-adjoint basis) such that $\hat{U}_j^{\intercal}U_j$ is invertible, defining

\begin{align}
\tilde{\mathbb{K}}^{(1)}_{A}&=\Omega_{A}^{[1,0,1+r]}({\hat{U}_{1}}^{\intercal}\Omega^{[1,1,1+r]})^+, \\
\tilde{\mathbb{K}}^{(N)}_{A}&={\hat{U}}^{\intercal}_{N-1}\Omega_{A}^{[N-l,N-1,N]},\\
\tilde{\mathbb{K}}^{(j)}_{A}&={\hat{U}}^{\intercal}_{j-1}\Omega_{A}^{[j-l,j-1,j+r]}({\hat{U}_{j}}^{\intercal}\Omega^{[j-l+1,j,j+r]})^+, \mathrm{for}\,2\leq j\leq N-1,
\end{align}
we have
\begin{align}
\omega(X_1\otimes \cdots\otimes X_N)
&=\tilde{\mathbb{K}}_{X_{1}}^{(1)}\cdots\tilde{\mathbb{K}}_{X_{N}}^{(N)}.
\end{align}
Any such specification of maps $\tilde{\mathbb{K}}^{(j)}_{(\cdot)}$ is called \textbf{empirical realization}.
\end{proposition}

A reconstruction algorithm similar to~\cref{alg:learnFCS} can be run in this case fixing a threshold $\eta$, with inputs  $\hat{\Omega}^{[j-l,j,j+r]}$ as an estimate of ${\Omega}^{[j-l,j,j+r]}$, $\hat{\Omega}_{(\cdot)}^{[j-l-1,j-1,j+r]}$ as an estimate of ${\Omega}_{(\cdot)}^{[j-l-1,j-1,j+r]}$ and computing $\hat{U}_{j}$ from obtained from the singular value decomposition of $\hat{\Omega}^{[j-l,j,j+r]}$ truncated to the singular values larger than $\eta/2$. As in the translation invariant case, as soon as the Hilbert-Schmidt norm errors are less than $\eta/3$, the matrices $\hat{U}_{j}$ correspond to those obtained from truncation to the first $m_j$ singular values. We refer to this algorithm as $\textsf{LearnFCS(l,r)}$. The estimated realization that it aims to computes is the following: 

\begin{definition}[Spectral state reconstruction, non-homogeneous case]\label{defstatereconnonhom}
Given estimates $\hat{\Omega}_{(\cdot)}^{[j-l-1,j-1,j+r]}$ of $\Omega_{(\cdot)}^{[j-l-1,j-1,j+r]}$, and estimates $\hat{\Omega}^{[j-l,j,j+r]}$, of ${\Omega}^{[j-l,j,j+r]}$ a state reconstruction $\hat\omega$ is obtained as

\begin{equation}
\hat\omega(X_1\otimes\cdots\otimes X_N)=\hat{\mathbb{K}}^{(1)}_{X_{1}}\cdots\hat{\mathbb{K}}^{(N)}_{X_{N}},
\end{equation}
where

\begin{align}
\hat{\mathbb{K}}^{(1)}_{A}&=\hat\Omega^{[1,0,1+r]}_{A}({\hat{U}_1}^{\intercal}\hat\Omega^{[1,1,1+r]})^+,\\
\hat{\mathbb{K}}^{(N)}_{A}&={\hat{U}_{N-1}}^{\intercal}\hat\Omega^{[N-l,N-1,N+r]}_{A},\\
\hat{\mathbb{K}}^{(j)}_{A}&={\hat{U}_{j-1}}^{\intercal}\hat\Omega^{[j-l,j-1,j+r]}_{A}({\hat{U}_j}^{\intercal}\hat\Omega^{[l-j+1,j,j+r]})^+,
\end{align}
and $\hat{U}_j\hat{D}_j\hat{O}_j$ is the singular value decomposition of $\hat{\Omega}^{[j-l+1,j,j+r]}$ truncated to the first $m_j$ singular values.
\end{definition}

Assuming we have operator systems $\V_j$, $j=0,\dots,N$ (with $\V_0$, $\V_N$ being $\mathbb{C}$ with unit $1$) with units $e_j$ and invertible linear maps $M_{j}:\V_j\rightarrow \mathbb{C}^{m_j}$, such that $M_{j}^{-1}\mathbb{K}^{[j+1,k]}M_k$ are unital and completely positive 
(where $M_0$ and $M_N$ are identity maps), we can thus define the following error parameters (where $\hat U^{\intercal}_{N}U_{N}$, $\hat U^{\intercal}_{0}U_{0}$ are just one dimensional identities):

\begin{align}
\Delta'&:=\max_{j\in [N]}\|M^{-1}_{j-1}(\hat{U}^{\intercal}_{j-1}U_{j-1})^{-1}(\tilde{\mathbb{K}}^{(j)}-\hat{\mathbb{K}}^{(j)})\hat{U}_{j}^{\intercal}U_{j}M_j\|_{\1_{\A}\otimes e_{j+1}\rightarrow e_{j},cb}.\label{errorpardefb32gen}
\end{align}

We also denote $\hat{M}_j\coloneqq\hat U^{\intercal}_{j}U_{j}M_j$. We then have

\begin{theorem}[Error propagation for general orders and finite-size chain]\label{theoerrprop22}

For the state reconstruction in Definition~\ref{defstaterecon}, and the error parameter defined in Eq.~\eqref{errorpardefb32gen}, if $\hat{U}_j^{\intercal}U_j$ are invertible, then

\begin{align}
\|\hat\omega-\omega\|_1\leq (1+\Delta')^N-1.
\end{align}
\end{theorem}

\begin{proof}[Proof of Theorem~\ref{theoerrprop}]

By definition we have
\begin{align}
\|\hat\omega_{t}-\omega_{t}\|_1&=\|\tilde{\mathbb{K}}^{[1,N]}-\hat{\mathbb{K}}^{[1,N]}\|_1\label{eqtriangleineq2}
\end{align}

We can bound this in a similar way to the translation invariant case:
\begin{align}
\|\tilde{\mathbb{K}}^{[1,N]}-\hat{\mathbb{K}}^{[1,N]}\|_1
&\leq \|\hat M_0\|_{e'_{0},*}\|\hat {M_0}^{-1}(\tilde{\mathbb{K}}^{[1,N]}-\hat{\mathbb{K}}^{[1,N]})\|_{\infty\rightarrow e'_0}\nonumber\\
&=\|\hat {M_0}^{-1}(\tilde{\mathbb{K}}^{[1,N]}-\hat{\mathbb{K}}^{[1,N]})\|_{\infty\rightarrow e'_0}
\end{align}

By definition of the cb norm and the triangle inequality, 
\begin{align}
\|\hat{M}_{j-1}^{-1}(\tilde{\mathbb{K}}^{[j,N]}-\hat{\mathbb{K}}^{j,N})\|_{\infty\to e'_{j-1}}
&\leq
\underbrace{\|\hat{M}_{j-1}^{-1}(\tilde{\mathbb{K}}^{(j)}-\hat{\mathbb{K}}^{(j)})\tilde{\mathbb{K}}^{[j+1,N]}\|_{\infty\rightarrow e'_{j-1}, cb}}_{(\operatorname{I}')}\\
&+\underbrace{\|\hat{M}_{j-1}^{-1}(\tilde{\mathbb{K}}^{(j)}-\hat{\mathbb{K}}^{(j)})(\hat{\mathbb{K}}^{[j+1,N]}-\tilde{\mathbb{K}}^{[j+1,N]})\|_{\infty\rightarrow e'_{j-1}, cb}}_{(\operatorname{II}')}\nonumber\\&+\underbrace{\|\hat{M}_{j-1}^{-1}\tilde{\mathbb{K}}^{(j)}(\hat{\mathbb{K}}^{[j+1,N]}-\tilde{\mathbb{K}}^{[j+1,N]})\|_{\infty\rightarrow e'_{j-1}, cb}}_{(\operatorname{III}')}\,.
\end{align}

For the first piece $(\operatorname{I}')$, we have
\begin{align}
\|\hat{M}_{j-1}^{-1}(\tilde{\mathbb{K}}^{(j)}-\hat{\mathbb{K}}^{(j)})\tilde{\mathbb{K}}^{[j+1,N]}\|_{\infty\rightarrow e'_{j-1}, cb}
&\leq \Delta' \|\hat{M}_{j}^{-1}\tilde{\mathbb{K}}^{[j+1,N]}\|_{\infty\rightarrow e'_{j},cb}=\Delta'\,,
\end{align}
where we recall that $\Delta$ is defined in Equation \eqref{errorpardefb32gen}, because \begin{equation}
\|\hat{M}_{j}^{-1}\tilde{\mathbb{K}}^{[j+1,N]}\|_{\infty\rightarrow e'_j,cb}=\|M_{j}^{-1}\mathbb{K}^{[j+1,N]}M_N M_{N}^{-1}\|_{\infty\rightarrow e'_j,cb}=1
\end{equation}
since $M_{j}^{-1}\mathbb{K}^{[j+1,N]}M_N$ is completely positive and unital. For the second piece $(\operatorname{II}')$, we have 
\begin{align}
\|\hat{M}_{j-1}^{-1}(\tilde{\mathbb{K}}^{(j)}-\hat{\mathbb{K}}^{(j)})(\hat{\mathbb{K}}^{[j+1,N]}-\tilde{\mathbb{K}}^{[j+1,N]})\|_{\infty\rightarrow e'_{j-1}, cb}
&\leq \Delta' \|\hat{M}_{j}^{-1}(\hat{\mathbb{K}}^{[j+1,N]}-\tilde{\mathbb{K}}^{[j+1,N]})\|_{\infty\rightarrow e'_j, cb}\,.
\end{align}
Finally, for $(\operatorname{III}')$ we have
\begin{align}
&\|\hat{M}_{j-1}^{-1}\tilde{\mathbb{K}}^{(j)}(\hat{\mathbb{K}}^{[j+1,N]}-\tilde{\mathbb{K}}^{[j+1,N]})\|_{\infty\rightarrow e'_{j-1}, cb}\nonumber\\
&\leq \|\hat{M}_{j-1}^{-1}\tilde{\mathbb{K}}^{(j)}\hat{M}_{j}\|_{\1 \otimes e'_j\rightarrow e'_{j-1}, cb}\|\hat{M}_{j}^{-1}(\hat{\mathbb{K}}^{[j+1,N]}-\tilde{\mathbb{K}}^{[j+1,N]})\|_{\infty\rightarrow e'_j, cb}\nonumber\\
&=\|\hat{M}_{j}^{-1}(\hat{\mathbb{K}}^{[j+1,N]}-\tilde{\mathbb{K}}^{[j+1,N]})\|_{\infty\rightarrow e'_j, cb}\,,\label{III'2nontrans}
\end{align}

since $\hat{M}_{j-1}^{-1}\tilde{\mathbb{K}}^{(j)}\hat{M}_{j}$ is unital. By combining the bounds we just found for $(\operatorname{I}')$, $(\operatorname{II}')$ and $(\operatorname{III}')$, we get
\begin{align}
\|\hat{M}_{j-1}^{-1}(\hat{\mathbb{K}}^{[j,N]}-\tilde{\mathbb{K}}^{[j,N]})\|_{\infty\rightarrow e'_{j-1}, cb}&\leq \Delta'+\Delta' \|\hat{M}_{j}^{-1}(\hat{\mathbb{K}}^{[j+1,N]}-\tilde{\mathbb{K}}^{[j+1,N]})\|_{\infty\rightarrow e'_{j}, cb}\nonumber\\&+\|\hat{M}_{j}^{-1}(\hat{\mathbb{K}}^{[j+1,N]}-\tilde{\mathbb{K}}^{[j+1,N]})\|_{\infty\rightarrow e'_{j}, cb}\,.
\end{align}

Therefore, we get by induction that
\begin{equation}
\|\hat{M}_{0}^{-1}(\hat{\mathbb{K}}^{[1,N]}-\tilde{\mathbb{K}}^{[1,N]})\|_{\infty\rightarrow e'_{0}, cb}\leq (1+\Delta')^N-1\,.\label{eq.inductionnontrans}
\end{equation}

\end{proof}

\subsection{Bound for $(r,l,m)$-FCS}

 As in the translation invariant case, we can construct operator systems $\mathcal{O}_j$ for which the maps $\mathbb{K}^{(j)}$, $j=1,\dots,N$ are completely positive (this amounts to choosing $M_j$ to be identity maps so that we can identify $\mathcal{O}_j$ with $\V_j$ of the previous section). Indeed, for $1\leq j \leq N-1$ take $\mathcal{O}_j$ to be matrix ordered by the cones $C_{n,j}:=\{\id_{\mathbb M_n}\otimes\mathbb{K}^{[j+1,N]}(X)\ |\ X\in \mathbb{M}_n(\A^{\otimes N-j}), n\in\mathbb{N}, X\geq 0 \}$, and the unit $e_{j}=\mathbb{K}^{[j+1,N]}(\1^{\otimes N-j})$ for $0\leq j \leq N-1$, while $\mathcal O_1$, $\mathcal O_N$ are $\mathbb{C}$ with unit 1. 
 Then by definition $\mathbb{K}^{(j)}:\mathbb{M}_{d_{\A}}(\mathbb{C}^{m_{j+1}})\rightarrow \mathbb{C}^{m_{j}}$ is a unital completely positive map between the operator systems $\mathbb{M}_{d_{\A}}(\mathcal{O}_{j+1})$ and  $\mathcal{O}_{j}$, for $j=0,...,N$.

The inequalities between the order norms and the Euclidean norm of \eqref{eq:normineqobserv} (induced by the Hilbert-Schmidt norm) are still true for each operator system $\mathcal{O}_j$. The error propagation bound can be obtained in the same way, with the caveat that the completely bounded norms appearing in the calculation are now norms of maps between operator systems that are possibly different. 
For $A\in (\mathbb{M}_n\otimes \mathbb{C}^{m_j})^*$, the dual norm (that is, the usual norm of a linear functional) is defined as 
\begin{equation}
\|A\|_{e^{(n)}_{l},*}:=\sup \{ |A(X)|\ |\  X \in \mathbb{M}_n\otimes\mathbb{C}^{m_j}, ||X||_{e_{j}^{(n)}}\leq 1\}.
\end{equation}
For $B\in \mathrm{Hom}(\mathbb{C}^{m_{j+1}},\mathbb{C}^{m_j})$, the completely bounded operator norm is defined as
\begin{equation}
\|B\|_{e_{j+1}\rightarrow e_{j},cb}:=\sup \{ \|(\id_{\mathbb M_n}\otimes B)(X)\|_{e_{j}^{(n)}}\ |\ X \in \mathbb{C}^{m_j}\otimes \mathbb{M}_n, \|X\|_{e_{j+1}^{(n)}}\leq 1, n\in \mathbb{N}\}.
\end{equation}

On the other hand, as we did in Section~\ref{sec:opsysqr}, we can also define norms from the inner product defined by the maps $\Omega^{[j-l,j,j+r]}$, observing that any vector $v$ in $\mathbb{C}^{m_j}$ can be written as $\mathbb{K}^{[j+1,N]}_{X(v)\otimes {\1}^{\otimes N-j-r}}$ for some $X(v)\in \A_{[j+1,j+r]}$, and then $(v,w)_{\Omega^{[j-l,j,j+r]}}:= \sum_{i=1}^{d_{\A}^{2l}}\omega(Y_i\otimes X(v))\omega(Y_i\otimes X(w)),$ with $\{Y_i\}$ being a self-adjoint orthonormal basis of $\A_{[j-l,j]}$. We denote such norms as $\|\cdot\|_{2,j}$, their amplifications as $\|\cdot\|_{2,j,n}$ and by the same argument of Lemma~\ref{lemnormineq} and Remark~\ref{rem:ineq} 
and therefore, as in~\eqref{eq:normineqobserv}, we have the inequalities

\begin{equation}
\norm{x}_{\1_{\mathbb{M}_n}\otimes e_j} \le \sigma_{m_j}(\Omega^{[j-l,j,j+r]})^{-1}\norm{x}_{2,j, n} \le \sqrt{n}\sigma_{m_j}(\Omega^{[j-l,j,j+r]})^{-1}\norm{x}_{\1_{\mathbb{M}_n}\otimes e_j}
\end{equation}

We can just take the $M_j$ in~\eqref{errorpardefb32gen} to be identity maps and then we have, in complete analogy with~\eqref{Deltaboundgen}

\begin{align}
\Delta'&\leq\max_{j\in [N]}\frac{8 m_{j-1}\sqrt{d_{\A}}}{\sqrt{3}\sigma_{m}\left(\Omega^{[j-l,j-1,j-1+r]}\right)}\nonumber\\&\times\left( \frac{\|\Omega^{[j-l+1,j,j+r]}  - \hat{\Omega}^{[j-l+1,j,j+r]} \|_{2}}{\sigma_m (\Omega^{[j-l+1,j,j+r]})^2} + \frac{\| \Omega_{(\cdot)}^{[j-l,j-1,j+r]}-\hat{\Omega}^{[j-l,j-1,j+r]} _{(\cdot)}\|_{2}}{3\sigma_m (\Omega^{[j-l+1,j,j+r]} )}\right).\label{bounderrorpardef24D}
\end{align}

By using the error progapagation bound in Theorem~\ref{theoerrprop22}, we obtain

\begin{theorem}\label{theonontrans}
Let $\omega$ be a $(l,r,m)$-finitely correlated state, with associated maps $\Omega^{[i,j,k]}$ satisfying $\eta\leq \min_{j=0,\dots,N}\sigma_{m}(\Omega^{[j-l+1,j,j+r]})/2$. Assume 
we have estimates of marginals  $\omega_{[\max(j-l,1),\min(j+r,N)]}$, $\omega_{[\max(j-l+1,1),\min(j+r,N)]}$, $j=0,...,N$, such that the Hilbert-Schmidt norm errors are smaller than $ \frac{\epsilon\eta^3 }{20 tm  \sqrt{d_{\A}}}$. 
The estimated realization parameters from $\textsf{LearnFCS(l,r)}$ (given by~\ref{defstatereconnonhom}) give rise to an estimated operator $\hat\omega$ such that
\begin{equation}
\frac{\|\hat{\omega}-\omega\|_1}{2}\leq \epsilon,
\end{equation}
\end{theorem}

A sample complexity bound follows immediately as in the non-translation invariant case.

\subsection{Bound for $C^*$-$(r,l,m)$-FCS}
We can also define $C^*$-realizations for states on a finite chain; note that any state on a finite chain admits a realization of this kind, but the dimension of the memory system vector space can grow exponentially with $N$ in general.

\begin{definition}\label{compposnontrans}
A \textbf{$C^*$-realization} of a state $\omega$ is defined by a unital $C^*$-algebra $\B$, a state $\rho_0\in\B^*$ and a collection of unital maps $\mathcal E^{(j)}:\A\otimes \B\rightarrow\B$ such that
\begin{equation}
\omega(X_1\otimes \cdots\otimes X_N)=\rho_0\mathcal E_{X_{1}}^{(1)}\cdots\mathcal E_{X_{N}}^{(N)}(\1_{\B}).
\end{equation}
\end{definition}

This representation suggests that the state is generated by sequential maps on a quantum memory system. Note that in favour of a simpler presentation we are choosing a (suboptimal) memory system that does not depend on $j$. If we allowed the memory system to change with $j$ we could absorb $\rho_0$ and $\1_B$ in the definitions of respectively $\mathcal{E}^{(1)}$ and $\mathcal{E}^{(N)}$ as in \cref{reconstructionnontrans}.

We denote 
\begin{equation}
\rho_{s}:=\rho_0\mathcal E_{\1_\A}^{(1)}\cdots\mathcal E_{\1_\A}^{(s)}.
\end{equation}
We also define the maps  $\mathcal{F}_j:\mathcal{B}\rightarrow \mathcal{A}_{[\max(j-l+1,1),j]}^*$ as 
\begin{equation}
\mathcal{F}_j[Z](X_{j-l+1}\otimes \cdots\otimes X_j)=\rho_{j-l}\mathcal{E}_{X_{j-l+1}}^{j-l+1}\cdots\mathcal{E}_{X_j}^{j}(Z),
\end{equation}

 while $\mathcal{F}_0:\mathcal{B}\rightarrow \mathbb{C}$ is simply $\rho_0$, and set
\begin{align}
\W_{j}&=\{\mathcal{E}^{[j+1,\min(j+r,N)]}(A\otimes\1_{\B}),A \in \mathcal{A}_{[j+1,\min(j+r,N)]} \},\\
\tilde{\W}_j&=\{\rho_{\max(j-l+1,1)}\mathcal{E}^{[\max(j-l+1,1),j]}(A\otimes \cdot) , A \in \mathcal{A}_{[\max(j-l+1,1),j]} \}.
\end{align}

As in the translation invariant case, we can give $\V_j=\W_j/(\W_j\cap \tilde\W^{\perp}_j)$, $1\leq j\leq N-1$, the structure of a quotient operator system, with projections $L_j:\W_j\rightarrow \V_j$, such that the units are $L_j\mathcal{E}^{[j+1,\min(j+r,N)]}(\1_{\A_{[j+1,\min(j+r,N)]}}\otimes\1_{\B})=L_j\1_{\B}$, the norms are denoted as $\norm{\cdot}_{L_j\1_{\B}}$ and  $L_j$ are completely positive and contractive. Via invertible maps denoted as $R_j:\V_j\rightarrow \V'_j$ we can then identify $\V_j$ with $\V'_j\subseteq \W_j$ being the subspaces of $\W_j$ orthogonal to $\W_j\cap \tilde{\W}_j^{\perp}$, while $\V_0$, $\V'_0$, $\V_N$, $\V'_N$ are just $\mathbb{C}$ with its trivial operator system and $R_0$ and $R_N$ are identity maps. We have that the maps $\hat{U}_j^{\intercal} \F_j R_j$, for $j=1,..,N-1$, are invertible (see Lemma~\ref{lemmainvertj}), and 
\begin{proposition}\label{proponontransquotorg}
For an observable realization as in Definition \ref{reconstructionnontrans}, and a realization as in Definiton \ref{compposnontrans}, we have
\begin{align}
{\mathbb{K}}^{(1)}_A{U}_1^{\intercal} \F_1|_{\W_1}&=\F_{0} \cE_{A}^{(1)}|_{\W_1},\\
{\mathbb{K}}^{(N)}_A&={U}_{N-1}^{\intercal}\F_{N-1} \cE_{A}^{(N)}(\1_B),\\
{\mathbb{K}}^{(j)}_A{U}_j^{\intercal} \F_j|_{\W_j}&={U}_{j-1}^{\intercal}\F_{j-1} \cE_{A}^{(j)}|_{\W_j},\mathrm{for}\,2\leq j\leq N-1,
\end{align}
and
\begin{align}\label{eqKerelationnontrans}
{\mathbb{K}}^{[j,N]}_{X_1\otimes\cdots\otimes X_N}&={U}_{j-1}^{\intercal}\F_{j-1}\mathcal{E}^{[j,N]}_{X_1\otimes\cdots\otimes X_{N}}(\1_{\B}), \mathrm{for}\,2\leq j\leq N,\\
{\mathbb{K}}^{[1,N]}_{X_1\otimes\cdots\otimes X_N}&=\rho_0\mathcal{E}^{[1,N]}_{X_1\otimes\cdots\otimes X_{N}}(\1_{\B}).
\end{align}
\end{proposition}
Let us define the linear maps $\mathcal{E}^{[i,j]}:\A_{{i,j}}\otimes \B\rightarrow \B$ which act on a product vector $X_{j}\otimes\cdots\otimes X_{k}\otimes Z$ as
\begin{equation}
\mathcal{E}^{[i,j]}(X_i\otimes \cdots\otimes X_j)=\mathcal{E}_{X_j}^{(j)}\cdots\mathcal{E}_{X_{k}}^{(k)}(Z).
\end{equation}
For $C^*$-finitely correlated state we can write
\begin{align}
\omega(X_1\otimes \cdots\otimes X_N)&=\rho_0{\mathcal{E}_{X_1}^{(1)}}\cdots{\mathcal{E}_{X_{N}}^{(N)}}(\1_{\B})=\rho_0\mathcal{E}^{[1,N]}(X_{1}\otimes\cdots\otimes X_{N}\otimes \1_{\B}),
\end{align}
and defining
$
\rho_{i-1}:=\rho_0{\mathcal{E}_{\1}^{(1)}}\cdots{\mathcal{E}_{\1}^{(i-1)}}
$
we see
\begin{align}
\Omega^{[i,j,k]}[X_{i}\otimes\cdots\otimes X_j](X_{j+1}\otimes\cdots\otimes X_{k})=(\mathcal{E}_{X_{i}}^{(i)}\cdots\mathcal{E}_{X_{j}}^{(j)})^{\dagger}(\rho_{i-1})(\mathcal{E}_{X_{j+1}}^{(j+1)}\cdots{\mathcal{E}_{X_{k}}^{(k)}}(\1_{\B}))
\end{align}

First we prove the following:

\begin{lemma}\label{lemmainvertj}
$U^{\intercal}_{j}\F_j$ is invertible on $\V'_j$.
\end{lemma}
\begin{proof}
We have that for $A\in\mathcal{A}_{[j+1,\min(j+r,N)]}$ such that $0\neq\mathcal{E}^{[j,\min(j+r,N)]}(A\otimes \1_{\mathcal B})\in \V'_j$, there exists $A'\in \mathcal{A}_{[\max(j-l+1,1),j]}$ such that
\begin{align}
\rho_{\max(j-l+1,1)} \mathcal{E}^{[\max(j-l+1,1),j]}(A'\otimes \mathcal{E}^{[j+1,j+r]}(A\otimes \1_{\mathcal B})) &=\Omega^{[j-l+1,j,\min(j+r,N)]}[A](A')\\&=U_j D_j O_j[A](A')\neq 0,
\end{align}
implying $O_j[A]\neq0$. Thus we have
\begin{align}
U_j^{\intercal}\F_j(\mathcal{E}^{[j+1,\min(j+r,N)]}(A\otimes \1_{\mathcal B}))&=U_j^{\intercal}\Omega^{[j-l+1,j,j+r]}[A]=D_jO_j[A]\neq 0.
\end{align}
\end{proof}
Then we can prove Proposition~\ref{proponontransquotorg}.
\begin{proof}[Proof of Proposition~\ref{proponontransquotorg}]
We have, for any $X_{j+1}\otimes\cdots\otimes X_{\min(j+r,N)} \in \mathcal{A}_{[j+1,\min(j+r,N)]}$, $X_{\max(j-l+1,1)}\otimes\cdots\otimes X_{j} \in \mathcal{A}_{[\max(j-l+1,1),j]}$,
\begin{align}
 &\mathcal{F}_j [\mathcal{E}_{X_{j+1}}^{(j+1)}\cdots{\mathcal{E}_{X_{\min(j+r,N)}}^{(\min(j+r,N))}}(\1)](X_{\max(j-l+1,1)}\otimes \cdots\otimes X_j)\nonumber\\&=\rho_{j-l} \mathcal{E}_{X_{\max(j-l+1,1)}}^{\max(j-l+1,1)}\cdots\mathcal{E}_{X_j}^{(j)}\mathcal{E}_{X_{(j+1)}}^{(j+1)}\cdots{\mathcal{E}_{X_{\min(j+r,N)}}^{(\min(j+r,N))}}(\1_{\B}),
\end{align}
in particular, for $X\in \mathcal{A}_{[j+1,\min(j+r,N)]}$, $Y \in \mathcal{A}_{[\max(j-l+1,1),j]}$
\begin{equation}
\mathcal{F}_j[\mathcal{E}^{[j+1,\min(j+r,N)]}(X\otimes \1_{\B})](Y)=\Omega^{[j-l+1,j,j+r]}[X](Y),
\end{equation}
and for $X\in \mathcal{A}_{[j+1,\min(j+r,N)]}$, $X_j\in \mathcal{A}_{[j,j]}$,$Y \in \mathcal{A}_{[\max(j-l,1),j-1]}$
\begin{equation}
\mathcal{F}_{j-1}[\mathcal{E}^{[j,\min(j+r,N)]}(X_j\otimes X \otimes \1_{\B})](Y)=\Omega^{[j-l,j-1,j+r]}_{X_j}[X](Y).
\end{equation}

Writing $(U^{\intercal}_j\F_j)^{-1}$ for the inverse of $U^{\intercal}_j\F_j$ on $\V'_j\subseteq\W_j$, for $2\leq j\leq N-1$ we have

\begin{align}
\mathbb{K}^{(j)}_A(x)&:=U_{j-1}^{\intercal}\F_{j-1}\cE^{[j,\min(j+r,N)]}(A \otimes(U^{\intercal}_j\Omega^{[j-l+1,j,j+r]})^{+}(x)\otimes \1_{\B})\\
&=U_{j-1}^{\intercal}\F_{j-1}\cE^{(j)}_A(\cE^{[j+1,\min(j+r,N)]}((U^{\intercal}_j\Omega^{[j-l+1,j,j+r]})^{+}(x)\otimes \1_{\B}) )\\
&=U_{j-1}^{\intercal}\F_{j-1}\cE^{(j)}_A( (U^{\intercal}_j\F_j)^{-1}U^{\intercal}_j\F_j\cE^{[j+1,\min(j+r,N)]}((U_j^{\intercal}\Omega^{[j-l+1,j,j+r]})^{+}(x)\otimes \1_{\B}) ) \\
&=U_{j-1}^{\intercal}\F_{j-1}\cE^{(j)}_A(U^{\intercal}_j\F_j)^{-1}(x),
\end{align}

and similarly for $j=1$ and $j=N$. The rest follows by composition.

\end{proof}


As an immediate consequence, we have that the maps $M_{j-1}^{-1}{\mathbb{K}}^{(j)}_{(\cdot)}M_j$ with $M_j={U}_j^{\intercal} \F_j R_j$ for $2\leq j\leq N-1$, $M_0(z)=z$, $M_N(z)=z$, are completely positive and unital with respect to the orders on $\V'_j$ and $\V'_{j-1}$, and the error propagation bound from Theorem~\ref{theoerrprop22} applies. Analogously to the translation invariant case, we have that

\begin{align}
\Delta'&\leq\max_{j\in [N]}\frac{8 d_{\B}\sqrt{d_{\A}}}{\sqrt{3}\sigma_{m}\left(\Omega^{[j-l,j-1,j-1+r]}\right)}\nonumber\\&\times\left( \frac{\|\Omega^{[j-l+1,j,j+r]}  - \hat{\Omega}^{[j-l+1,j,j+r]} \|_{2}}{\sigma_m (\Omega^{[j-l+1,j,j+r]})^2} + \frac{\| \Omega_{(\cdot)}^{[j-l,j-1,j+r]}-\hat{\Omega}^{[j-l,j-1,j+r]} _{(\cdot)}\|_{2}}{3\sigma_m (\Omega^{[j-l+1,j,j+r]} )}\right).\label{bounderrorpardef24D2}
\end{align}

By using the error progapagation bound in Theorem~\ref{theoerrprop22}, we obtain

\begin{theorem}\label{theonontransC}
Let $\omega$ be a $C^*$, $(l,r,m)$-finitely correlated state, with associated maps $\Omega^{[i,j,k]}$ satisfying $\eta\leq \min_{j=0,\dots,N}\sigma_{m}(\Omega^{[j-l+1,j,j+r]})/2$. Assume 
we have estimates of marginals  $\omega_{[\max(j-l,1),\min(j+r,N)]}$, $\omega_{[\max(j-l+1,1),\min(j+r,N)]}$, $j=0,...,N$, such that the Hilbert-Schmidt norm errors are smaller than $ \frac{\epsilon\eta^3 }{20 td_{\B}  \sqrt{d_{\A}}}$. 
The estimated realization parameters from $\textsf{LearnFCS(l,r)}$ (given by~\ref{defstatereconnonhom}) give rise to an estimated operator $\hat\omega$ such that
\begin{equation}
\frac{\|\hat{\omega}-\omega\|_1}{2}\leq \epsilon,
\end{equation}

\end{theorem}

\section{Learning states $\epsilon$-close to FCS}\label{sec.learningclose}


A feature of our learning algorithm is that it works even if the true state is sufficiently close to a finitely correlated state. The following proposition holds.

\begin{proposition}\label{robustnessalgo}
Suppose that a state $\sigma$ is close to an $(l,r,m)$-finitely correlated $\omega$, say $d_{\mathrm{Tr}}(\sigma_t,\omega_t)\leq \xi(t)$, for any marginals of size $t$, and that estimates of the marginals in Hilbert-Schmidt distance at precision $\epsilon_{HS}$ are given to the reconstruction algorithm of Theorem~\ref{theonontrans}.
Then, if $\xi(l+r)\leq \sigma_{m_j}(\Omega^{[j-l+1,j,j+r]})/8$ for $j=0,\dots,N$, the estimate $\hat\sigma$ obtained from the reconstruction algorithm satisfies $\frac{\|\sigma-\hat{\sigma}\|_1}{2}\leq \xi(N)+\frac{\|\omega-\hat{\sigma}\|_1}{2}$,
with $\frac{\|\omega-\hat{\sigma}\|_1}{2}<D\epsilon$, where $D>0$ is a universal constant and $\epsilon_{HS}+\xi(l+r+1)= \min_{j=0,\dots,N} \frac{\epsilon\sigma^3_{m_j}(\Omega^{[j-l+1,j,j+r]})}{20 N m \sqrt{d_{\A}}}$.
\end{proposition}

\begin{proof}
Observe that $\|\hat{\Omega}^{[i,j,k]}-{\Omega}^{[i,j,k]}\|_2 \leq \epsilon_{HS}+\xi(k-i+1)$ for the relevant marginals. Provided that $\xi(l+r)\leq \sigma_{m_j}(\Omega^{[j-l+1,j,j+r]})/8$, then the truncated singular value decompositions return maps $\hat{U}_j$ of rank $m_j$ and $\|\hat{\Omega}^{[i,j,k]}-{\Omega}^{[i,j,k]}\|_2\leq \sigma_{m_j}({\Omega}^{[i,j,k]})/3$ for the relevant marginals. Therefore the error propagation bound can be invoked. 
\end{proof}

Many physically relevant many-body quantum states can be efficiently approximated by finitely correlated states. One such class corresponds to states that are efficiently approximated by quantum circuits: for an $n=2t$ qudit system initiated in the tensor product state $|0\rangle^{\otimes n}$, consider the following possibly non-unitary depth $D$ quantum circuit with brickwork architecture (see left part of Figure \ref{fig:circuits} for an illustration in the case of $D=2$):
\begin{align*}
\mathcal{C}:= \prod_{\ell=1}^D\, \Phi^{(\ell)}\,,
\end{align*}
where for each layer $\ell$, the quantum channel $\Phi^{(\ell)}$ factorizes as 
\begin{align*}
\Phi^{(\ell)}:= \bigotimes_{j=1}^{t}\Phi^{(\ell)}_{j}\,,
\end{align*}
with each $\Phi_j^{(l)}$ a completely positive, trace preserving map acting on qudits $j$ and $j+1$ if $\ell$ is odd, and on qudits $j-1$ and $j$ if $\ell$ is even. By commuting through maps acting on different qudits, the circuit $\mathcal{C}$ can equivalently be written as a concatenation of maps $\Psi_j$, $j\in[t]$ (see right part of Figure \ref{fig:circuits}):
\begin{align*}
\mathcal{C}:=\prod_{j=1}^t\Psi_j\,.
\end{align*}
With this alternative decomposition, it becomes clear that the state $\mathcal{C}(|0\rangle\langle 0|^{\otimes t})$ is a $C^*$-finitely correlated state, with realization: $(\mathcal{B},\1_\B,\mathcal{E}^{(1)},\dots,\mathcal{E}^{(t)},|0\rangle\langle 0|)$, where $\B:=\mathbb{M}_d^{\otimes {D-1}}$, $\mathcal{A}:=\mathbb{M}_d^2$ and, for any $j\in[t]$, $A\in\mathcal{A}$ and $B\in \mathcal{B}$,
\begin{align*}
\mathcal{E}^{(j)}(A\otimes B):=\langle 0^2|\Psi^*_j(A\otimes B)|0^2\rangle\in\B\,,
\end{align*}
where the evaluation in the tensor product state $|0^2\rangle=|0\rangle\otimes |0\rangle$ is with respect to the qudit systems at the output of the first channel included in the concatenation $\Phi_j$ acting on the input state $|0\rangle^{\otimes n}$. With these notations, it is easy to see that 
\begin{align*}
\tr{\mathcal{C}(|0^n\rangle\langle 0^n|) A_1\otimes \dots\otimes A_t}=\langle 0^2|\mathcal{E}_{A_t}\dots \mathcal{E}_{A_1}(\1_\B)|0^2\rangle\,.
\end{align*}
\begin{figure}[h]
    \centering  \includegraphics[scale = 0.15]{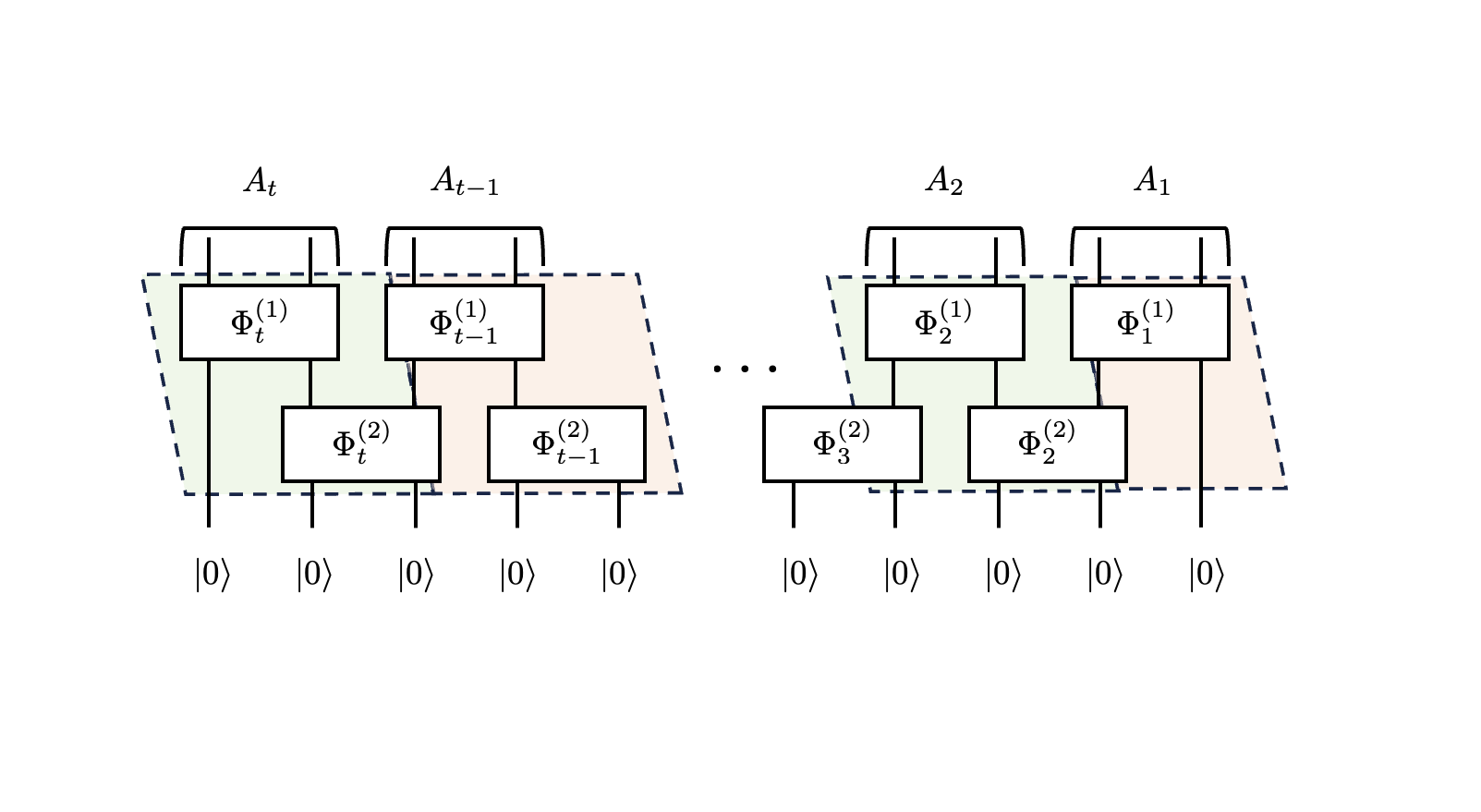} \includegraphics[scale = 0.15]{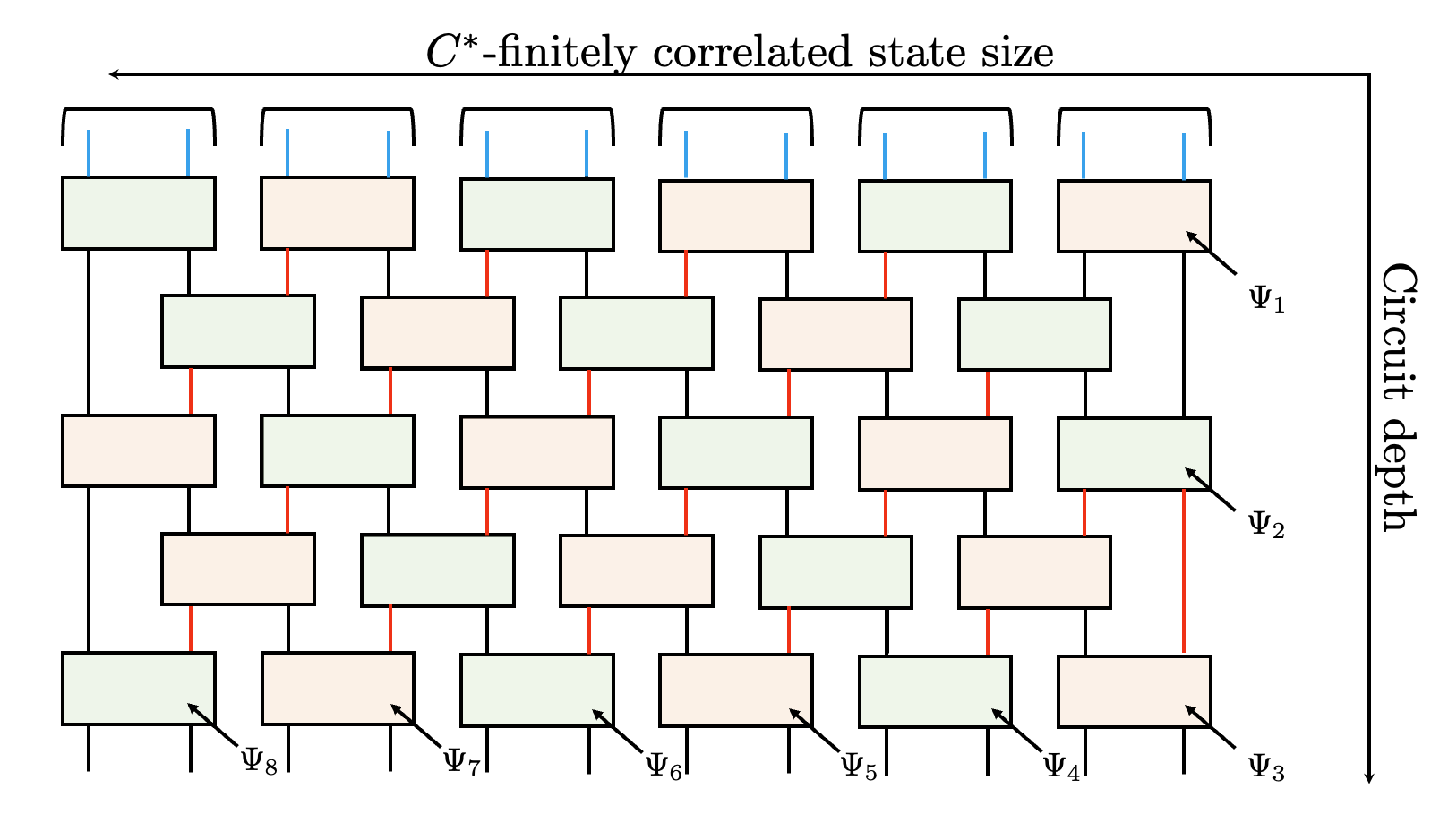}
    \caption{$C^*$-finitely correlated description of quantum circuits. On the right figure, the legs in red correspond to the system $\B$ at the entry of a map $\mathcal{E}^{(j)}$, whereas the blue legs correspond to the system $\mathcal{A}$.}
    \label{fig:circuits}
\end{figure}

An important class of many-body quantum states which are known to be approximated by outputs of (non-unitary) circuits as described above, namely 
Gibbs states of one-dimensional local Hamiltonians. More in general, approximation of Gibbs states in terms of matrix product density operators have been studied \cite{CiracMPDOGibbs,hastings2006solving,Alhambra2021, gondolf2024conditional}.
 The result of~\cref{robustnessalgo} suggests a path for learning one-dimensional Gibbs states in terms of their approximations by matrix product operators.
The recent years have seen the development of many learning algorithms for Gibbs states ~\cite{anshu2020sample,anshu2021sample,anshuweb,haah2022optimal,bakshi2023learning, de2021quantum, rouze2021learning,onorati2023efficient}. 
 These approximations involve blocking sites of the original chain, so that the local dimension of the new chain is $d_{\mathcal{A}}=d^L$ for some $L$ equal to the size of each block of sites and $d$ the original site dimension. To give learning guarantees via these results together with~\cref{robustnessalgo} we would need to
\begin{itemize}
\item (i) The smallest singular values $\sigma_{m_j}({\Omega}^{[j-l+1,j,j+r]})$ of the MPO approximations are large compared with the approximation error. At present we do not have any
https://optimization-online.org/wp-content/uploads/2018/01/6441.pdway to prove bounds for these quantities for MPO approximations of Gibbs states or for outputs of quantum circuits. We only have the upper bound $\sigma^2_{m_j}({\Omega}^{[j-l+1,j,j+r]})\leq \Tr[({\Omega}^{[j-l+1,j,j+r]})^{\dagger}{\Omega}^{[j-l+1,j,j+r]}]/m=\Tr[\omega_{[j-l+1,j+r]}^2]/m$.
\item (ii) To make the final trace norm bound in~\cref{robustnessalgo} non-trivial, the quantity 

$\min_{j=0,\dots,N} \frac{\sigma^3_{m_j}(\Omega^{[j-l+1,j,j+r]})}{20 N m \sqrt{d_{\A}}}$ must also be sufficiently large compared to the approximation error $\xi(l+r+1)$, otherwise $D\epsilon$ may be larger than $1$.  
\item (iii) check that the MPO approximations are $(l,r,m)$-finitely correlated states for sufficiently small $l,r$. This condition is not easy to verify even if the MPO is a circuit. However, it is reasonable to believe that it holds generically as long as $l$ and $r$ are $\Omega(\log m)$; therefore, if it does not hold from the start, it can be verified by an arbitrary small perturbation of the MPO approximation. A formal analysis of this fact and of how the perturbation affects the other conditions may be feasibile. Since the other conditions are anyway hard to verify, we avoid pursuing this.
\end{itemize}

 We mention that the positive MPO approximation in~\cite{gondolf2024conditional} has values of $m$ and $L$ which are compatible with conditions $(i)$ and $(ii)$ being satisfied, but since we do not know how large $l$ and $r$ should be, and we cannot lower bound $\sigma_{m_j}(\Omega^{[j-l+1,j,j+r]})$, we do not explore this further.

These observations suggests that to guarantee that spectral reconstruction algorithm work for Gibbs states, we need either to improve our analysis of MPO approximations or to modify the algorithm and its analysis. We leave these questions to future work.

In contrast, it was shown in \cite[Theorem 1.3]{devroye2020minimax} that learning classical Gibbs measures with $\epsilon$ average recovery in trace distance requires at least $n/\epsilon^2$ samples. Since classical Gibbs states are finitely correlated states, the latter directly implies the following:

\begin{proposition}
Let $A$ be a learning algorithm such that, upon measuring $N$ independent copies of an unknown finite-size, finitely correlated state $\rho$ on $\mathcal{A}_{[1,N]}$, outputs an estimated state $\hat{\rho}$ such that $\mathbb{E}\big[\|\rho-\hat{\rho}\|_1\big]\le \epsilon$, where the expected value refers to the inherent randomness of the measurement outcomes. Then necessarily $N=\Omega(N/\epsilon^2)$.
\end{proposition}

\section{Numerical experiments}\label{sec.numer}
\subsection{The model: ground state of the AKLT Hamiltonian}

In this section, we want to numerically test the performance of Algorithm~\ref{alg:learnFCS} on a concrete family of $C^*$-finitely correlated states that was given in \cite{Fannes1992} {[Section 2, Example 1]}. This family of states has the one site observable algebra $\mathcal{A} = \mathbb{M}_3$ and is parameterized by a single parameter $\theta\in [0,\pi)$. In order to give the explicit realization we need to define the following linear map 
\begin{align}
    V_\theta : \mathbb{C}^2\rightarrow\mathbb{C}^{3}\otimes \mathbb{C}^2,
\end{align}
that is completely defined by
\begin{align}
    V_\theta \ket{\tfrac{1}{2}} = \cos \theta \ket{1,-\tfrac{1}{2}} - \sin \theta \ket{0,\tfrac{1}{2} }, \\
    V_\theta \ket{-\tfrac{1}{2}} = \sin \theta \ket{0,-\tfrac{1}{2}} - \cos \theta \ket{-1,\tfrac{1}{2} },
\end{align}
where $\ket{\pm\frac{1}{2}}$ and $\ket{1},\ket{0},\ket{-1}$ denote orthonormal bases of $\mathbb{C}^2$ and $\mathbb{C}^{3}$ respectively.
Then the states $w_\theta$ that we consider are given by the realization $(\mathcal{V},e,\mathbb{E}_{\theta},\rho)$ described by,
\begin{itemize}\label{eq:aklt_quasirealization}
    \item $\mathcal{V} = \mathbb{M}_2$.
    \item $\mathbb{E}_{\theta,A}(B) = V^*_\theta ( A \otimes B )V_\theta $ for $A\in\mathbb{M}_3$ and $B\in\mathbb{M}_2$.
    \item $\rho (B ) = \frac{1}{2}\Tr (B)$ for $B\in\mathbb{M}_2$.
    \item $e = \1_{\mathbb M_2} \in \mathbb{M}_2$.
\end{itemize}
Thus, for any translation invariant state $w_\theta$ described by the above realization, we can recover any correlation function as 
\begin{align}
    \omega_\theta(A_1 \otimes  \cdots \otimes A_k ) = \frac{1}{2}\Tr \left(\mathbb{E}_{\theta,A_k} \cdots \mathbb{E}_{\theta,A_1}(\1) \right),
\end{align}
for any $A_1,\dots,A_k\in\mathbb{M}_2$. This family of states is interesting because for the particular value of $\cos \theta = \sqrt{\frac{2}{3}}$ the state $w_\theta$ coincides with the ground state of the AKLT Hamiltonian introduced in~\cite{affleck1988valence} and given by 
\begin{align}\label{eq:aklt_hamiltonian}
    H_{\text{aklt}} = \sum_i \frac{1}{2}\mathbf{S}_i \cdot \mathbf{S}_{i+1} + \frac{1}{6}\left( \mathbf{S}_i \cdot \mathbf{S}_{i+1} \right)^2 + \frac{1}{3},
\end{align}
where $\mathbf{S}_i = (S_x,S_y,S_z ) $ denote the spin 1 irreducible representation of $SU(2)$ acting at site $i$.

\subsection{Results}
For the simulations, we restrict our attention to the ground state of the AKLT Hamiltonian~\eqref{eq:aklt_hamiltonian} and fix $\cos \theta = \sqrt{\frac{2}{3}}$ in the realization of the previous Section~\ref{eq:aklt_quasirealization}. In order to fix a basis for the one site algebra $\mathbb{M}_3$ we use the normalized Gell-Mann matrices $\lbrace \lambda_i \rbrace_{i=0}^8$ \begin{center}
\begin{align}
    \lambda_0 = \frac{1}{\sqrt{3}}\begin{pmatrix}
        1 & 0 & 0 \\
        0 & 1 & 0 \\
        0 & 0 & 1
    \end{pmatrix} \quad
    \lambda_1  = \frac{1}{\sqrt{2}}\begin{pmatrix}
        0 & 1 & 0 \\
        1 & 0 & 0 \\
        0 & 0 & 0
    \end{pmatrix} \quad
    \lambda_2  = \frac{1}{\sqrt{2}}\begin{pmatrix}
        0 & -i & 0 \\
        i & 0 & 0 \\
        0 & 0 & 0
    \end{pmatrix} \nonumber \\
        \lambda_3 = \frac{1}{\sqrt{2}}\begin{pmatrix}
        1 & 0 & 0 \\
        0 & -1 & 0 \\
        0 & 0 & 0
    \end{pmatrix} \quad
    \lambda_4  = \frac{1}{\sqrt{2}}\begin{pmatrix}
        0 & 0 & 1 \\
        0 & 0 & 0 \\
        1 & 0 & 0
    \end{pmatrix} \quad
    \lambda_5  = \frac{1}{\sqrt{2}}\begin{pmatrix}
        0 & 0 & -i \\
        0 & 0 & 0 \\
        i & 0 & 0
    \end{pmatrix} \\
         \lambda_6 = \frac{1}{\sqrt{3}}\begin{pmatrix}
        0 & 0 & 0 \\
        0 & 0 & 1 \\
        0 & 1 & 0
    \end{pmatrix} \quad
    \lambda_7 = \frac{1}{\sqrt{2}}\begin{pmatrix}
        0 & 0 & 0 \\
        0 & 0 & -i\\
        0 & i & 0
    \end{pmatrix} \quad
    \lambda_8  = \frac{1}{\sqrt{6}}\begin{pmatrix}
        1 & 0 & 0 \\
        0 & 1 & 0 \\
        0 & 0 & -2
    \end{pmatrix}, \nonumber
\end{align}
\end{center}
where $\lambda_0$ is the normalized identity such that $\Tr (\lambda_i\lambda_j) = \delta_{i,j}$. We are interested in recovering the reduced density matrices of $k$ contiguous sites that are described by
\begin{align}\label{eq:reduced_density_basisexpansion}
    \rho^{(k)}_\theta = \sum_{i_1,\dots,i_k} \omega_\theta (\lambda_{i_1}\otimes \cdots \otimes \lambda_{i_k} ) \lambda_{i_1}\otimes \cdots \otimes \lambda_{i_k} ,
\end{align}
where $i_1,\dots,i_k \in\lbrace 0,1,\dots,8 \rbrace^{\otimes k}$. For the observable regular realization (Proposition~\ref{prop.observablequasireal}), we fix the right and left finite-dimensional unital subalgebras $\mathcal{C}_R = \mathcal{C}_L = \mathbb{M}_3$ for two contiguous sites with basis $\lbrace \lambda_i \rbrace_{i=0}^8$ for both $\mathcal{C}_R$ and $\mathcal{C}_L$. In order to test our algorithm, we will use some simulated data. Specifically, since we have access to the realization~\eqref{eq:aklt_quasirealization} we can compute the maps $\Omega$ and $\Omega_{A}$ for $A\in\lbrace \lambda_i\rbrace_{i=0}^8$ which are the observables that we want to read from the experiment. For completeness, we numerically check that $\text{rank}(\Omega ) = 4$ which is consistent with the realization given in Section~\ref{eq:aklt_quasirealization} since $\text{dim}(\mathcal{V}) = 4$. Then we fix the errors $\epsilon,\epsilon'>0$ and generate the simulated observables $\hat{\Omega}$ and $\hat{\Omega}_A$ as follows
\begin{align}
   \hat{\Omega} = \Omega + \epsilon \frac{P}{\| P \|_2} \quad
   \hat{\Omega}_{\lambda_{i}} = \hat{\Omega}_{\lambda_i} + \epsilon' \frac{P'_i}{\| P \|_2},
\end{align}
for all $i = 0,\dots,8$ and $P,P_i'\in\mathbb{R}^{9\times 9}$ are matrices where each entry is drawn from a Gaussian distribution with zero mean and unit variance. With this construction, we can fix the errors $\| \Omega - \hat{\Omega} \|_{2}$ and $\| \Omega_{(\cdot )} - \hat{\Omega}_{(\cdot )} \|_{2}$. Given the singular value decomposition of $\hat{\Omega}$ we need to achieve the invertibility condition of Theorem~\ref{theoerrprop} that says that $\hat{U}^{\intercal}U$ is invertible. In general $\hat{U}^{\intercal}\in\mathbb{R}^{9\times 9}$ since the perturbation $P$ will make $\hat{\Omega}$ full rank. In order to avoid this issue we can use as input rank $m=4$ and we truncate $\hat{U}$ such that we keep the columns corresponding to the 4 largest singular values of $\hat{\Omega}$. Another possibility is to use Lemma~\ref{lem:error_perturbation} in order to justify that the singular values with absolute values smaller than a threshold $\eta/2$ can be considered 0s. Thus, setting $\eta$ small enough following conditions in Definition~\ref{defrelclassc} we can truncate $\hat{U}$ keeping the columns corresponding to singular values greater than the threshold $\eta/2$. We do not simulate estimates of $\Omega(\1)$ and $\tau\Omega$, as we use the corresponding values of $\Omega$. This is fine for the purpose of this numerical illustration, as $\Omega$ acts on just two sites and thus the Hilbert-Schmidt error on $\widehat{\Omega(\1)}$ and $\widehat{\tau\Omega}$ computated from $\hat{\Omega}$ is of the same order of $\| \Omega - \hat{\Omega} \|_{2}$.

Using the simulated observables $\hat{\Omega},\hat{\Omega}_{\lambda_i}$ and the truncated $\hat{U}$ we use the spectral state decomposition of Definition~\ref{defstaterecon}
and compute the correlations
\begin{align}
    \hat{\omega}_{\theta}(\lambda_{i_1}\otimes \cdots \otimes \lambda_{i_k}) = \hat{\rho}\hat{\mathbb{K}}_{\lambda_{i_1}}\cdots\hat{\mathbb{K}}_{\lambda_{i_k}}\hat{e},
\end{align}
for different numbers of sites $k$. Then we reconstruct the full state $\hat{\rho}_k$ as in~\eqref{eq:reduced_density_basisexpansion} replacing $\omega_\theta$ by $\hat{\omega}_\theta$. Finally we compute the trace distance $\frac{1}{2} \| \hat{\rho}_k - \rho_k \|_1$ and plot our results in Figure~\ref{fig:results_simulation}.

\section{Matrix perturbation theory}\label{appendixpert}
We list the matrix perturbation theory results that we use, following the analysis of~\cite{Hsu2008,Siddiqi2009, balle2013learning}.

The following perturbation bounds can be found in \cite{stewart1990matrix}. Note that the original reference uses the convention of matrices with more rows than columns, while we need the opposite one in our analysis. This is not an issue as the norms and the Moore-Penrose pseudo-inverse are respectively invariant and commuting with transposition.

\begin{lemma}[Theorem 4.11, in~\cite{stewart1990matrix}, p.204]\label{lem:error_perturbation}
    Let $A,\tilde{A}\in\mathbb{R}^{m\times n}$ with $n\geq m$. Then if $\sigma_1 \geq \cdots \geq \sigma_n$ and $\tilde{\sigma}_1 \geq \cdots \geq \tilde{\sigma}_n$ are the singular values of $A$ and $A'$ respectively then 
    \begin{align}
        | \sigma_i - \tilde{\sigma}_i|\leq \|A - \tilde{A} \|_{2\rightarrow 2} \quad \text{for} \quad i = 1,\dots,n.
    \end{align}
\end{lemma}

\begin{lemma}[Theorem 3.8 in~\cite{stewart1990matrix}, p.143]\label{lem:error_pseudoinverses}
    Let $A,\tilde{A}\in\mathbb{R}^{m\times n}$, then the error for the pseudo-inverses has the following bound
    \begin{align}
        \|  \tilde{A}^+ - A^+\|_{2\rightarrow 2} \leq \frac{1+\sqrt{5}}{2}\max \lbrace \|\tilde{A}^+ \|^2_{2\rightarrow 2},\| A^+ \|^2_{2\rightarrow 2}\rbrace \|  \tilde{A}- A\|_{2\rightarrow 2}.
    \end{align}
\end{lemma}

We also need the following results.

\begin{lemma}[Corollary 17 of~\cite{Siddiqi2009}]\label{cor:svd}
Let $A \in \mathbb{R}^{m \times n}$, with $m \geq n$, have rank $k\leq n$, and let $U \in
\mathbb{R}^{m \times k}$ be the matrix of $n$ left singular vectors corresponding
to the non-zero singular values $\sigma_1 \geq \ldots \geq \sigma_k > 0$ of
$A$. Let $\tilde A = A + E$. Let $\tilde U \in \mathbb{R}^{m \times k}$ be the matrix of
$k$ left singular vectors corresponding to the largest $k$ singular values
$\tilde \sigma_1 \geq \ldots \geq \tilde \sigma_n$ of $\tilde A$, and let $\tilde
U_\perp \in \mathbb{R}^{m \times (m-k)}$ be the remaining left singular vectors.
Assume $\|E\|_{2\rightarrow 2} \leq \epsilon \sigma_k$ for some $\epsilon < 1$. Then:
\begin{enumerate}
\item $\tilde \sigma_k \geq (1-\epsilon) \sigma_k$,
\item $\|\tilde U_\perp^\intercal U\|_{2\rightarrow 2} \leq \|E\|_{2\rightarrow 2}/\tilde \sigma_k$.
\end{enumerate}
\end{lemma}

\begin{lemma}\label{specialpert}[Special case of Corollary 2.4 of ~\cite{li1999lidskii}]
Let $A\in\mathbb{R}^{d_1\times d_2}$ of rank $m$, and $S\in\mathbb{R}^{d_1\times d_1}$, $T\in\mathbb{R}^{d_2\times d_2}$. Then, for any $1\leq i\leq m$
\begin{align}
\sigma_{d_1}(S)\sigma_{d_2}(T)\leq \frac{\sigma_i(SAT)}{\sigma_i(A)}\leq \sigma_{1}(S)\sigma_{1}(T)
\end{align}
\end{lemma}

\begin{lemma} \label{lemma:subspace}
In the notations of Section \ref{sec.statereconstruction}, suppose $\|\Omega-\hat{\Omega}\|_{2\rightarrow 2} \leq \eps \cdot \sigma_m(\Omega)$ for some
$\eps < 1/2$.
Then $\eps_0 := \|\Omega-\hat{\Omega}\|_{2\rightarrow 2}^2/((1-\eps)\sigma_m(\Omega))^2 <1$ and:
\begin{enumerate}
\item $\sigma_m(\hat U^\intercal \hat \Omega) \geq (1-\eps) \sigma_m( \Omega)$,
\item $\sigma_m(\hat U^{\intercal} U) \geq \sqrt{1-\eps_0}$, 
\item $\sigma_m(\hat U^\intercal  \Omega) \geq \sqrt{1-\eps_0} \sigma_m( \Omega)$.
\end{enumerate}
\end{lemma}
\begin{proof}
Since $\sigma_m(\hat U^\intercal \hat \Omega) = \sigma_m(\hat \Omega)$, the
(1.) is immediate from Lemma~\ref{cor:svd}.
(2.) follows since for any $x \in \mathbb{C}^m$, $\|\hat U^\intercal U x\|_{2} \ge \|x\|_2 \sqrt{1 - \|\hat
U_\perp^\intercal U\|_{2\rightarrow 2}^2} \geq \|x\|_2 \sqrt{1 - \eps_0}$ by
Lemma~\ref{cor:svd} and the fact that $\eps_0 < 1$ (see Lemma 21 of~\cite{Siddiqi2009} for a full proof). (3.) follows from Lemma~\ref{specialpert} noticing that 
$\sigma_{m}(\hat U^{\intercal}\Omega)=\sigma_{m}((\hat U^{\intercal} U) (U^{\intercal}\Omega))$, $\sigma_{m}(U^{\intercal} \Omega)=\sigma_{m}(\Omega)$,
and choosing $S=\hat U ^{\intercal}U$, and $A$ obtained from padding $D$ with zeros and $T$ a unitary which extends $O$.
\end{proof}

\end{appendix}

\end{document}